\documentclass{article}
\usepackage{amssymb,amsmath,amsthm,mathtools}
\usepackage[utf8]{inputenc}

\usepackage{hyperref}

\newtheorem{theorem}{Theorem}

\newtheorem{lemma}[theorem]{Lemma}
\newtheorem{cor}[theorem]{Corollary}

\newtheorem{prob}[theorem]{Problem}

\newtheorem{condition}[theorem]{Condition}
\newtheorem{question}[theorem]{Question}
\newtheorem{fact}[theorem]{Fact}

\providecommand{\keywords}[1]
{
  \small	
  \textbf{\textit{Keywords---}} #1
}

\begin{document}

\title{On Monotone Sequences of Directed Flips, Triangulations of Polyhedra, and Structural Properties of a Directed Flip Graph}

\author{Hang Si\\
        \small Weierstrass Institute (WIAS), Berlin
}

\maketitle

\begin{abstract}
This paper studied the geometric and combinatorial aspects of the classical Lawson's flip algorithm~\cite{Lawson1972,Lawson1977}. 
Let ${\bf A}$ be a finite point set in $\mathbb{R}^2$ and $\omega : {\bf A} \to \mathbb{R}$ be a height function which lifts the vertices of ${\bf A}$ into $\mathbb{R}^3$.  
Every flip in triangulations of ${\bf A}$ can be associated with a direction~\cite[Definition 6.1.1]{TriangBook}. 
We first established a relatively obvious relation between monotone sequences of directed flips on triangulations of ${\bf A}$ and triangulations of the lifted point set ${\bf A}^{\omega}$ in $\mathbb{R}^3$. 
We then studied the structural properties of a directed flip graph (a poset) on the set of all triangulations of ${\bf A}$.
We proved several properties of this poset which clearly explain when Lawson's algorithm works and why it may fail in general. 
We further characterised the triangulations which cause failure of Lawson's algorithm, and showed that they must contain redundant interior vertices which are not removable by directed flips. 
A special case of this result in 3d has been shown in~\cite{Joe1989}.
As an application, we described a simple algorithm to triangulate a special class of 3d non-convex polyhedra without using additional vertices.  We prove sufficient conditions for the termination of this algorithm, and show it runs in $O(n^3)$ time, where $n$ is the number of input vertices.
\end{abstract}

\keywords{weighted Delaunay triangulations, non-regular triangulations, Lawson's flip algorithm, directed flips, monotone sequence, flip graph,  Higher Stasheff-Tamari poset, redundant interior vertices, Sch\"onhardt polyhedron, Steiner points
}




\section{Introduction}

The motivation of this paper is to study and understand the geometric meaning of the Lawson's flipping algorithm~\cite{Lawson1972,Lawson1977}. At first, it is well-understood that this algorithm in $\mathbb{R}^d$ acts like to add $d+1$ simplices to the lower envelope of a $d+1$ polyhedron. This makes the lower envelope of this polyhedron towards convex.  However, it only guarantees the termination in the plane. 

This paper takes another view of Lawson's flip algorithm by focusing on the sequence of flips produced by this algorithm. It turns out that this sequence allows a nice geometric interpretation. It corresponds to a triangulation of the volume of a $d+1$ polyhedron. This relation has already been discovered for a long time.  Sleater et al~\cite[Lemma 5]{Sleator1988} proved that a sequence of edge-flips converting one triangulation into another of a convex $n$-gon corresponds to a tetrahedralisation of a 3d polyhedron with these two triangulations as its boundary, see Figure~\ref{fig:sequence-of-flips}. They used this fact to prove that the maximum flip distance between two triangulations of a convex $n$-gon is $2n-10$, and it is tight for sufficiently large $n$~\cite[Theorem 2]{Sleator1988}~\footnote{Due to the isomorphism between the flip graph of $n+2$ nodes and the binary tree rotation graph of $n$ nodes~\cite[Lemma 1]{Sleator1988}. It is equivalent to the maximum rotation distance between two $n$-node binary trees, which is $2n - 6$ for sufficiently large $n$.}. 

\begin{figure}[ht]
\centering
\includegraphics[width=.8\textwidth]{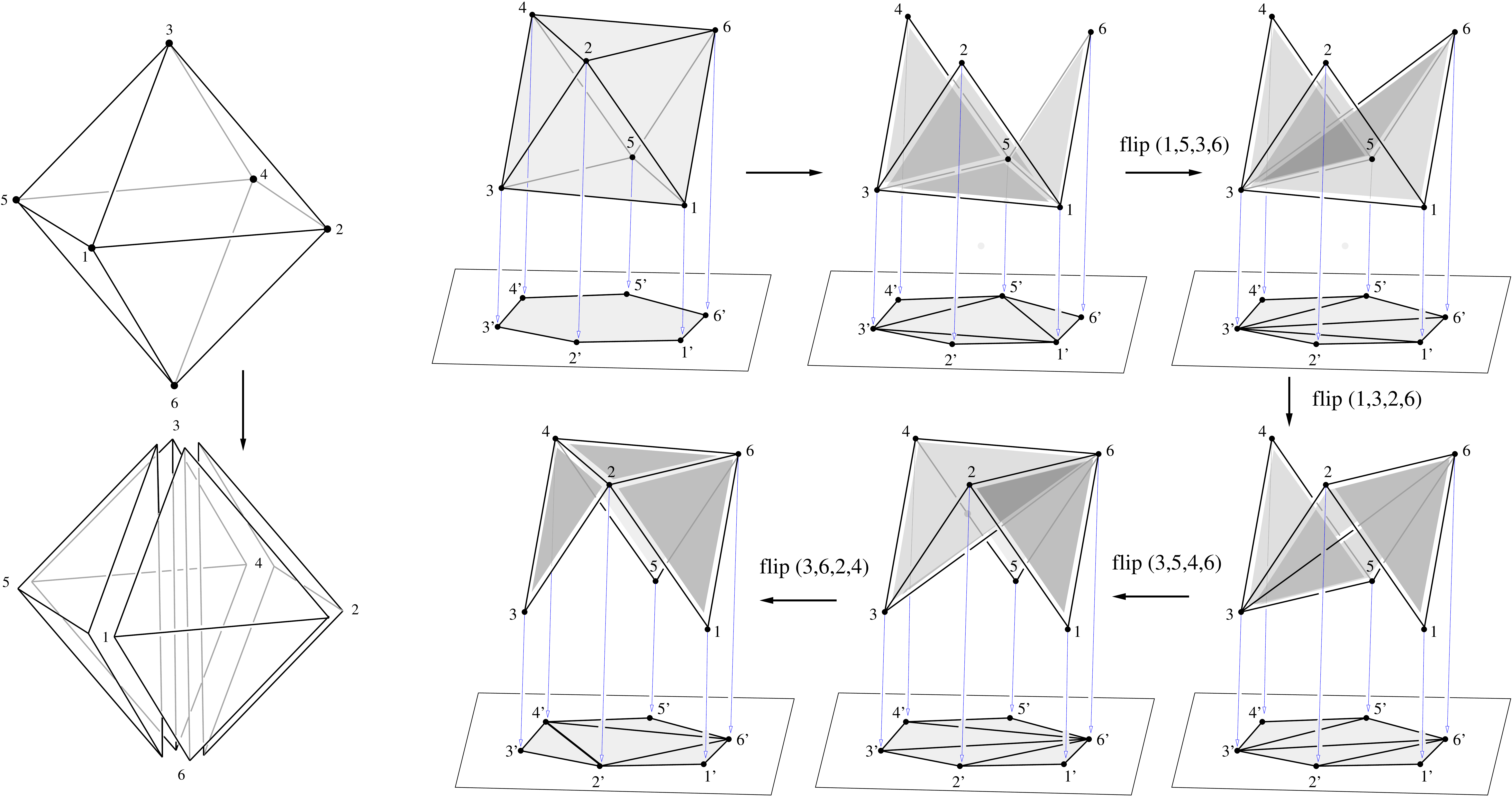}
\caption{Left: the regular octahedron is decomposed into $4$ tetrahedra.  Right: a sequence of edge flips corresponds to this tetrahedralisation.}
\label{fig:sequence-of-flips}
\end{figure}

All flips in the sequences produced by Lawson's algorithm have the same properties, such as \emph{(i)} every flip corresponds to a unique tetrahedron inside a 3d polyhedron, and \emph{(ii)} every flip in this sequence follows the same ``direction". This means that this sequence is monotone with respect to this direction.  

\subsection{Our contributions and outline}

The excellent book by De Lorea, Rambau, and Santos~\cite{TriangBook} gives a comprehensive exposition on the topics of this paper and laid down the fundamental mathematical concepts and definitions which are followed by this paper. 

Section~\ref{sec:def} gives the basic definitions together with an overview of the most related topics and previous works.  In particular, the related theorems are introduced. 
This section is self-contained. It may be skipped for the first reading. \\

The precise definitions of directed flips, monotone sequence of direct flips, and the directed flip graph are given in section~\ref{sec:graph-poset}. We briefly introduce them here for presenting our results. 
Let ${\bf A}$ be a finite point set in $\mathbb{R}^d$ and let $\omega: {\bf A} \to  \mathbb{R}$ be a height function which lifts every vertex in ${\bf A}$ into a lifted point in $\mathbb{R}^{d+1}$. Let ${\bf A}^{\omega}$ in $\mathbb{R}^{d+1}$ be the set of lifted points of ${\bf A}$. Then every flip in a triangulation of ${\bf A}$ corresponds to a $d+1$-simplex whose vertices are in ${\bf A}^{\omega}$.  
Let ${\cal T}_1$ and ${\cal T}_2$ be two triangulations of ${\bf A}$, and ${\cal T}_2$ is obtained from ${\cal T}_1$ by a flip. We call this flip {\it down-flip} if the upper faces of the $d+1$ simplex are in ${\cal T}_1$ and lower faces are in ${\cal T}_2$, otherwise, it is an {\it up-flip}.  Figure~\ref{fig:directed_flips} illustrates these two directed flips. 

\begin{figure}[ht]
  \centering
  \begin{tabular}{c c}
  \includegraphics[width=.45\textwidth]{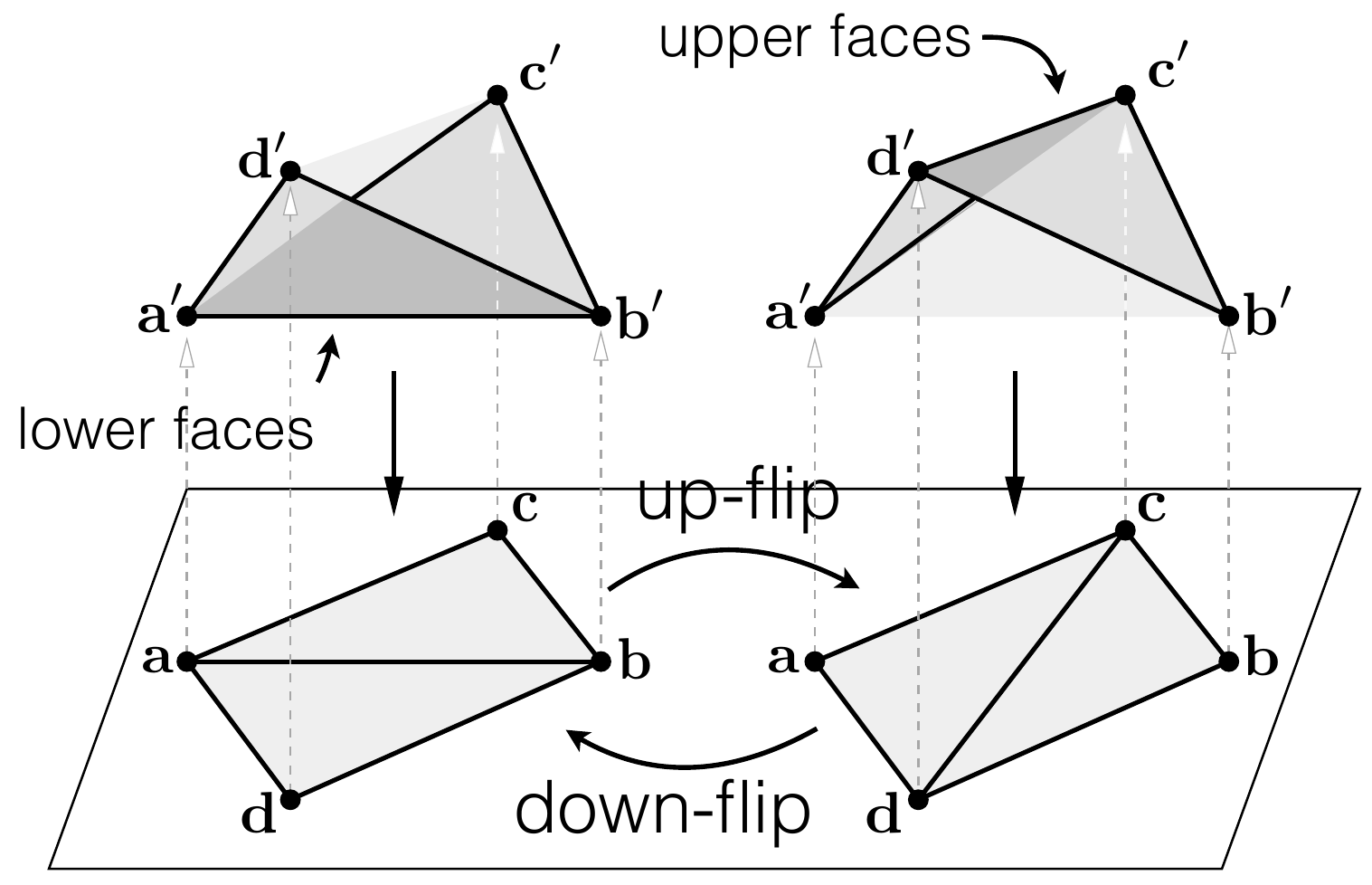} &
  \includegraphics[width=.45\textwidth]{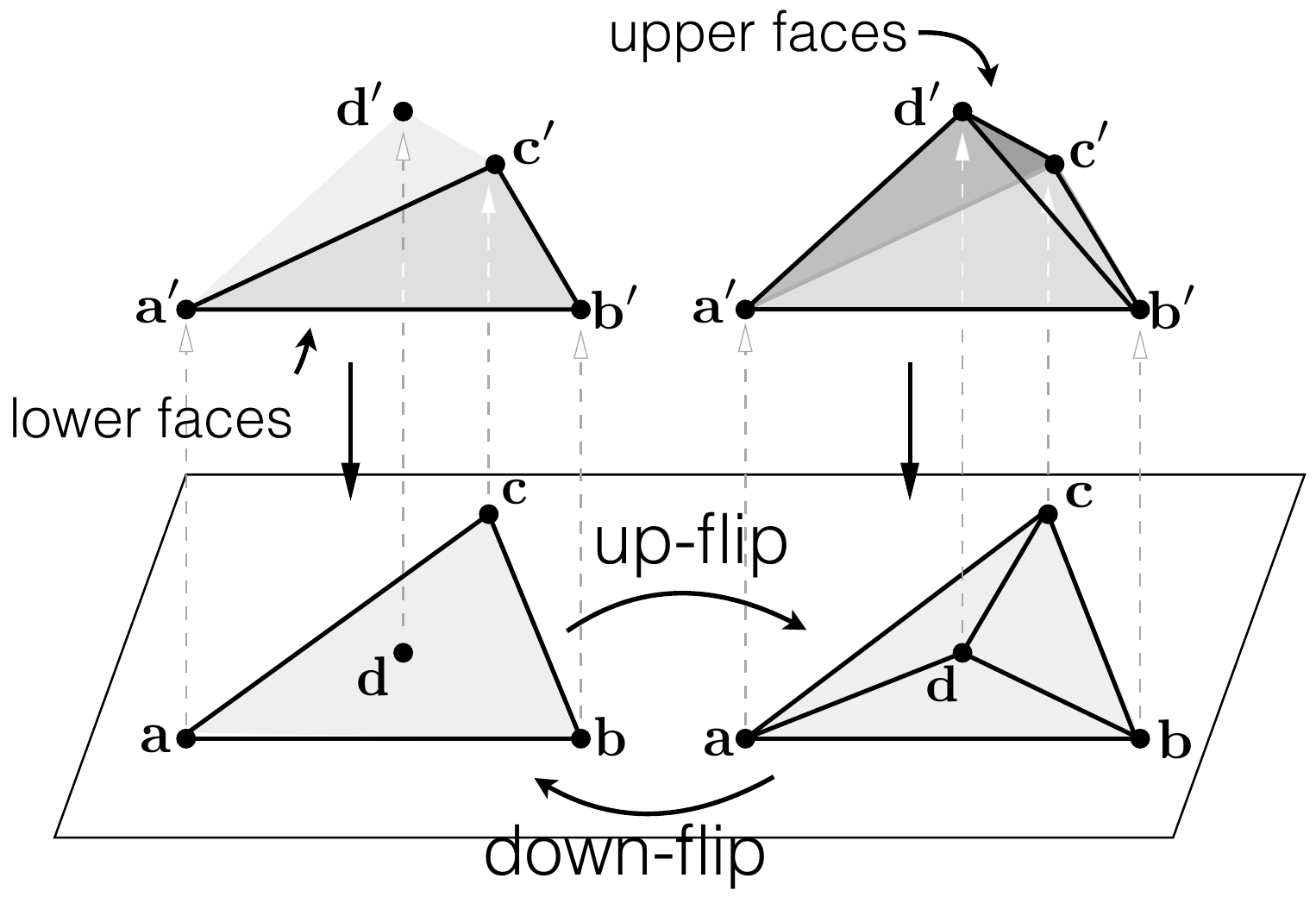} \\
   2-2 flips (edge flips) &  1-3 and 3-1 flips \\
   & vertex insertions and deletions
  \end{tabular}
  \caption{The two directed flips in the plane. In particular, an up-flip replaces the lower by the upper faces of the support of this flip, and a down-flip does the opposite way as an up-flip does.}
  \label{fig:directed_flips}
\end{figure}

By this definition, any flip has a direction. 
A sequence of directed flips is {\it monotone} if all flips in this sequence have the same direction.\\

Section~\ref{sec:monotone-to-triang} and Section~\ref{sec:triang-to-monotone} are devoted to the proof of a  relation between monotone sequences of directed flips in the triangulations of a point set in the plane and triangulations of a polyhedron in $\mathbb{R}^3$.  

\begin{itemize}
\item[(1)] Any monotone sequence of directed flips between two triangulations ${\cal T}_1$ and ${\cal T}_2$ of ${\bf A}$ (or vice versa) corresponds to a 3d triangulation of a polyhedron whose boundary consists of these two triangulations (Theorem~\ref{thm:monotone-to-triang} and Corollary~\ref{cor:triang}).

\item[(2)] A 3d triangulation ${\cal T}_{uv}$ of a polyhedron with vertices in ${\bf A}^{\omega}$ corresponds to a monotone sequence of directed flips between ${\cal T}_u$ and ${\cal T}_v$ (or vice versa) if it is acyclic with respect to a viewpoint in the direction of the directed flips (Theorem~\ref{thm:triang-to-monotone} and Corollary~\ref{thm:triang-to-monotone-viewpoint}).
\end{itemize}

This relation is relatively obvious and has been partially proven and remarked in many previous works which will be reviewed in Section~\ref{sec:def}. 
We comment that (1) is a straightforward generalisation of many special cases.
It is proven for the set of points of convex $n$-gons in~\cite[Lemma 5]{Sleator1988}~\footnote{In their proof, only the case of edge-flips is considered, but it could be generalised to show the cases which involve vertex insertion and deletion as well.}.
It is also proven for the set of triangulations of cyclic polytopes~\cite[Theorem 6.1.19]{TriangBook}. 
It is implicitly proven in~\cite[Theorem 5.3.7]{TriangBook} for the set of regular triangulations (see Theorem~\ref{thm:secondarypolytope}). 
However, we're not aware of a complete proof of this fact for an arbitrary point set. 
We prove this relation for an arbitrary finite point set in $\mathbb{R}^2$ which involves all types of elementary flips of $\mathbb{R}^2$, i.e., edge flips (2-2 flips), vertex insertions (1-3 flips), and vertex deletions (3-1 flips).
In particular, we show that non-regular triangulations may appear within the sequence.  
The proof of (1) is given in Section~\ref{sec:monotone-to-triang}. 
 
Fact (2) has been partially proven in the theory of Secondary polytopes, see e.g.~\cite[Theorem 5.3.7]{TriangBook} as well as the set of  triangulations of cyclic polytopes~\cite[Theorem 6.1.19]{TriangBook}. To the best of our knowledge, this fact has not been established for an arbitrary finite point set.  The proof of (2) is given in Section~\ref{sec:triang-to-monotone}. \\

\begin{figure}[ht]
  \centering
  \begin{minipage}{.25\textwidth}
  \includegraphics[width=0.9\textwidth]{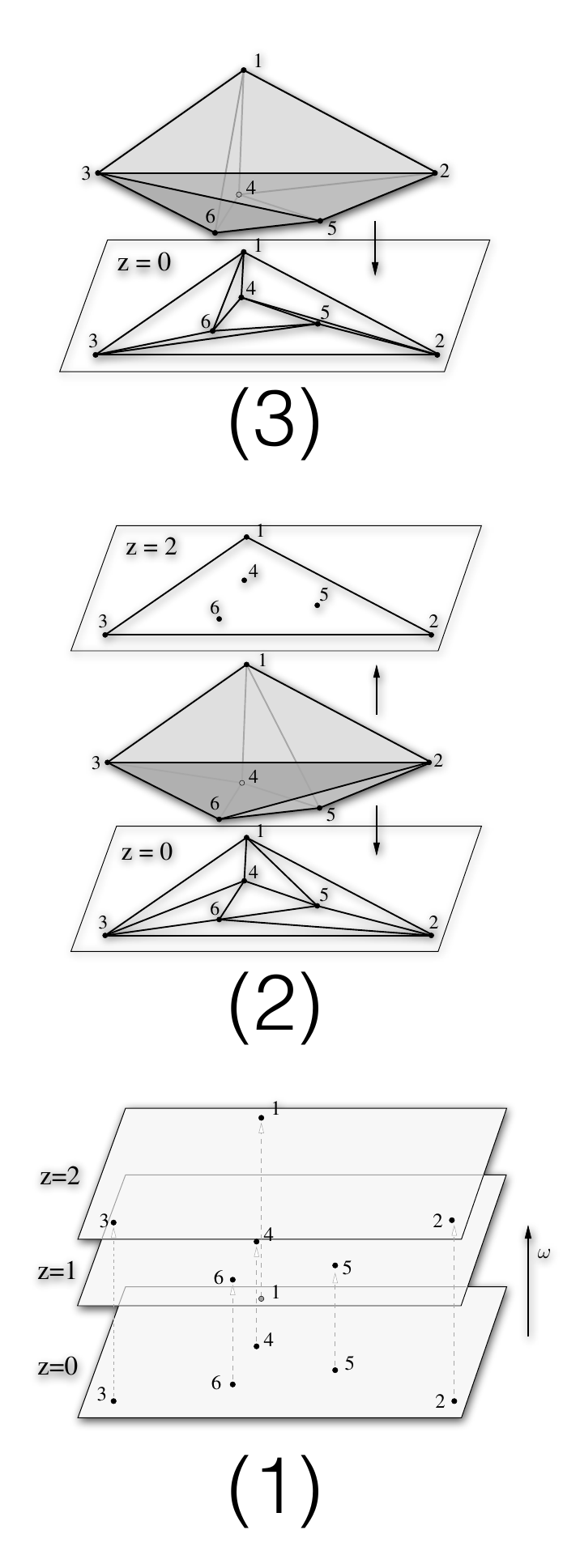}
  \end{minipage}
  \begin{minipage}{.7\textwidth}
  \includegraphics[width=1.0\textwidth]{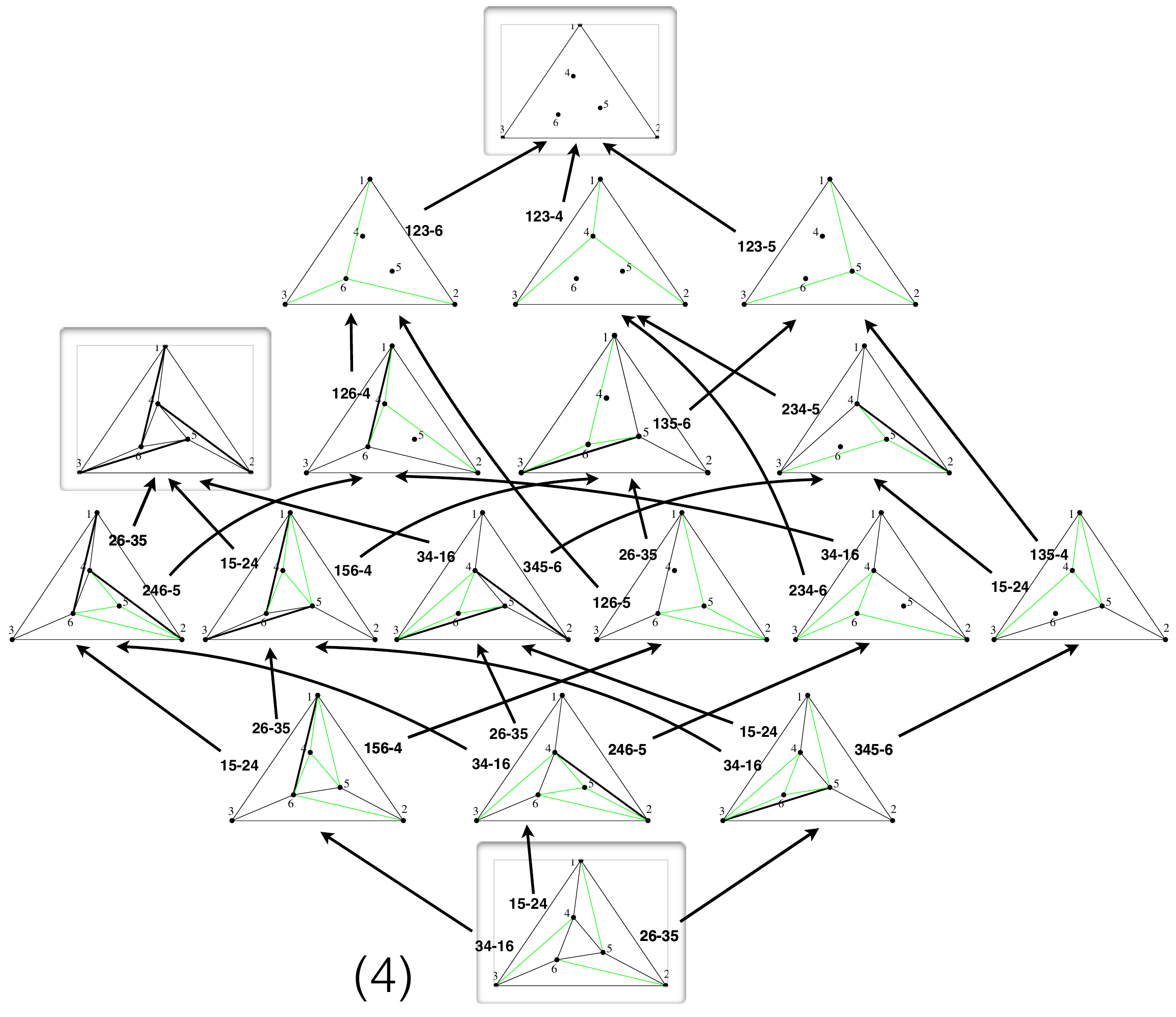}
  \end{minipage}
\caption{
The poset of monotone sequence of directed flips of a point set of $6$ points in the plane and a height function $\omega$.   
\emph{(1)} The lifted point set ${\bf A}^{\omega}$ is produced by the height function $\omega$ which sends the vertices $1,2,3$ to the plane $z = 2$, and the vertices $4,5,6$ to the plane $z = 1$. 
\emph{(2)} The regular and farthest point regular triangulations of $({\bf A}, \omega)$. 
\emph{(3)} A non-regular triangulation of ${\bf A}$. 
\emph{(4)} The green edges are locally non-regular and flippable edges.
}
\label{fig:prism_poset_1}
\end{figure}

From the directed flips, one can construct a directed flip graph (a poset) on the set of all triangulations of ${\bf A}$, exactly the same as the construction of the first Higher Stasheff-Tamari (${\bf HST}$)  poset from triangulations of point sets of cyclic polytopes~\cite{Edelman1996,Rambau96-thesis}. The structure properties of the ${\bf HST}$ has been well understood~\cite[Theorem 6.1.19]{TriangBook}. 
However, the structure of this poset for an arbitrary finite point set is more complicated and interesting.  Figure~\ref{fig:prism_poset_1} shows a motivation example of this poset which is discussed in detail in Section~\ref{sec:graph-poset-example}. 

The main result of this paper is to illustrate several general and special structural properties of this poset. Section~\ref{sec:poset-structure} is devoted to this purpose. 

In Section~\ref{sec:poset_general_properties}, a general structural theorem (Theorem~\ref{thm:poset-general}) of this poset which includes various obvious properties is proven. These properties are easily derived from the relations proven in (1) and (2). 
In particular, 
we showed that this poset is in general not bounded, i.e., it has neither unique minimum nor unique maximum.  On the other hand, the convex or concave property of the height function $\omega$ gives special property of this poset. 
In Section~\ref{sec:poset_convex_concave_omega}, we proved that such a poset will contains either unique minimum or unique maximum, but not both in general (Theorem~\ref{thm:poset-convex-heights}).  
The existence of a unique minimum or maximum guarantees that the Lawson's flip algorithm will terminate when its flips are coincident with this direction. 

\begin{figure}[ht]
  \centering
  \includegraphics[width=0.7\textwidth]{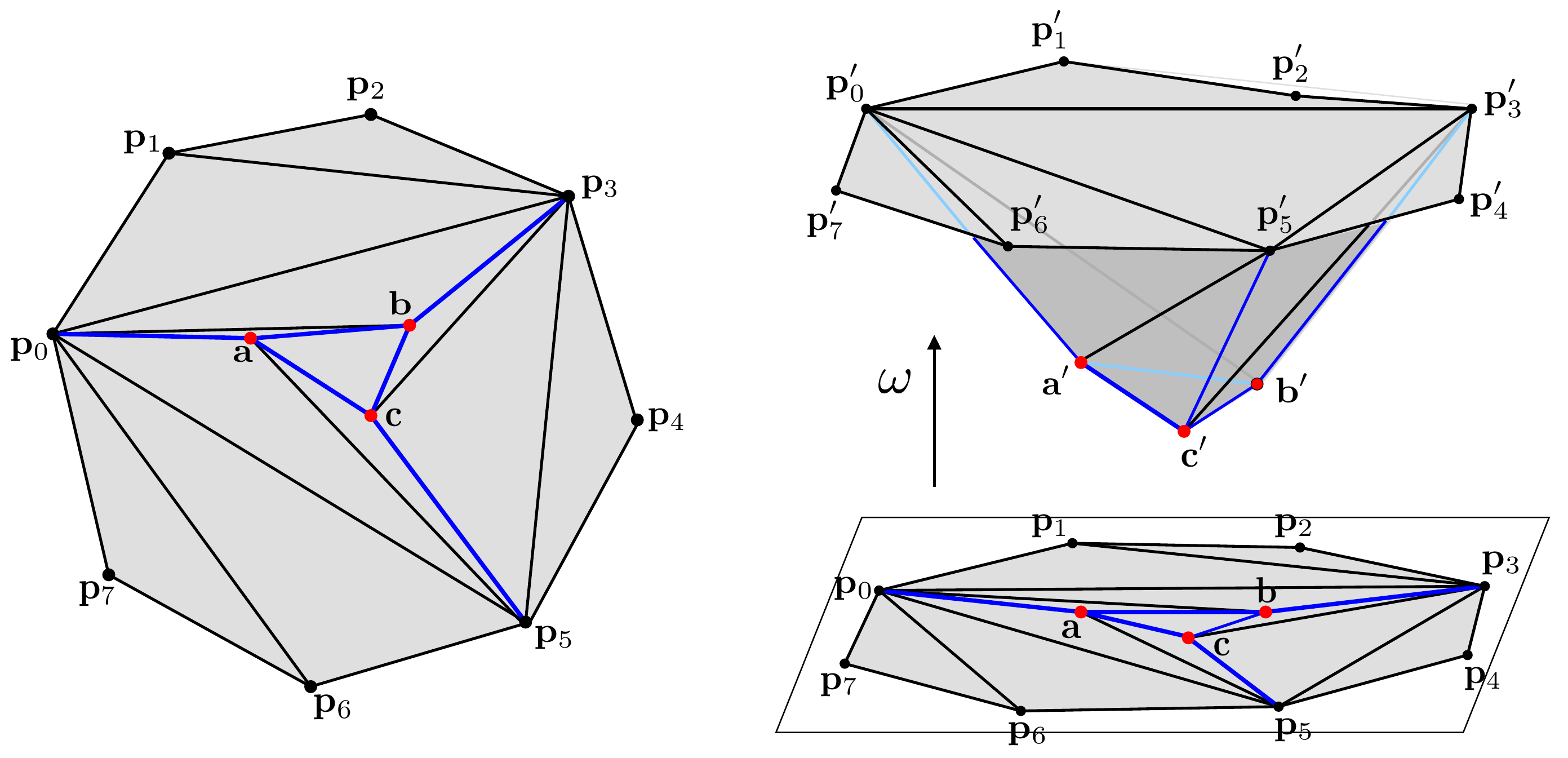}
\caption{Redundant interior vertices (definition is given in Section~\ref{sec:redundant_vertices}) are shown in red. Left: A triangulation corresponds to an upper non-extreme node in this poset. The red vertices are lower interior vertices.  The blue edges are locally non-regular with respect to up-flips. Moreover,  all blue edges are  unflippable. Right: The lifted triangulation are shown. The three edges: $e_{\bf ab}, e_{\bf bc}, e_{\bf ca}$ form a cycle of unflippable locally non-regular edges with respect to up-flips.}
\label{fig:redundant-vertices} 
\end{figure}

We further characterised those triangulations which cause the failure of Lawson's algorithm. 
Figure~\ref{fig:redundant-vertices} shows such as example.
Theorem~\ref{thm:cycles_edges} showed that such triangulations must contains 
 {\it redundant interior vertices} (defined in Section~\ref{sec:redundant_vertices}) which cannot be removed by the corresponding directed flips. 
We proved this by showing that these vertices are connected by cycles of connected unflippable locally non-regular edges with respect to the corresponding directed flips.  (The definition of unflippability is given in Section~\ref{subsec:flips}.) 
Interestingly, 
Joe~\cite{Joe1989} proved a special case in 3d triangulations, which showed that the failure of Lawson's algorithm in 3d is due to the existence of a cycle of connected unflippable locally non-Delaunay faces~\cite[Lemma 4, 5]{Joe1989}. This results implies that there must exist cycles of unflippable locally non-Delaunay edges in such 3d triangulations. Moreover, these edges are redundant interior edges.\\

In Section~\ref{sec:triang3d}, we showed an application of the obtained structural property theorem. We showed that a slightly modified Lawson's algorithm can be used to triangulate a special class of 3d non-convex polyhedra without using additional vertices.   We prove sufficient conditions for the termination of this algorithm, and show it runs in $O(n^3)$ time, where $n$ is the number of vertices of the polyhedron. \\

Finally, We present some open questions regarding the full structural properties of this poset and the set of all triangulations of ${\bf A}^{\omega}$ in Section~\ref{sec:discussion}.

\section{Definitions and Related Works}
\label{sec:def}

\subsection{Triangulations of a point set}

Let ${\bf A}$ be a finite point set in $\mathbb{R}^d$. A {\it triangulation} of ${\bf A}$~\cite[Def 2.2.1]{TriangBook} is a collection ${\cal T}$ of $i$-dimensional simplices, where $i = -1, 0, \ldots, d$, whose vertices are in ${\bf A}$, and it satisfies three conditions, which are:
\begin{itemize}
\item[(1)] All faces of simplices of ${\cal T}$ are in ${\cal T}$. (Closure Property)
\item[(2)] The intersection of any two simplices of ${\cal T}$ is a (possibly empty) face of both. (Intersection Property)
\item[(3)] The union of these simplices equals to the convex hull of ${\bf A}$, denoted as $\textrm{conv}({\bf A})$. (Union Property) 
\end{itemize}
In other words, a triangulation of ${\bf A}$ is a geometrically realised simplical complex whose underlying space is $\textrm{conv}({\bf A})$.  A triangulation of ${\bf A}$ needs not contain all vertices of ${\bf A}$. 

\begin{figure}[ht]
  \centering
  \includegraphics[width=0.9\textwidth]{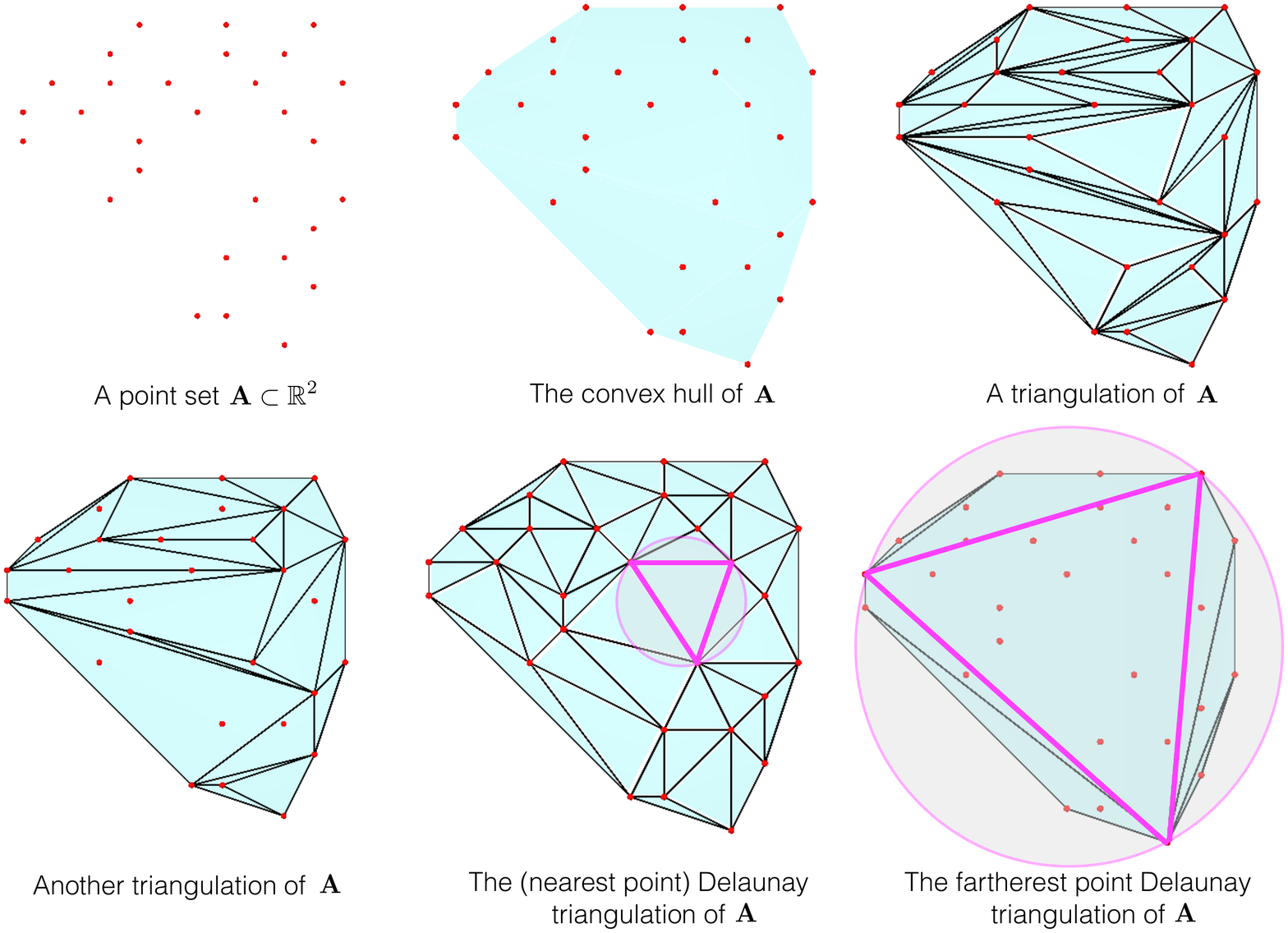} 
\caption{ Triangulations of a point set ${\bf A}$ in the plane. }
\label{fig:triang_points}
\end{figure}

Figure~\ref{fig:triang_points} shows several triangulations of a two dimensional point set.  
In particular, the famous {\it (nearest point) Delaunay triangulation}, is the triangulation of ${\bf A}$ such that no vertex of ${\bf A}$ is inside the circumcircle of any triangle of it~\cite{Delaunay1934}. Likewise, the {\it farthest point Delaunay triangulation}  is the triangulation of ${\bf A}$ such that no vertex of ${\bf A}$ is outside the circumcircle of every triangle of it, see e.g.~\cite{EPPSTEIN1992143}.

\subsection{Regular triangulations}

Given a point set, there are many triangulations of it. 
A (nice) family of triangulations of a finite point set in $\mathbb{R}^d$ are canonical projections of the envelopes of a convex $d+1$ dimensional polytope. 
They are called {\it regular triangulations}, also known as  {\it weighted Delaunay}~\cite{EdelsbrunnerShah96}, {\it Gale}, or {\it coherent} triangulations, see e.g.,~\cite{Ziegler1995-book, TriangBook}. 

\begin{figure}[ht]
  \centering
  \includegraphics[width=0.8\textwidth]{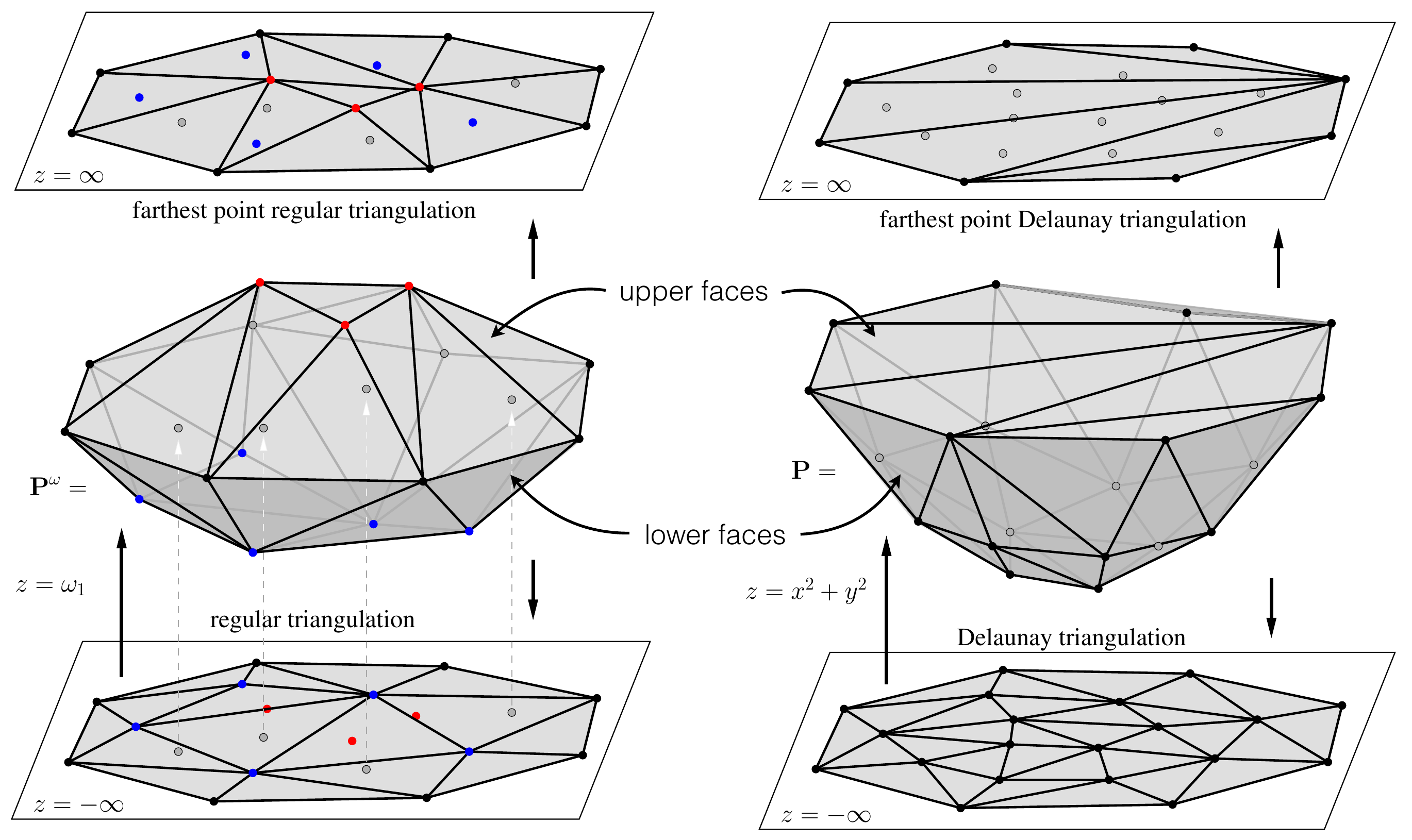}
\caption{
Left: the regular and farthest point regular triangulations of $({\bf A}, \omega)$ in $\mathbb{R}^2 \times \mathbb{R}$. 
Right: the Delaunay and farthest point Delaunay triangulations of ${\bf A}$ when $\omega = x^2 + y^2$. 
}
\label{fig:DT_and_FDT}
\end{figure}

Let ${\bf A}$ be a finite point set in $\mathbb{R}^d$. Pick a {\it height function} $\omega: {\bf A} \to \mathbb{R}$. For a point ${\bf p} = (p_1, \ldots, p_d) \in {\bf A}$, define its {\it lifted point} ${\bf p}' = (p_1, \ldots p_d, p_{d+1}) \in \mathbb{R}^{d+1}$, where $p_{d+1} = \omega({\bf p})$ is called the {\it height} of ${\bf p}$. The point set ${\bf A}^{\omega}  := \{{\bf p}' \,|\, {\bf p} \in {\bf A}\} \subseteq \mathbb{R}^3$ is called the {\it lifted point set} of ${\bf A}$. 

The convex hull of ${\bf A}^{\omega}$, $\textrm{conv}({\bf A}^{\omega})$, is a 3d polytope. 
A {\it lower face} of $\textrm{conv}({\bf A}^{\omega})$ is a facet of $\textrm{conv}({\bf A}^{\omega})$  that has a non-vertical supporting hyperplane with $\textrm{conv}({\bf A}^{\omega})$ lies entirely in the half-space which does not contain the point $({\bf 0}, -\infty)$.
Simply saying, it is a face of $\textrm{conv}({\bf A}^{\omega})$ that is ``visible from below" by placing the viewpoint at $({\bf 0}, -\infty)$.  
Likewise, an {\it upper face} of $\textrm{conv}({\bf A}^{\omega})$ is a facet of $\textrm{conv}({\bf A}^{\omega})$  that has a non-vertical supporting hyperplane with $\textrm{conv}({\bf A}^{\omega})$ lies entirely in the half-space which does not contain the point $({\bf 0}, +\infty)$. Simply saying, it is a face that is ``visible from above" by placing the viewpoint at $({\bf 0}, +\infty)$, see Figure~\ref{fig:DT_and_FDT} for examples.

The {\it canonical projection} is a map $\pi : \mathbb{R}^{d+1} \to \mathbb{R}^d, (p_1, \ldots, p_d, p_{d+1}) \to (p_1, \ldots, p_d)$ which ``deletes the last coordinate" of a given point. 
Let us assume that $\omega$ is {\it sufficiently generic}, by which we mean that no $d+1$ vertices of ${\bf A}$ are lifted to lie in a non-vertical hyperplane in $\mathbb{R}^{d+1}$({This implies that ${\bf A}$ is in general position). 
The canonical projection of the set of lower faces of $\textrm{conv}({\bf A}^{\omega})$ gives a unique   triangulation of ${\bf A}$, which is called the {\it (nearest point) regular triangulation} of $({\bf A}, \omega)$. 
Likewise, the canonical projection of the set of upper faces of $\textrm{conv}({\bf A}^{\omega})$ is also a unique triangulation of ${\bf A}$, which is the {\it farthest point regular triangulation} of $({\bf A}, \omega)$.
Figure~\ref{fig:DT_and_FDT} Left shows an example of these two triangulations of ${\bf A}$.
Note that some points of ${\bf A}$ may be contained in neither the regular nor farthest point regular triangulations of $({\bf A}, \omega)$. 
By varying the heigh function $\omega$, we will obtain different regular triangulations of ${\bf A}$.  
In particular, the Delaunay and the farthest point Delaunay triangulations of ${\bf A}$ are two special triangulations when $\omega = x^2+y^2$, see Figure~\ref{fig:DT_and_FDT} Right. 

Both the regular and farthest point regular triangulations of $({\bf A}, \omega)$ in $\mathbb{R}^d \times \mathbb{R}$ can be constructed efficiently using incremental vertex insertion  algorithms~\cite{Bowyer81,Watson81,Joe91-flip, EdelsbrunnerShah96}.  
In particular, Joe first proved an incremental flip algorithm to construct 3d Delaunay triangulations. The proof~\cite[Lemma 8, 9]{Joe91-flip} is based on the Acyclic Theorem of Edelsbrunner~\cite{Edelsbrunner90acy} (see Section~\ref{sec:acyc}). Edelsbrunner and Shah~\cite{EdelsbrunnerShah96} proved that this algorithm can construct regular triangulations in any dimensions. 

\begin{figure}
  \centering
  \begin{tabular}{c c c}
  \includegraphics[width=0.3\textwidth]{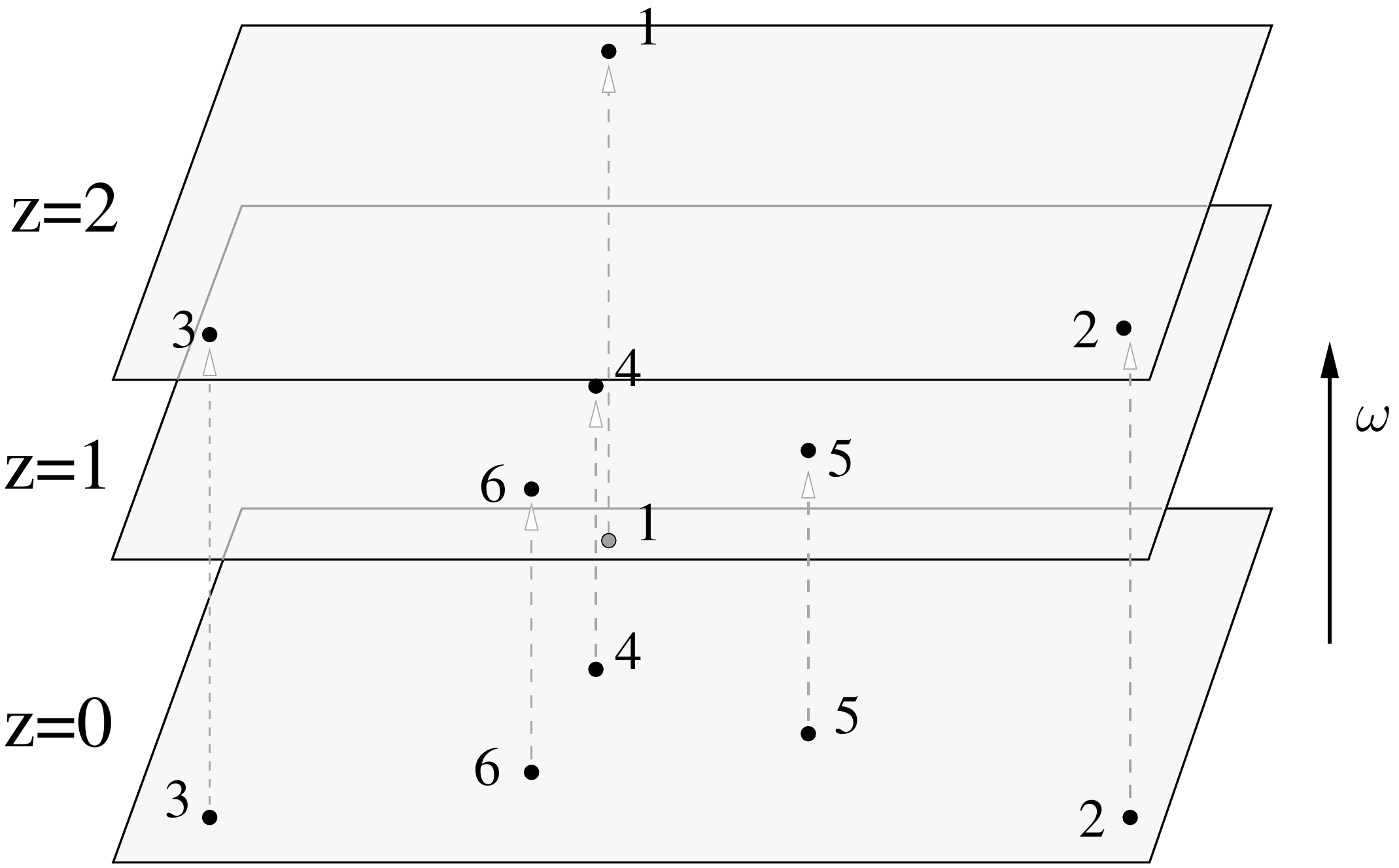} &
  \includegraphics[width=0.3\textwidth]{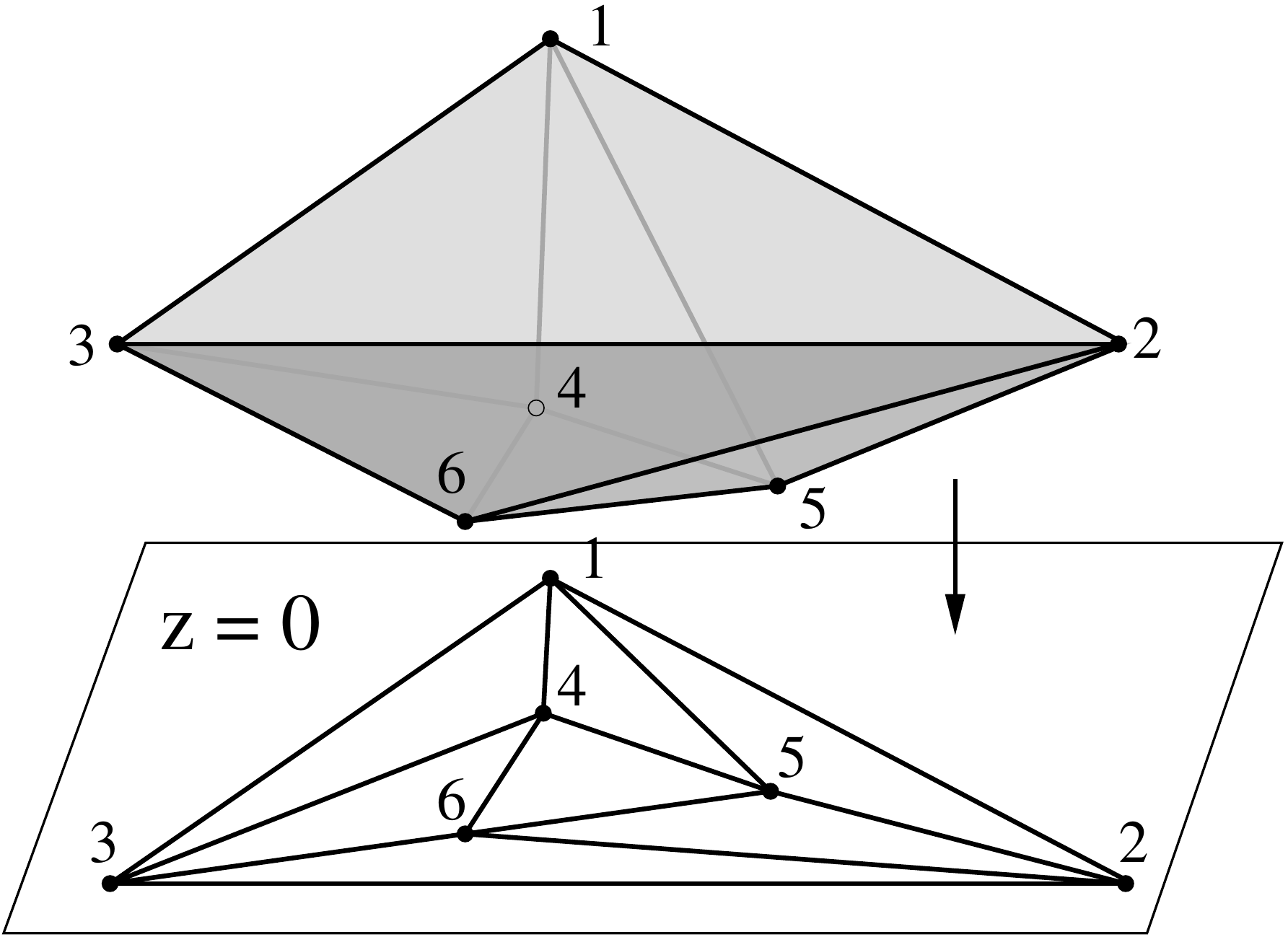} & 
  \includegraphics[width=0.3\textwidth]{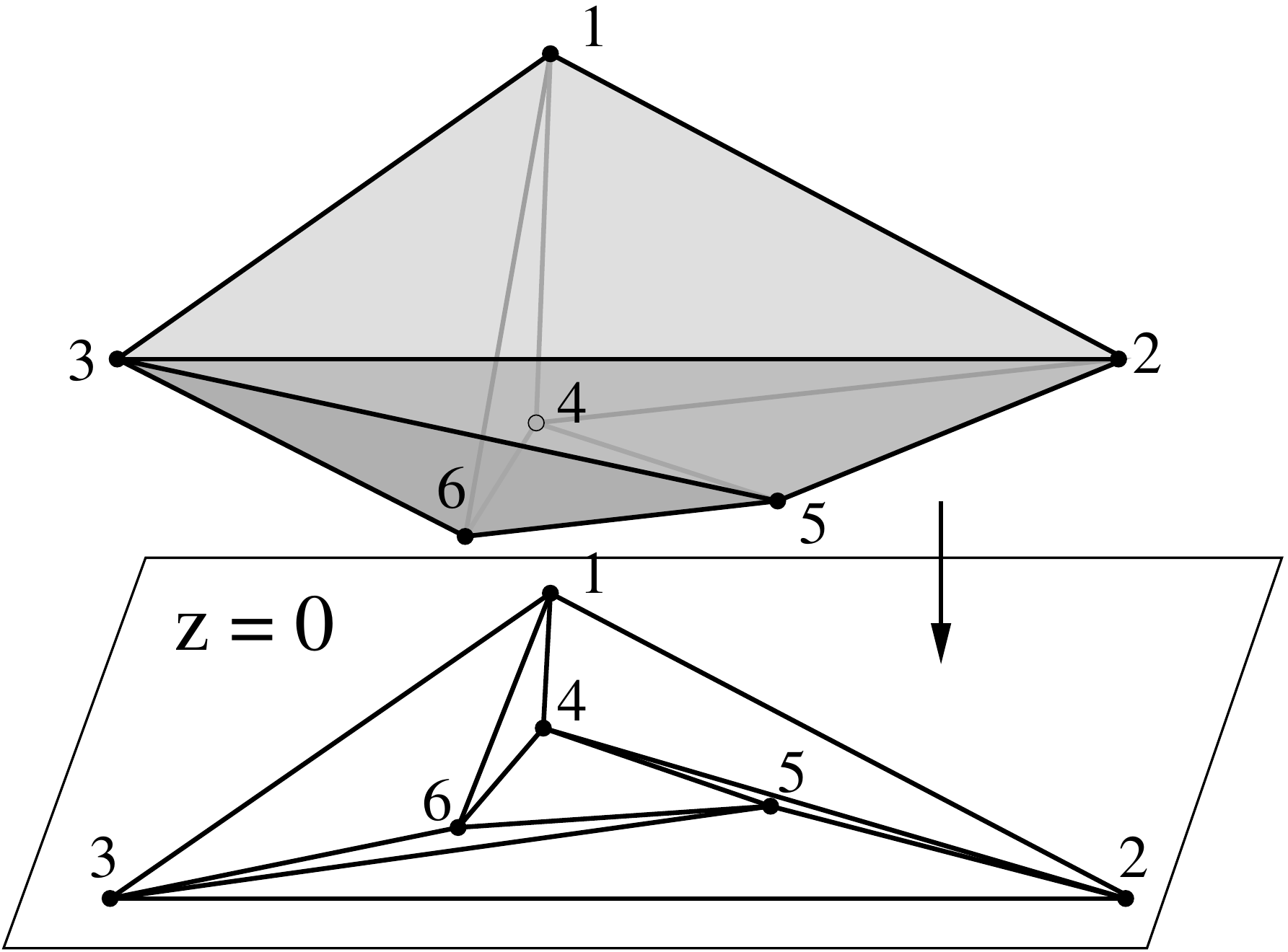} \\
  (1) & (2) & (3)
  \end{tabular}
\caption{Triangulations of a (weighted) point set ${\bf A}$ of $6$ points in the plane.  (1) The lifted point set ${\bf A}^{\omega}$ is produced by the height function $\omega$ which sends the vertices $1,2,3$ to the plane $z = 2$, and the vertices $4,5,6$ to the plane $z = 1$. 
(2) The regular triangulation of $({\bf A}, \omega)$. 
(3) A non-regular triangulation of ${\bf A}$. 
}
\label{fig:nonregular-prism-pts}
\end{figure}

\subsection{Non-regular triangulations}

A triangulation of ${\bf A}$ which is not regular is called a {\it non-regular triangulation} of ${\bf A}$. 
Equivalently, a {\it non-regular triangulation} of ${\bf A}$ is a triangulation of ${\bf A}$ such that there exists no height function $\omega$ which can ``lift" it into the lower or upper envelope of the convex hull of ${\bf A}^{\omega}$.
Although the properties of regular triangulations are well studied, e.g, by the theory of secondary polytopes~\cite{Gelfand1994Discriminants} (see Section~\ref{sec:undirected-flip-graph}), ``not so much is known about non-regular triangulations" (quoted from~\cite{JaumeRote2016}). 

Figure~\ref{fig:nonregular-prism-pts} illustrates the simplest example of a non-regular triangulation of  set of $6$ vertices in the plane. In this example, the height function is chosen such that the outmost three vertices are on a higher plane than the three innermost vertices, see (1). The regular triangulation is shown in (2). The triangulation in (3) is non-regular. 
In fact, no matter which height function $\omega$ is, the triangulation in (3) will correspond to neither the lower nor the upper faces of $\textrm{conv}({\bf A}^{\omega})$. 

\begin{figure}[ht]
  \centering
  \includegraphics[width=0.8\textwidth]{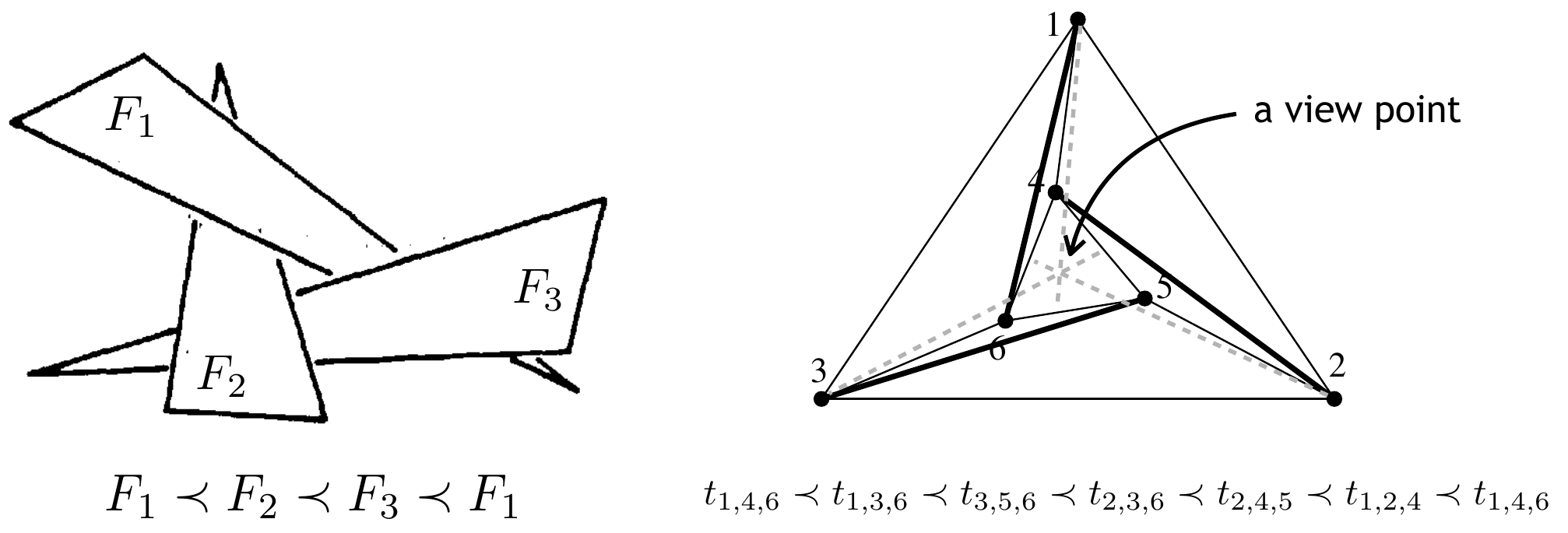}
\caption{Cycles of simplifies in the in\_front/behind relation. Left: the three triangles forms a cycle with respect to your eye position. Right: the three triangles form a cycle with respect to the fixed point. }
\label{fig:cycles_of_simplices}
\end{figure}

Given a triangulation of a point set, to check whether it is regular or not is to find whether there exists a height function $\omega$ (a set of weights on the set of points). This is a linear program problem.  

\subsection{The Acyclic Theorem}
\label{sec:acyc}

Edelsbrunner proved an important property which owed by all regular triangulations~\cite{Edelsbrunner90acy}. 
Let ${\bf x} \in \mathbb{R}^d$ be an arbitrary but fixed viewpoint. 
A simplex $\sigma$ lies {\it in front of} another simplex $\tau$ if there is a half-line starting at ${\bf x}$ that first passes through $\sigma$ and then through $\tau$, denoted as $\sigma \prec_{\bf x} \tau$.  The set of simplices in a triangulation together with $\prec_{\bf x}$ forms a relation. In general, this relation can have cycles, which are sequences $\tau_0 \prec_{\bf x} \tau_1, \ldots, \tau_k \prec_{\bf x} \tau_0$, see Figure~\ref{fig:cycles_of_simplices} Left. 

\begin{theorem}[Edelsbrunner~\cite{Edelsbrunner90acy}]
From any fixed viewpoint in $\mathbb{R}^d$, the in\_front/behind relation for simplices of regular triangulations forms no cycle.
\end{theorem}

This theorem, besides its many applications, provides a useful way to prove the non-regularity of a triangulation:
{\it If a triangulation contains a cycle with respect to any viewpoint, then it is non-regular.}
The cycle with respect to the non-regular triangulation in Figure~\ref{fig:nonregular-prism-pts} (3) is shown in Figure~\ref{fig:cycles_of_simplices} Right.

However, this theorem does not imply that the acyclic property belongs to regular triangulations only.
Indeed, there are non-regular triangulations which are also acyclic viewed from any viewpoint. 
This is first shown by Jaume and Rote~\cite{JaumeRote2016}. 

\subsection{Flips}
\label{subsec:flips}

All point sets in this subsection are assumed to be in {\it general position}, which we mean that no $d+2$ points of a set of $d$-dimensional points lie on the same hyperplane.  

Let ${\bf T}$ be a set of $d+2$ points in $\mathbb{R}^d$, the Radon's theorem~\cite{Radon1921} states that  there exists a partition of ${\bf T}$ into two subsets ${\bf U}, {\bf V} \subset {\bf T}$ such that $\textrm{conv}({\bf U}) \cap \textrm{conv}({\bf V}) \neq \emptyset$.   
This means every triangulation of ${\bf T}$ either contains ${\bf U}$ as a subset, or ${\bf V}$, but not both. 
It is shown by Lawson~\cite{LAWSON1986231}, that a set of $d+2$ points can only have two triangulations. 
A {\it flip} is the operation that substitutes one triangulation of ${\bf T}$ for the other. 

In $\mathbb{R}^2$, there are two Radon partitions of $4$ points. This defines three types of flips, shown in Figure~\ref{fig:flip2d}.
We can distinguish these types of flips by the number of triangles in the triangulations before and after the flip. Hence there are {\it 2-2 flips} (or {\it edge-flips}), {\it 1-3 flips} (or {\it vertex-insertion}), and {\it 3-1 flips} (or {\it vertex-deletion}),   
see Figure~\ref{fig:flip2d}. Flips in $\mathbb{R}^3$ are classified accordingly.

\begin{figure}[ht]
  \centering
  \includegraphics[width=0.6\textwidth]{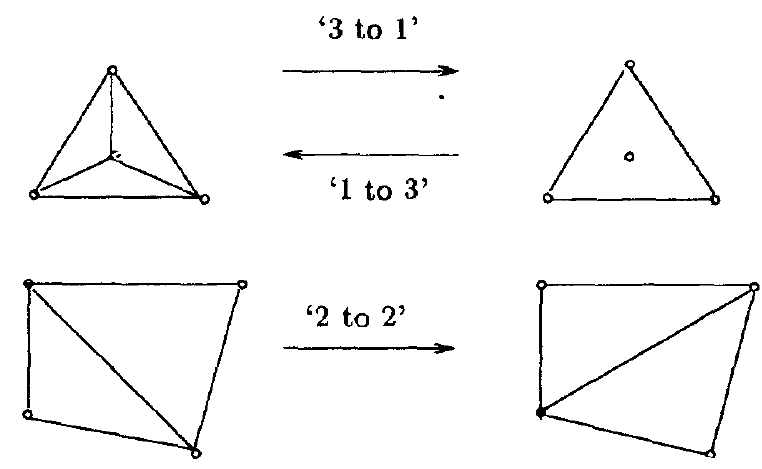}
\caption{Flips in $\mathbb{R}^2$.}
\label{fig:flip2d}
\end{figure}

\paragraph{Flippability} 
Let ${\cal K}$ be  triangulation of a point set ${\bf A}$ in $\mathbb{R}^2$. Consider an edge $e_{\bf ab} \in {\cal K}$. Assume it belongs to two triangles $f_{\bf abc}, f_{\bf abd} \in {\cal K}$. We say $e_{\bf ab}$ is {\it flippable} in ${\cal K}$ if there exists a subcomplex ${\cal S} \subseteq {\cal K}$ such that the underlying space $|{\cal S}|$ is equal to the convex hull of $\{{\bf a,b,c,d}\}$. The edge $e_{\bf ab}$ is flippable implies that either $|{\cal S}| = f_{\bf abc} \cup f_{\bf abd}$ or there exists a third triangle, without loss of generality, assume  ${\bf a}$ is the reflex vertex, $f_{\bf acd} \in {\cal K}$, such that $|{\cal S}| =f_{\bf abc} \cup f_{\bf abd} \cup f_{\bf acd} $.  In the former case, $e_{\bf ab}$ admits a 2-2 flip (an edge flip), and in the latter, it admits a 3-1 flip (vertex removal).  
We say that the convex hull of $\{{\bf a,b,c,d}\}$ is the {\it support} of this flip. 
Otherwise, $e_{\bf ab}$ is {\it unflippable}. This means, if we perform a flip on this edge, it will result a simplical complex which is not a geometric triangulation according to our definition.  Especially, the Intersection Property (IP) does not hold, 
see Figure~\ref{fig:flippability} for examples.  

\begin{figure}[ht]
  \centering
  \includegraphics[width=0.6\textwidth]{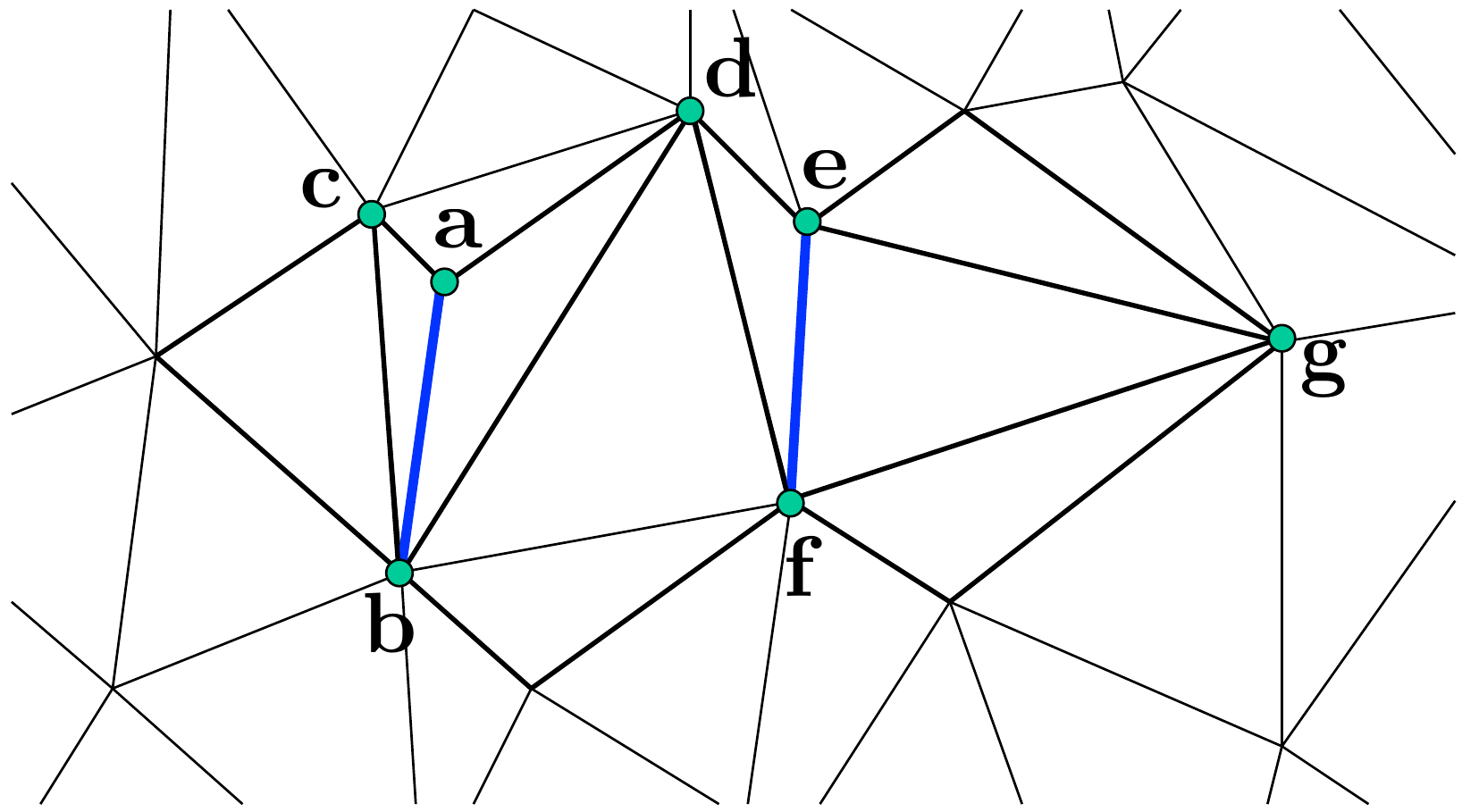}
\caption{Flippability of edges in $\mathbb{R}^2$. Edge $e_{\bf ab}$ is flippable (by a 3-1 flip) whose support of this flip is the convex hull of the set $\{{\bf a},{\bf b}, {\bf c}, {\bf d}\}$, so do the edges $e_{\bf bd}$ and $e_{\bf df}$ (by 2-2 flips), while the edge $e_{\bf ef}$ is unflippable.}
\label{fig:flippability}
\end{figure}

\subsection{Lawson's flip algorithm}

A classical algorithm (due to Charles Lawson~\cite{Lawson1977}) shows that one can transform any triangulation of ${\bf A}$ into the Delaunay triangulation of ${\bf A}$ by a sequence of edge flips. 

Let ${\mathcal K}$ be a triangulation of ${\bf A}$ in $\mathbb{R}^2$. An edge $e_{\bf ab} \in {\mathcal K}$ is {\it locally Delaunay} if either
(i) it is on the convex hull of ${\bf A}$, or
(ii) it belongs two triangles, $f_{\bf abc}, f_{\bf abd} \in {\cal K}$, and ${\bf d}$ does not lie inside of the circumcircle of $f_{\bf abc}$.
A locally Delaunay edge is not necessarily a Delaunay edge. The Delaunay Lemma~\cite{Delaunay1934} proved  that if every edge of ${\mathcal K}$ is locally Delaunay, then ${\mathcal K}$ is the Delaunay triangulation of ${\bf A}$.

\begin{figure}[ht]
  \centering
  \includegraphics[width=0.8\textwidth]{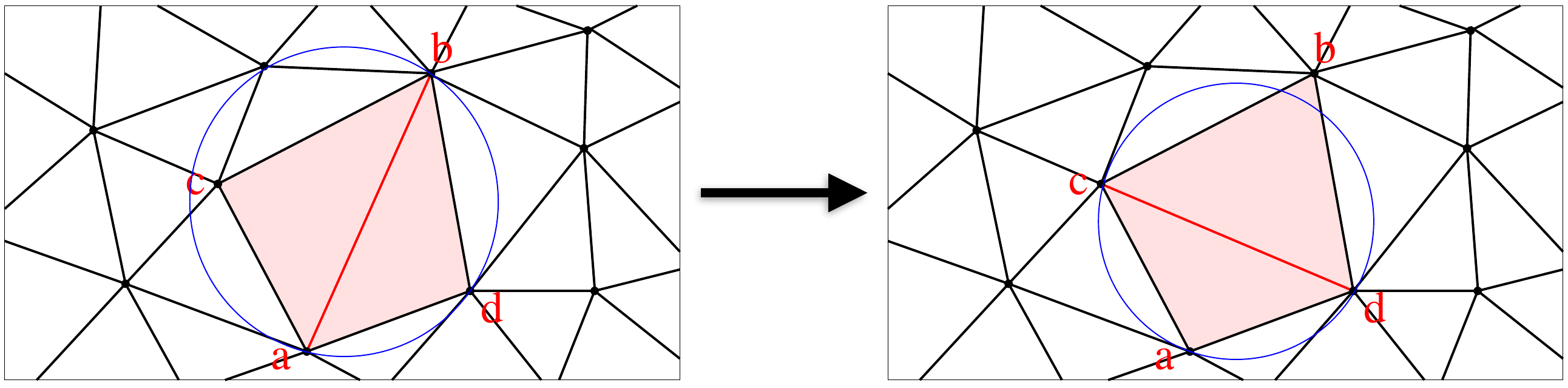}
  \caption{The edge flips from $e_{\bf ab}$ to $e_{\bf cd}$ replaces the locally non-Delaunay edge $e_{\bf ab}$ by the locally Delaunay edge $e_{\bf cd}$.}
  \label{fig:edge_flip}
\end{figure}

If an edge $e_{\bf ab} \in {\mathcal K}$ is not locally Delaunay, then it must be flippable (an easy exercise). A 2-2 flip can be used to replace $e_{\bf ab}$ by the edge $e_{\bf cd}$ in ${\mathcal K}$ such that  $e_{\bf cd}$ must be locally Delaunay, see Figure~\ref{fig:edge_flip}. We can use edge flips as elementary operations to convert an arbitrary triangulation ${\mathcal K}$ to the Delaunay triangulation. 

The Lawson's algorithm uses a stack to maintains all edges which may be locally non-Delaunay. Initially, all edges of ${\mathcal K}$ are pushed on the stack. It then pop these edges (one after one) and performs edge flips on those edges which are not locally Delaunay. After each flip, the stack will be updated by new edges on the convex hull of the support of this flip. The algorithm finishes when the stack is empty.

Lawson's algorithm can be understood as gluing a sequence of tetrahedra. A locally non-Delaunay edge $e_{\bf ab}$ corresponding a locally non-convex edge $e_{\bf a'b'}$ in $\mathbb{R}^3$. Flipping $e_{\bf ab}$ to $e_{\bf cd}$ is likely gluing a tetrahedron $t_{\bf a'b'c'd'}$ from below to the triangles $f_{\bf a'b'c'}$ and $f_{\bf a'b'd'}$ in $\mathbb{R}^3$. 
Once we glue $t_{\bf a'b'c'd'}$ we cannot glue another tetrahedron right below $e_{\bf a'b'}$. In other words, once we flip $e_{\bf ab}$ we cannot introduce $e_{\bf ab}$ again by some other flip.  This implies that Lawson's algorithm will eventually terminate when all locally non-Delaunay edges are flipped. By the Delaunay lemma, the triangulation is Delaunay. Figure~\ref{fig:LawsonflipAlg_lift} illustrates the algorithm in both 2d and 3d.
This also implies there are at most as many flips as there are edges connecting $n$ points, namely $n \choose 2$. Each flip takes constant time, hence the total running time is $O(n^2)$. 

\begin{figure}[ht]
  \centering
  \includegraphics[width=0.9\textwidth]{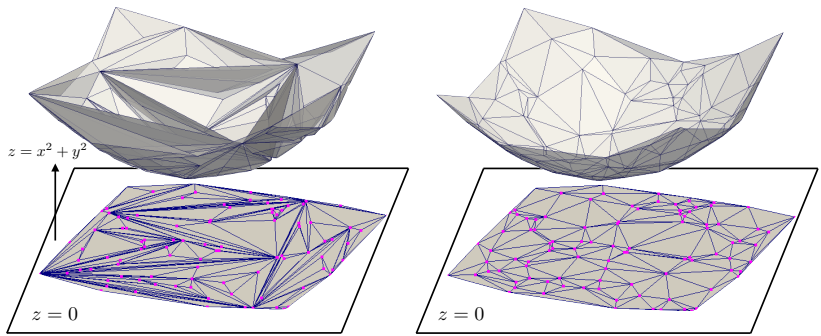}
  \caption{The lifted view of Lawson's flip algorithm which transforms a non-convex surface (left) in 3d into a convex one (right). }
  \label{fig:LawsonflipAlg_lift}
\end{figure}

\subsection{The undirected flip graph}
\label{sec:undirected-flip-graph}

\begin{figure}
  \centering
  \includegraphics[width=0.6\textwidth]{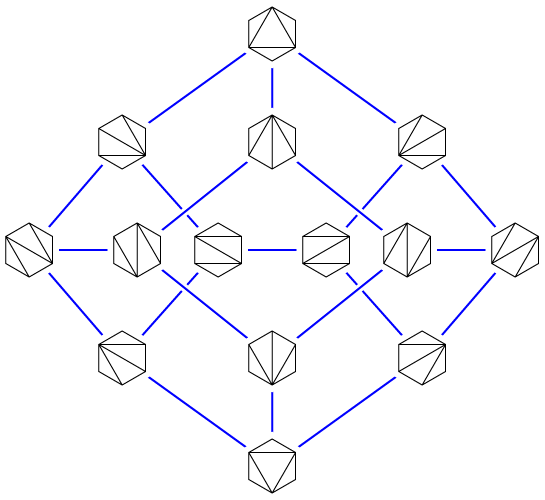}
\caption{The (undirected) flip graph of a the vertex set of a hexagon (\cite{EPPSTEIN2009790}).}
\label{fig:flipgraph-Eppstein}
\end{figure}

The {\it undirected flip graph} of a finite point set ${\bf A}$ is defined as follows: Each triangulation of ${\bf A}$ is a node, and there is an edge from $\mu$ to $\nu$ if there is a flip that changes the triangulation $\mu$ to $\nu$, see Figure~\ref{fig:flipgraph-Eppstein} for an example. 

Lawson's algorithm~\cite{Lawson1977} shows that one can transform any triangulation of ${\bf A}$ into the Delaunay triangulation of ${\bf A}$ by a sequence of edge flips. 
Since one can always add or remove a vertex in a 2d triangulation by a sequence of 2-2, 1-3 and 3-1 flips. 
This shows that the (undirected) flip graph of any 2d point set is connected. 

The undirected flip graph of a point set in three and higher dimensions are much more complicated. 
It was first proven by Joe~\cite{Joe1989} that Lawson's algorithm may fail in 3d. Moreover, the failure is due to the existence of a cycle of connected unflippable locally non-Delaunay faces~\cite[Lemma 4, 5]{Joe1989}. 

A construction of a triangulation of a point set in $\mathbb{R}^5$ given by Santos~\cite{Santos2000,Santos2000a} shows that the undirected flip graph of a point set in $\mathbb{R}^5$ is disconnected.  This leads to the following open problem.

\begin{prob}[Open Problem]
Whether the undirected flip graph of a point set in $\mathbb{R}^3$ or $\mathbb{R}^4$ is connected or not. 
\end{prob}

A nice result from the secondary polytope theory~\cite{Gelfand1994Discriminants} shows that the undirected flip graph of the set of all regular triangulations of any point set in $\mathbb{R}^d$ is connected.  This implies the following very useful theorem which will be extensively used in this paper. A proof of this theorem is found in the book~\cite[Theorem 5.3.7]{TriangBook}. The definition of ``directed flips" is given later in Section~\ref{sec:graph-poset}. 

\begin{theorem}~\label{thm:secondarypolytope}
Between every two regular triangulations ${\cal T}_1$ and ${\cal T}_2$ of ${\bf A}$, there exists a monotone sequence of directed flips which starts at ${\cal T}_1$ and ends at ${\cal T}_2$ or vice versa. 
\end{theorem}

Another interesting example is the vertex sets of {\it cyclic polytopes}, see e.g.~\cite[Chap 0]{Ziegler1995-book} and~\cite[Section 6.1]{TriangBook}.  It has been shown that there is a ``friendly" poset structure on the set of all triangulations of cyclic polytopes (introduced below). The Hasse diagram of this poset will be the graph of all triangulations. In particular, this  graph is connected, see~\cite{Edelman1996,Rambau96-thesis} and~\cite[Section 6.1]{TriangBook}. 

For the most recent and comprehensive exposition of undirected flip graphs, as well as the most latest results, we refer to the book~\cite{TriangBook}.

\subsection{Higher Stasheff-Tamari (HST) posets}

Besides the undirected flip graph, there is another structure on a point set by placing a partial order on triangulations of this point set.

Let ${\bf C}$ be the vertex set of a {\it cyclic polytope} ${\bf C}(n, 2)$ in $\mathbb{R}^2$ (see e.g.~\cite[Chap 0]{Ziegler1995-book}), where $n$ is the number of points of ${\bf C}$. 
Let the four points ${\bf p}_i, {\bf p}_j, {\bf p}_k, {\bf p}_l \in {\bf A}$, where $i < j < k < l$, form a quadrilateral.  A flip is called {\it up-flip} if it replaces the diagonal edge $[{\bf p}_i, {\bf p}_k]$ by $[{\bf p}_j, {\bf p}_l]$. It is a {\it down-flip} otherwise.  Let ${\cal T}_1, {\cal T}_2$ be two triangulations  of ${\bf C}$, we say ${\cal T}_1 \le_1 {\cal T}_2$ when ${\cal T}_2$ can be produced from ${\cal T}_1$ by a (possibly empty) sequence of up-flips. 
Then the set of all triangulations of ${\bf C}$ is a partially ordered set (poset), see Figure~\ref{fig:HST1-Rambau} Left.  
It is known as the {\it first Higher Stasheff-Tamari poset}, denoted as ${\bf HST}_1$.
By using the height of the characteristic sections~\cite[Def. 6.1.15]{TriangBook} as the partial order, one obtain another  poset structure known as the {\it second Higher Stasheff-Tamari poset}, denoted as ${\bf HST}_2$, see Figure~\ref{fig:HST1-Rambau} Right.

\begin{figure}
  \centering
  \includegraphics[width=0.8\textwidth]{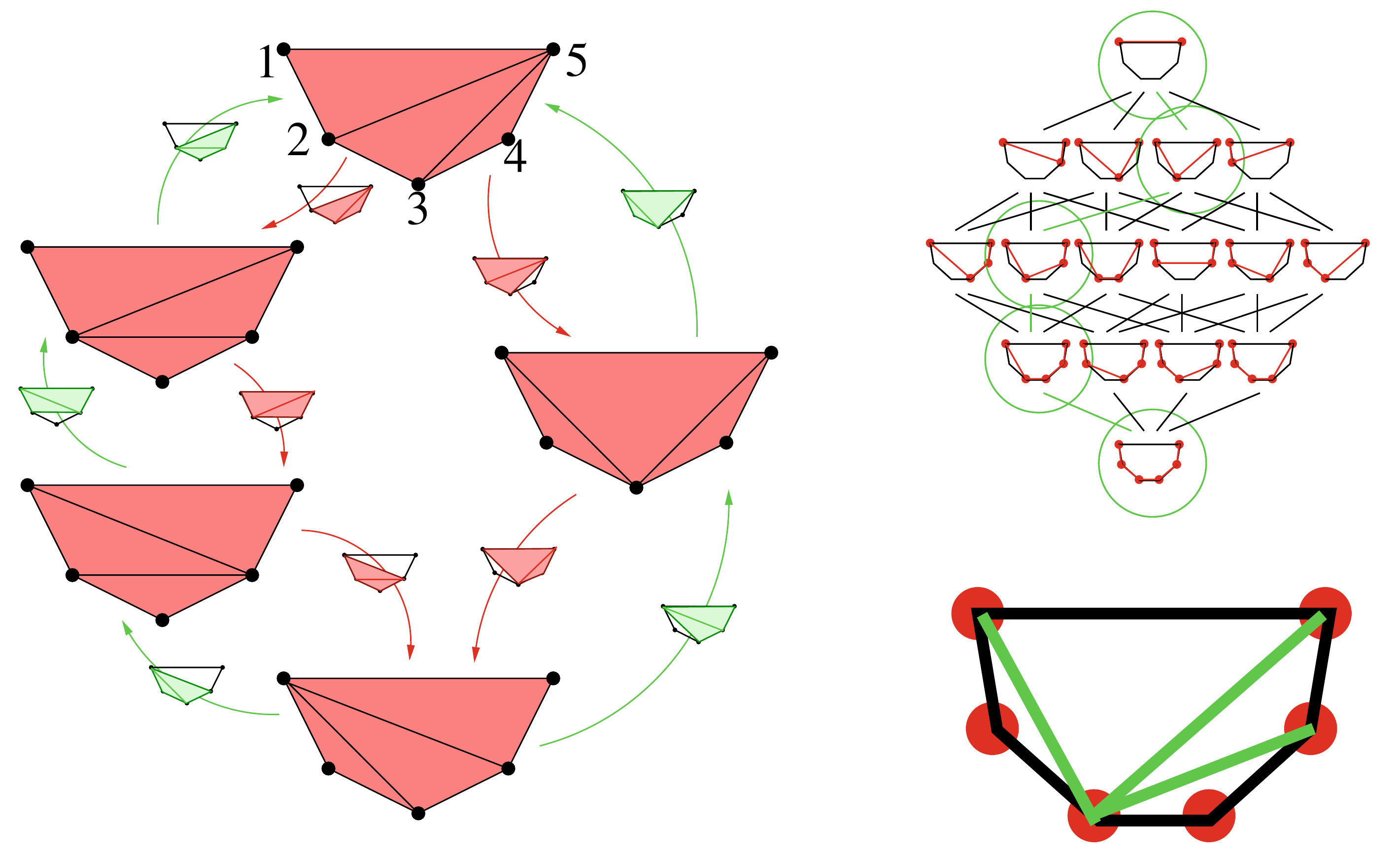}
\caption{Examples of higher Stasheff-Tamari  (HST) Posets for triangulations of cyclic polytopes (Courtesy of~\cite{TriangBook}).  Left: ${\bf HST}_1$ of the vertex set a $5$-gon. It is the vertex set of the $2d$ cyclic polytope ${\bf C}(5,2)$. Right: ${\bf HST}_2$ of the vertex set ${\bf C}(6,1)$, and a triangulation of ${\bf C}(6,2)$ resulted by a maximal chain.}
\label{fig:HST1-Rambau}
\end{figure}

These two poset link to many nice combinatorial and geometric objects. 
Its property has been well studied~\cite{Edelman1996,Rambau96-thesis}. 
The following theorem is about the structure properties of triangulations of cyclic polytopes proven in~\cite{TriangBook}. It is proven for cyclic polytope of any dimension.  
We state the two-dimension case only. 

\begin{theorem}~\cite[Theorem 6.1.19]{TriangBook}
\label{thm:poset-cyclicpoly}
Let ${\bf A}$ be the point set of $C(n, 2)$. Let $\omega$ be a 
height function defined on the vertices of ${\bf A}$. Then
\begin{itemize}
\item[(1)] This poset is bounded. It has a unique minimum which is the  regular triangulation of $({\bf A}, \omega)$. It has a unique maximum, which is the farthest-point regular triangulation of $({\bf A}, \omega)$.

\item[(2)] For any triangulation of ${\bf A}^{\omega}$ there exists a maximum sequence of flips in the poset, i.e., all triangulations of ${\bf A}^{\omega}$ are coded in this poset. 


\item[(3)] All internal nodes of this poset are 2d regular triangulations of ${\bf A}$. 
\end{itemize}
\end{theorem}

It is shown that the Hasse diagram of this poset will be the graph of all triangulations. 
It is a lattice, which means, it has unique minimum and maximum elements and every two elements have a unique minimum upper bound. 
In particular, there is a bijective relation between the maximal chain (sequence) of flips and all triangulations of the lifted cyclic polytope in $\mathbb{R}^{3}$, see an example in Figure~\ref{fig:HST1-Rambau} Right. This is the key fact to show that the flip graph of the vertex set (up to $\mathbb{R}^5$) is connected. 

\subsection{Triangulating non-convex polyhedra}

Decomposing a 3d polyhedron into a set of tetrahedra that forms a simplicial complex is a classical problem in computational geometry. 

It is shown that this problem has many theoretical difficulties. There exists simple non-convex  polyhedra whose interior cannot be triangulated with its own vertices. The smallest example is constructed by Sch\"onhardt in 1928~\cite{Schonhardt1928}, which is a twisted triangular prism of $6$ vertices, see Figure~\ref{fig:indecomposable} Left. Other construction of such polyhedra are reported in~\cite{Bagemihl48-decomp-polyhedra,Jessen1967,Chazelle1984,Rambau05,Bezdek2016,SI201892}.  Ruppert and Seidel~\cite{RuppertSeidel92} proved that the problem to determine whether a simple non-convex polyhedron can be triangulated without Steiner points is NP-complete.  

\begin{figure}[ht]
  \centering
  \includegraphics[width=0.8\textwidth]{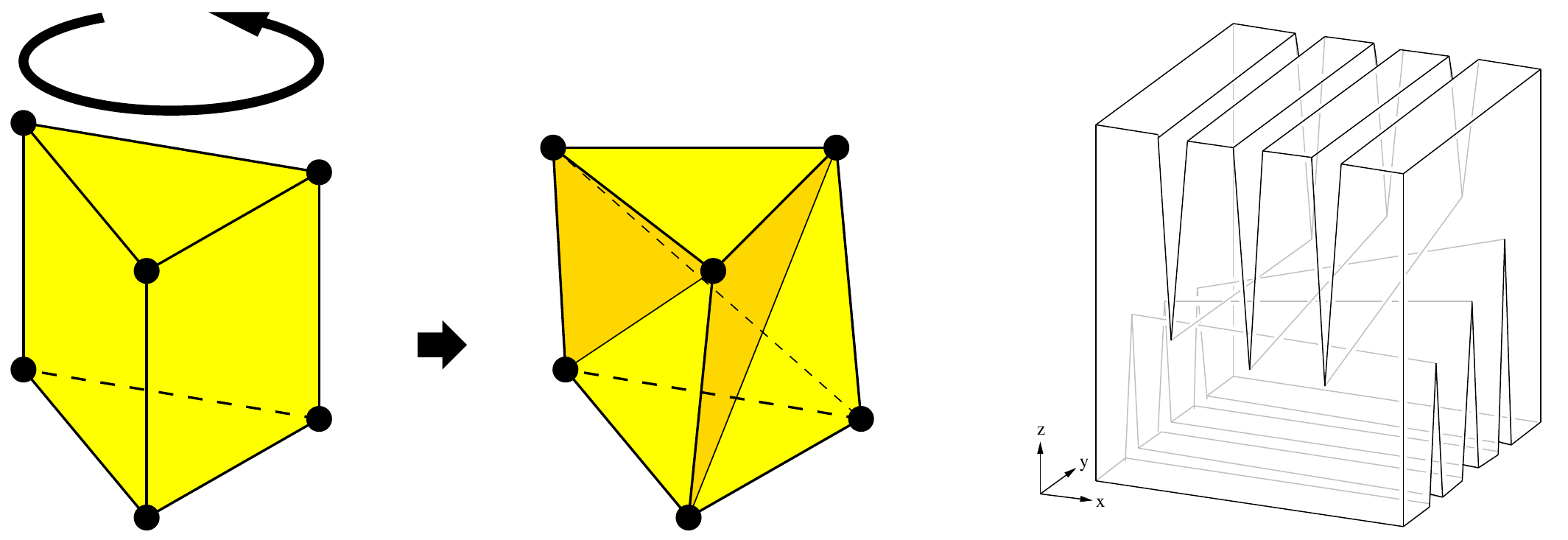}
\caption{Left: The Sch\"onhardt polyhedron. Right: A Chazelle polyhedron.}
\label{fig:indecomposable}
\end{figure}

If addition vertices are allowed, so-called {\it Steiner points}, Chazelle~\cite{Chazelle1984} constructs a family of polyhedra with $n$ vertices and proved that $\Omega(n^2)$ Steiner points are needed to triangulate them, see Figure~\ref{fig:indecomposable} Right.  Chazelle also proved that any simple polyhedron of $n$ vertices can be triangulated into $O(n^2)$ tetrahedra~\cite{Chazelle1984}.  In practice, a challenging question is: {\it Given any polyhedron, how to triangulate it with the number of Steiner points as small as possible?}

The problem is much difficult when the boundary of the polyhedra needs to be preserved, i.e., Steiner points, if they are needed, are only allowed in the interior of the polyhedra. 
There are algorithms based on heuristics, see e.g~\cite{GeorgeHechtSaltel91,WeatherillHassan94,GeorgeBorouchakiSaltel03,Si2015-TetGen}, but none of them provides any bound on the number of interior Steiner points. 

There are some interesting results on triangulating special classes simple and non-simple polyhedra without Steiner points.  We refer to the work of Goodman and Pach~\cite{Goodman1988}, Bern~\cite{Bern93-tetra}, and Toussaint et al~\cite{Toussaint93}. 
In particular, Toussaint et al~\cite{Toussaint93}  showed that certain classes of rectilinear (isothetic) simple polyhedra can always be triangulated in $O(n^2)$ time where $n$ is the number of vertices in the polyhedron. They also showed that polyhedral slabs (even with holes) as well as subdivision slabs can always be triangulated  in $O(n \log n)$ time. Furthermore, for simple polyhedral slabs $O(n)$ time suffices. 
They showed that polyhedra that are the union of three convex polyhedra can always be triangulated in $O(n^2)$ time.  They remarked that: {\it one direction of future research is to find other nontrivial classes polyhedra
which can be triangulated.}

\section{Monotone Sequences of Directed Flips and the Directed Flip Graph}
\label{sec:graph-poset}



Let ${\bf A}$ be a finite point set in $\mathbb{R}^2$ and let $\omega: {\bf A} \to  \mathbb{R}$ be a height function which lifts every vertex ${\bf p} = (p_1, p_2) \in {\bf A}$ into a lifted point ${\bf p}' = (p_1, p_2, \omega({\bf p})) \in \mathbb{R}^{3}$. 
Let ${\bf A}^{\omega} = \{{\bf p}' \;|\; {\bf p} \in {\bf A}\}$ in $\mathbb{R}^{3}$ be the set of lifted points of ${\bf A}$. 

Let $\{{\bf a}, {\bf b},{\bf c},{\bf d}\} \subset {\bf A}$ be the support of a flip (2-2 flip, 1-3 flip, or 3-1 flip) in a triangulation of ${\bf A}$. This flip corresponds to a tetrahedron $t_{\bf a'b'c'd'}$, where $\{{\bf a}', {\bf b}',{\bf c}',{\bf d}'\} \subset {\bf A}^{\omega}$.  (The definition of the support of a flip is given in Section~\ref{subsec:flips}.)
Let ${\cal T}_1$ and ${\cal T}_2$ be two triangulations of ${\bf A}$ differing by a flip. We say ${\cal T}_1$ is {\it decreasing with respect to the height function $\omega$} if the upper faces of $t$ appear in ${\cal T}_1$ and the lower faces of $t$ appear in ${\cal T}_2$. In this case, we call this flip a {\it down-flip}. 
Likewise, we say  ${\cal T}_1$ is {\it increasing with respect to the height function $\omega$} if it is the other way around. In this case, we call this flip an {\it up-flip}. The two types of directed flips are illustrated in Figure~\ref{fig:directed_flips}.


We say a sequence of directed flips {\it monotone} if all directed flips in this sequence are of the same type. 


Similar to the construction of the first Higher Stasheff-Tamari (${\bf HST}_1$) poset~\cite{Rambau96-thesis,TriangBook}.
We can construct a directed flip graph on the set of triangulations of a point set by using one of the two types of directed flips.

We first define a partial order on the set of triangulations of ${\bf A}$ which has a height function $\omega$ defined on vertices of ${\bf A}$.  Let ${\cal T}_1$ and ${\cal T}_2$ be two triangulations of ${\bf A}$. We say that ${\cal T}_1 \le_{1} {\cal T}_2$ if ${\cal T}_2$ can be produced from ${\cal T}_1$ by a (possibly empty) sequence of up-flips. Likewise, we say ${\cal T}_1 \ge_{1} {\cal T}_2$ if ${\cal T}_2$ can be produced from ${\cal T}_1$ by a (possibly empty) sequence of down-flips.

The $\le_1$ relation defines a poset on the set of triangulations of ${\bf A}$.  If ${\bf A}$ is the vertex of a cyclic polytope, then this poset is exactly ${\bf HST}_1$.  What would be a poset of an arbitrary point set look like?

\subsection{A motivation example}
\label{sec:graph-poset-example}

Figure~\ref{fig:prism_poset_1} shows the poset of the point set ${\bf A}$ of $6$ points with a given height function $\omega$.  It is the smallest example which contains a non-regular triangulation. 

This poset is produced using the up-flips. All sequence of up-flips are starting from the regular triangulation  (bottom) to the farthest point regular triangulation (top) of $({\bf A}, \omega)$. In each triangulation of ${\bf A}$, the green edges are locally non-regular with respect to up-flips and flippable. The (thin) black edges are also locally non-regular, but not flippable. The thick black edges are locally regular.  

Each arrow in this graph represents an up-flip, which is either a $\operatorname{2-2}$ or a $\operatorname{3-1}$ flip. 
The numbers in front of each arrow, like $\operatorname{34-56}$ or $\operatorname{123-4}$,  indicate the involved vertices which support this flip. For example, $\operatorname{34-56}$ means a $\operatorname{2-2}$ flip which replaces edge $e_{\bf 34}$ by $e_{\bf 56}$, and $\operatorname{123-4}$ means a $\operatorname{3-1}$ flip which removes the vertex $v_{\bf 4}$ from the triangle $f_{\bf 123}$. 
The four vertices of each flip are the vertices of a tetrahedron.   
Note that in this example, there is no $\operatorname{3-1}$ flip. 

From this graph, we observe some interesting properties of this poset. 
\begin{itemize}
\item There is a unique minimum, which is the regular triangulation of $({\bf A}, \omega)$, denoted as ${\cal T}_{reg}$.
\item There is no unique maximum. Instead, there are two top elements, which are the farthest point regular triangulation of $({\bf A}, \omega)$, denoted as ${\cal T}_{far}$, and the non-regular triangulation of ${\bf A}$, denoted as ${\cal T}_{nonreg}$.
\item There is no path (up-flips) from ${\cal T}_{nonreg}$ to ${\cal T}_{far}$. 
\end{itemize}

This graph could partially explain the behaviour of Lawson's flip algorithm~\cite{Lawson1977}:

\begin{itemize}
\item Any sequence of down-flips (which are reversing the arrows) will reach ${\cal T}_{reg}$. Even it starts from ${\cal T}_{nonreg}$. This explains that the Lawson's algorithm with respect to the down-flips must terminate at the unique minimum, the regular triangulation of $({\bf A}, \omega)$.  
\item However, it may fail to reach the ${\cal T}_{far}$ with respect to up-flips. There exists a sequence which ends at ${\cal T}_{nonreg}$.   
\end{itemize}

We call a sequence of up-flips {\it maximal} if it starts from the ${\cal T}_{reg}$ and ends at the ${\cal T}_{far}$ of $({\bf A}, \omega)$. 
From~\ref{fig:prism_poset_1} (4) we could observe there are many maximal sequences. 
In the next section, we will prove that every maximal sequence gives a tetrahedralisation of ${\bf A}^{\omega}$. 
It is not hard to observe that several maximal sequences may lead to the same tetrahedralisation. 
They  contain the same set of flips but with have different orders. 
By a filtering of maximal sequences with the same set of flips, we can get all six tetrahedralisations of ${\bf A}^{\omega}$, shown in Figure~\ref{fig:prism_poset_2} Top. 

\begin{figure}[ht]
  \centering
  \includegraphics[width=.9\textwidth]{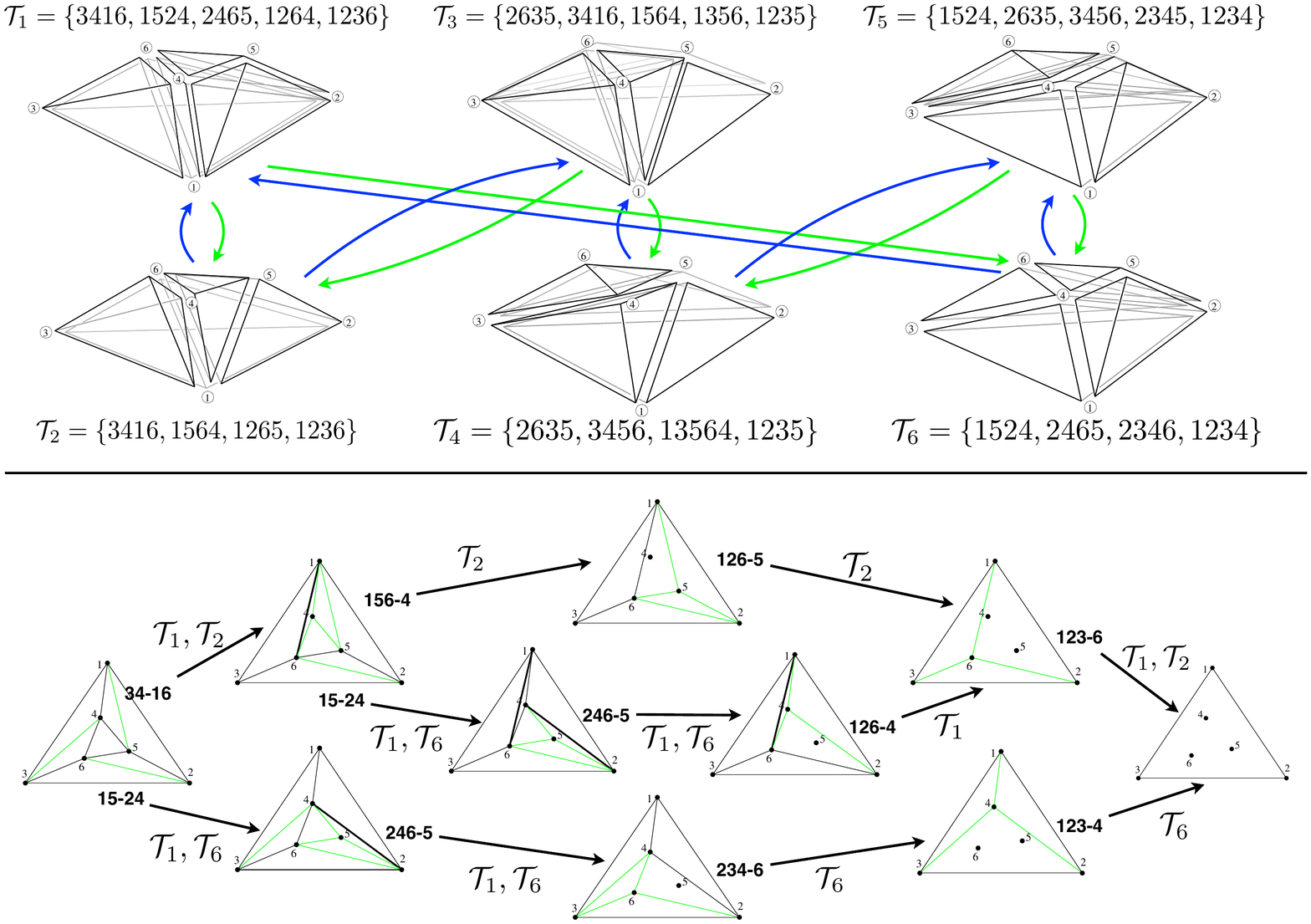}
  \caption{Top: the (indirected) flip graph of the six tetrahedralisations ${\cal T}_1, \ldots, {\cal T}_6$ of ${\bf A}^{\omega}$. A green arrow indicates a $\operatorname{3-2}$ flip and a blue arrow indicates a $\operatorname{2-3}$ flip, respectively.  Bottom: The three sequences of flips corresponding to the tetrahedralisations ${\cal T}_1$, ${\cal T}_2$, and ${\cal T}_6$ are shown, respectively. The 3d flips between ${\cal T}_1$,  ${\cal T}_2$ and ${\cal T}_1$,  ${\cal T}_6$ can be observed. }
  \label{fig:prism_poset_2}
\end{figure}

Note that the lengths of the six tetrahedralisations are either $4$ or $5$. This is due to 3d flips, $\operatorname{2-3}$ or $\operatorname{3-2}$ flips, between these tetrahedralisations. It could be indeed observed from this poset (see Figure~\ref{fig:prism_poset_2} Bottom):

\begin{itemize}
\item If two maximal sequences have the same length, then there is no 3d flip between them.  
\item If two maximal sequences have two elements in common, and their lengths are exactly different by one, then there is a 3d flip between their corresponding tetrahedralisations.
\end{itemize}

By the above observation, we could ``see" that the set of six tetrahedralisations is connected by 3d flips. In other words, the undirected flip graph of ${\bf A}^{\omega}$ is connected. 

Finally, we discuss the existence of the non-regular triangulation ${\cal T}_{nonreg}$ of ${\bf A}$.  
Since there is no path between ${\cal T}_{nonreg}$ and ${\cal T}_{far}$, this shows that the 3d non-convex polyhedron with ${\cal T}_{nonreg}$ and ${\cal T}_{far}$ as its boundary is indecomposable with its own vertices. This links to the well-known Sch\"onhardt polyhedron~\cite{Schonhardt1928}. 

\section{From Monotone Sequences of Directed Flips to Triangulations of Polyhedra}
\label{sec:monotone-to-triang}

Recall a {\it triangulation} of a 3d point set is a 3d simplical complex such that its underlying space is this convex hull of this point set.  
Let ${\bf A}$ be a finite point set in $\mathbb{R}^2$ and let $\omega: {\bf A} \to  \mathbb{R}$ be a height function defined on the vertices of ${\bf A}$.  In this section, we prove the following theorem.

\begin{theorem}~\label{thm:monotone-to-triang}
Any monotone sequence of directed flips which transforms from the regular triangulation into the farthest point regular triangulation of $({\bf A}, \omega)$ (or vice versa) corresponds to a 3d triangulation of ${\bf A}^{\omega}$. 
\end{theorem}

The theorem in one dimension lower is illustrated in Figure~\ref{fig:HST1-Rambau} Right. A nice example shows that a monotone sequence of up-flips between a set of 1-dimensional triangulations produces a triangulation of 2-dimensional polygon. 

\subsection{Characteristic sections and lifted triangulations}

We first introduce a very useful geometric connection between a triangulation in the plane and another one in $\mathbb{R}^3$. 

Let ${\cal T}$ be a triangulation of ${\bf A}$ in $\mathbb{R}^2$, 
 and let $\omega : {\bf A} \to \mathbb{R}$ be a height function defined on vertices of ${\bf A}$.  
We can define a piecewise-linear function by
\[
  g_{\omega, {\cal T}} : \textrm{conv}({\bf A}) \to \mathbb{R}^3, \; g_{\omega, {\cal T}}({\bf p}) := \omega(\bf p), \; \forall {\bf p} \in {\bf A},
\]
and extended affinely on each simplex $\sigma \in {\cal T}$.  
This function is called the {\it characteristic section} of ${\cal T}$~\cite[Chap 5, Definition 5.2.12, page 229]{TriangBook}.  

Geometrically, $g_{\omega, {\cal T}}$ is a piecewise flat surface in $\mathbb{R}^3$ whose canonical projection (which deletes the last coordinate of any point ${\bf x} \in g_{\omega, {\cal T}}$) is the planar domain $\textrm{conv}({\bf A})$. Equivalently, it lifts the planar domain $\textrm{conv}({\bf A})$ into a piecewise flat surface in $\mathbb{R}^3$. 

We define the {\it lifted triangulation} of ${\cal T}$ via $\omega$ as a two-dimensional triangulation ${\cal T}^{\omega}$ embedded in $\mathbb{R}^3$, such that:
\begin{itemize}
\item[(1)] The vertex set of ${\cal T}^{\omega}$ is the lifted vertex set of ${\cal T}$. It is a subset of ${\bf A}^{\omega}$.
\item[(2)] There is a bijection between the sets of simplices of ${\cal T}^{\omega}$ and ${\cal T}$, i.e., every simplex of ${\cal T}$ is mapped uniquely to a simplex of ${\cal T}^{\omega}$. 
\item[(3)]  The union of all simplices of ${\cal T}^{\omega}$ is equal to the characteristic section $g_{\omega, {\cal T}}$. 
\end{itemize}
It is easy to verify that the lifted triangulation ${\cal T}^{\omega}$ is a triangulation of the characteristic section $g_{\omega, {\cal T}}$. Figure~\ref{fig:LawsonflipAlg_lift} shows two planar triangulations and their corresponding lifted triangulations in $\mathbb{R}^3$. 

Every triangulation of ${\bf A}$ corresponds to a lifted triangulation via $\omega$.
In particular, the regular and the farthest point regular triangulations of $({\bf A}, \omega)$ correspond to the two extreme lifted triangulations of ${\bf A}$, denoted as ${\cal R}^{\omega}$ and ${\cal F}^{\omega}$, which are the sets of lower and upper faces of $\textrm{conv}({\bf A}^{\omega})$, respectively.  All other lifted triangulations of ${\bf A}$ must lie between ${\cal R}^{\omega}$ and ${\cal F}^{\omega}$. 


We say that the characteristic section $g_{\omega, {\cal T}_1}$ {\it lies vertically below} $g_{\omega, {\cal T}_2}$ 
if for any two points ${\bf x}_1 = (x, y, z_1) \in g_{\omega, {\cal T}_1}$ and ${\bf x}_2 = (x, y, z_2) \in g_{\omega, {\cal T}_2}$, the inequality $z_1 \le z_2$ holds. This means $g_{\omega, {\cal T}_1}$ is nowhere higher than $g_{\omega, {\cal T}_2}$.

The relation ``characteristic section lies vertically below" defines a partial order on the set of triangulations of ${\bf A}$ as well as the set of lifted triangulations of ${\bf A}$~\cite[Chap 6.1.3, page 284]{TriangBook}. 
Let ${\cal T}_1$ and ${\cal T}_2$ be two triangulations of ${\bf A}$.  
We say ${\cal T}_1$ is less than ${\cal T}_2$, denoted as ${\cal T}_1 \le_2 {\cal T}_2$ if the characteristic sections $g_{\omega, {\cal T}_1}$ lies vertically below $g_{\omega, {\cal T}_2}$.  The same, we say that 
the lifted triangulation ${\cal T}_1^{\omega}$ {\it lies lower than} ${\cal T}_2^{\omega}$, denoted as ${\cal T}_1^{\omega} \le_2 {\cal T}_2^{\omega}$.  

With the concept of lifted triangulations,
let us re-look at what a directed flip (up-flip or down-flip) in a triangulation ${\cal T}$ of ${\bf A}$ changes the characteristic sections in $\mathbb{R}^3$. 
Let ${\cal T}_1$ be a triangulation of ${\bf A}$.  An up-flip in ${\cal T}_1$ is equivalent to exchange the lower faces by upper faces of a tetrahedron $t_{\bf a'b'c'd'}$ in $\textrm{conv}({\bf A}^{\omega})$, where $\textrm{conv}\{{\bf a, b, c, d}\}$ is the support of this flip. As a result, the lifted triangulation ${\cal T}_1^{\omega}$ is ``deformed" into the lifted triangulation ${\cal T}_2^{\omega}$ which lies vertically above ${\cal T}_1^{\omega}$. 

This fact gives a 3d view of a monotone sequence of directed flips which transforms the regular into the farthest point regular triangulation of $({\bf A}, \omega)$. 
It ``deforms" the ${\cal R}^{\omega}$ gradually into the ${\cal F}^{\omega}$ by passing through a sequence of lifted triangulations of ${\bf A}$,
\[
 {\cal S} := \{{\cal R}^{\omega} = {\cal T}_0^{\omega}, {\cal T}_1^{\omega} \ldots, {\cal T}_m^{\omega} = {\cal F}^{\omega}\}.
\]
Every two adjacent lifted triangulations ${\cal T}_1^{\omega}$ and ${\cal T}_2^{\omega}$ in this sequence is connected by an up-flip in the planar triangulation.  This sequence of lifted triangulations have the following properties:

\begin{itemize}
\item ${\cal S}$ is ordered, i.e., ${\cal T}_i^{\omega} \le_2 {\cal T}_{i+1}^{\omega}, \; \forall i = 0, \ldots, m-1$.
\item Let $\textrm{Vol}({\cal R}^{\omega}, {\cal T}_i^{\omega})$ denote the volume between ${\cal R}^{\omega}$ and ${\cal T}_i^{\omega}$. Then
\[
   \textrm{Vol}({\cal R}^{\omega}, {\cal T}_{i+1}^{\omega}) - \textrm{Vol}({\cal R}^{\omega}, {\cal T}_{i}^{\omega}) = t_i,
\]
where $t_i$ is a tetrahedron in $\textrm{conv}({\bf A}^{\omega})$. The faces of $t_i$ are exactly the triangles involved in the up-flip which transforms ${\cal T}_i$ into ${\cal T}_{i+1}$. This shows that the set of lifted triangulation ${\cal S}$ corresponds to a set of tetrahedra in $\textrm{conv}({\bf A}^{\omega})$,
\[
    T := \{t_0, t_1, \ldots, t_{m-1}\}.
\]

\end{itemize}

With the above definitions and facts, we can prove our theorem. 

\subsection{The proof}

We prove the above theorem by showing that the set of tetrahedra related to a monotone sequence of directed flips fulfils all of the three properties of a triangulation, which are the (i) Closure Property, (ii) Intersection Property, and (iii) Union Property. 

\begin{proof}
Let $V$, $E$, $F$ be the set of all vertices, edges, and triangles of ${\cal S}$.  
Let ${\cal T}_{\cal S} = \{V, E, F, T\}$. We want to show that ${\cal T}_{\cal S}$ fulfils all the three properties of a triangulation, i.e., it is a triangulation of ${\bf A}^{\omega}$. 

First, every $t_i, i = 0, \ldots, m-1$, connects two lifted triangulations ${\cal T}_{i}^{\omega}$ and ${\cal T}_{i+1}^{\omega}$. 
The fact that the sequence of lifted triangulations sweeps through the volume of $\textrm{conv}({\bf A}^{\omega})$ implies that there is no hole left in the volume. Therefore $\cup T = \textrm{conv}({\bf A}^{\omega})$, i.e., the Union Property (iii) holds. 

Second, it is obvious that the faces of every tetrahedron $t_i$ are in $F$, and the edges of $t_i$ are in $E$, and the vertices of $t_i$ are in $V$.  Moreover, $E$ and $V$ are the subsets of edges and vertices of faces of $F$. 
It is sufficient to show that every face in $F$ belongs to at least one tetrahedron $t_i \in T$. 
By the general position assumption there is no identical triangle in ${\cal R}$ and ${\cal F}$ of $({\bf A}, \omega)$. 
Then every face $f \in F$ is either the lower or the upper face of at least one tetrahedron in $T$. 
These facts together ensure that the Closure Property (i) holds for ${\cal T}_{\cal S}$. 

Last, the Intersection Property (ii) holds in ${\cal T}_{\cal S}$ by the fact that every two lifted triangulations are either disjoint or share at their common triangles in $F$. 

This proves that ${\cal T}_{\cal S}$ is a triangulation of $\textrm{conv}({\bf A}^{\omega})$. 
\end{proof}

Although the above theorem consider the biggest convex polytope whose boundaries are the ${\cal R}^{\omega}$ and ${\cal F}^{\omega}$ of ${\bf A}^{\omega}$. Our proof shows that it holds for non-convex polyhedra as well, as long as the boundary of a polyhedron can be 
represented by two planar triangulations which are canonical projections of their lifted triangulations, respectively.  

\begin{cor}~\label{cor:triang}
Let ${\cal T}_u$ and ${\cal T}_v$ be two triangulations of ${\bf A}$ in the plane. Assume the lifted triangulation ${\cal T}_u^{\omega}$ lies strictly lower than ${\cal T}_v^{\omega}$. Let $P_{uv}$ be the 3d polyhedron whose boundary is ${\cal T}_u^{\omega} \cup {\cal T}_v^{\omega}$, then any monotone sequence of directed flips which transform ${\cal T}_u$ into ${\cal T}_v$ corresponds to a triangulation of $P_{uv}$.
\end{cor}

Based on this fact, we present a triangulation algorithm for a special class of 3d non-convex polyhedra in Section~\ref{sec:triang3d}. 

\section{From Triangulations of Polyhedra to Monotone Sequences of Directed Flips}
\label{sec:triang-to-monotone}

In this section, we consider the reverse of the Theorem~\ref{thm:monotone-to-triang}, which is the following {\bf question}: {\it given a triangulation of a 3d polyhedron, does it corresponds to a monotone sequence of directed flips?}

So far, it is only proven for a special case when ${\bf A}$ is the vertex set of a cyclic polytope~\cite[Section 6.1.5]{TriangBook}. The essential part of its proof is to show that between any two triangulations of ${\bf A}$ in $\mathbb{R}^d$, 
there exists a triangulation of the polyhedron in $\mathbb{R}^{d+1}$ 
bounded by these two lifted triangulations without new vertex. 
It is achieved by an explicit (very technical) construction of such triangulation in $\mathbb{R}^{d+1}$. 
However, there are difficulties to directly generalise this result to an arbitrary point set in $\mathbb{R}^2$. 
The main difficulty is that there are 3d non-convex polyhedron which can not be tetrahedralised without additional vertices, such as the Sch\"onhardt polyhedron~\cite{Schonhardt1928}. 

Let ${\bf A}$ be a finite point set in $\mathbb{R}^2$, and $\omega$ be a heigh function defined on the vertices of ${\bf A}$. Let ${\cal T}_u$ and ${\cal T}_v$ be two triangulations of ${\bf A}$.  Without loss of generality, assume the lifted triangulation ${\cal T}_u^{\omega}$ lies vertically below ${\cal T}_v^{\omega}$, and assume that the characteristic sections of ${\cal T}_u^{\omega}$ and ${\cal T}_v^{\omega}$ only intersect at their boundary. Let $P_{uv}$ be the 3d polyhedron whose lower faces are  ${\cal T}_u^{\omega}$ and upper faces are ${\cal T}_v^{\omega}$. 
Assume a triangulation ${\cal T}_{uv}$ of ${P}_{uv}$ is given.  
Moreover, the vertex set of ${\cal T}_{uv}$ is exactly the vertex set of ${P}_{uv}$, i.e., ${\cal T}_{uv}$ contains no additional vertex.  
We describe an algorithm to find a monotone sequence of directed flips between ${\cal T}_u$ and ${\cal T}_v$. 

\subsection{A tetrahedra-driven flipping algorithm}

\paragraph{Notations} We will simultaneously work in two and three dimensions. 
We will relate a vertex in $\mathbb{R}^3$ and its canonical projection in a plane at the same time.  
Since the vertices in ${\cal T}_{uv}$ are in one-to-one correspondence with the vertices in ${\cal T}_u$ and  ${\cal T}_v$, to simplify the notations, we do not distinguish them between $\mathbb{R}^2$ and $\mathbb{R}^3$. For example, we denote $t_{\bf abcd}$ be a tetrahedron in ${\cal T}_{uv}$ with vertices ${\bf a, b, c, d} \in \mathbb{R}^3$. At the same time, we also denote $f_{\bf abc}$ a triangle in ${\cal T}_u$ and ${\cal T}_v$ with vertices ${\bf a, b, c} \in \mathbb{R}^2$.  

\paragraph{Removability and flippability}
We assume that the vertex set ${\bf A}$ is in {\it general position} in the plane which means, no three vertices of ${\bf A}$ lie on a common line.  So the only flips we need are: 2-2, 1-3, and 3-1 flip. 

The idea of our algorithm is to incrementally remove all tetrahedra of ${\cal T}_{uv}$, one at a time. 
Every tetrahedron in ${\cal T}_{uv}$ corresponds a directed flip, which changes its lower and upper faces (or vice versa). 
However, not every tetrahedron is ``removable" at any time. 
Let $t_{\bf abcd}$ be a tetrahedron in ${\cal T}_{uv}$. Let ${\cal T}$ be a triangulation produced by the corresponding sequence of directed flips.  We say that a face $f_{\bf abc}$ of $t_{\bf abcd}$ is {\it exposed} in ${\cal T}$ if $f_{\bf abc}$ is a triangle in ${\cal T}$. 
We say that a tetrahedron $t_{\bf abcd} \in {\cal T}_{uv}$ is {\it removable} (from ${\cal T}_{uv}$) if it in one of the following three cases.  


\begin{itemize}
\item[(1)] One face of $t_{\bf abcd}$ is exposed in ${\cal T}$. Let the face be $f_{\bf abc}$. 
If ${\bf d}$ lies in the interior of the triangle $f_{\bf abc}$, then $t_{\bf abcd}$ is {\it removable}. It corresponds to a 1-3 flip which inserts the vertex ${\bf d}$ into ${\cal T}$, see Figure~\ref{fig:flip2d_flippable} (1).

\item[(2)] Two faces of $t_{\bf abcd}$ are exposed in ${\cal T}$. Let the two faces be $f_{\bf abc}$ and $f_{\bf abd}$, respectively. If the vertices ${\bf a,b,c,d}$ form a strictly convex quadrilateral in ${\cal T}$, 
then $t_{\bf abcd}$ is {\it removable}. It corresponds to a 2-2 flip which replaces edge $e_{\bf ab}$ by $e_{\bf cd}$ in ${\cal T}$,  see Figure~\ref{fig:flip2d_flippable} (2).

\item[(3)] Three faces of $t_{\bf abcd}$ are exposed in ${\cal T}$. Let the three faces be $f_{\bf abd}$, $f_{\bf bcd}$, and $f_{\bf cad}$, then $t_{\bf abcd}$ is {\it removable}. It corresponds to a 3-1 flip which removes the vertex ${\bf d}$ from ${\cal T}$, see Figure~\ref{fig:flip2d_flippable} (3).
\end{itemize}


\begin{figure}[ht]
  \centering
  \begin{tabular}{c c c}
   \includegraphics[width=0.32\textwidth]{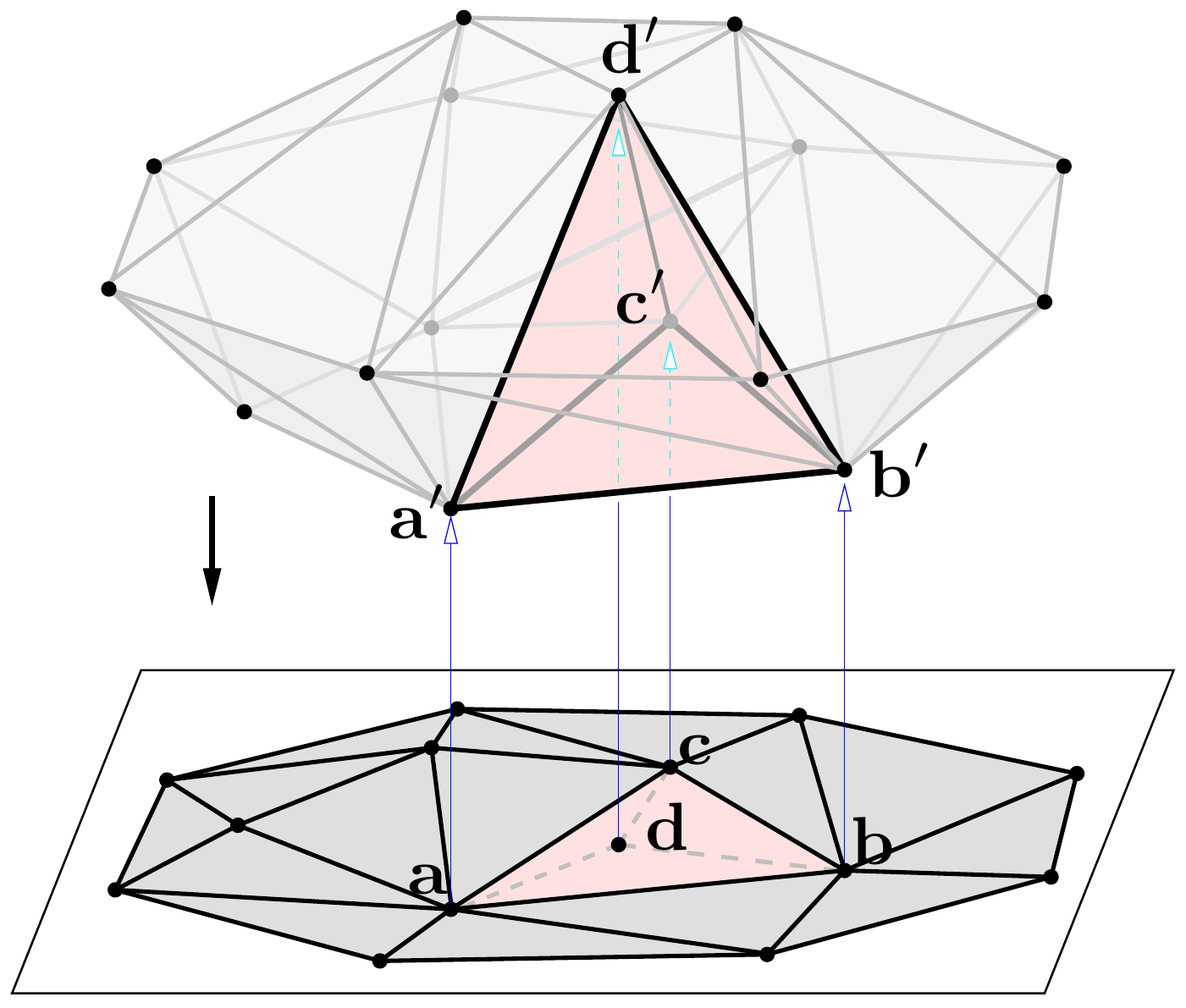} &
   \includegraphics[width=0.32\textwidth]{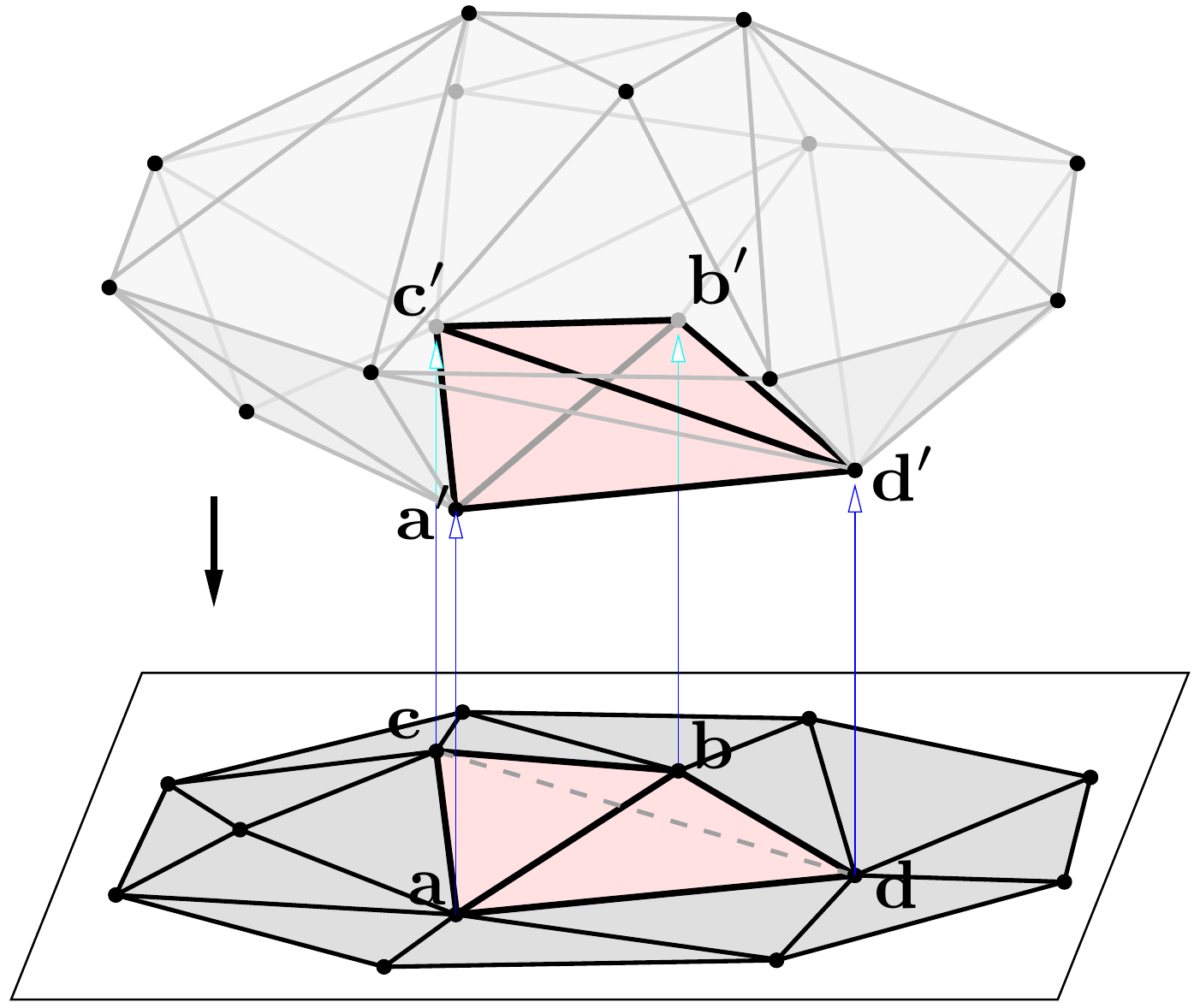} &
   \includegraphics[width=0.32\textwidth]{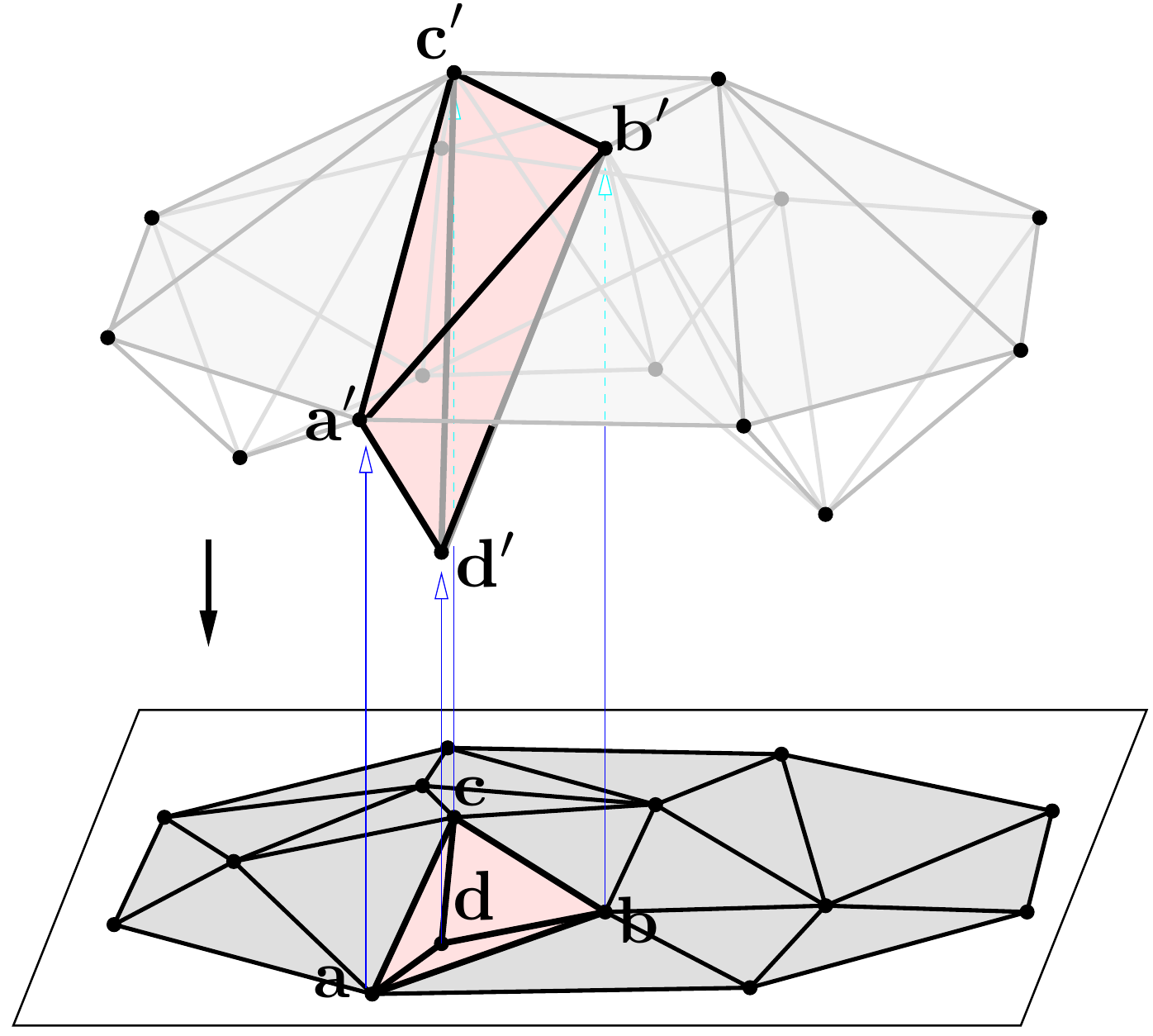} \\ 
   (1) & (2) & (3)
  \end{tabular}
\caption{A removable tetrahedron $t_{\bf abcd}$. (1) The face $f_{\bf abc}$ is exposed in ${\cal T}$, $t_{\bf abcd}$ corresponds to a 1-3 flip. (2) Two faces $f_{\bf abc}$ and $f_{\bf abd}$ are exposed in ${\cal T}$, $t_{\bf abcd}$ corresponds to a 2-2 flip. (3) Three faces $f_{\bf abd}, f_{\bf bcd}, f_{\bf acd}$ are exposed in ${\cal T}$, $t_{\bf abcd}$ corresponds to a 3-1 flip.}
  \label{fig:flip2d_flippable}
\end{figure}

The cases when $t_{\bf abcd}$ is not removable will be discussed in the analysis of the algorithm in the next subsection. 

\paragraph{The algorithm}
The tetrahedron-driven flip algorithm is described in Figure~\ref{fig:flip-algo}. It takes a tetrahedralisation ${\cal T}_{uv}$ and a triangulation ${\cal T}$ (${\cal T}$ is either ${\cal T}_u$ or ${\cal T}_v$). 
This algorithm initialises two lists, ${\cal Q}$, which contains all tetrahedra in ${\cal T}_{uv}$, and ${\cal L}$, which is empty at the beginning. Then it runs into a loop until ${\cal Q}$ is empty. At each iteration, it searches a removable tetrahedron $t_{\bf abcd} \in {\cal Q}$. If such a tetrahedron exists, the triangulation ${\cal T}$ is updated by performing the corresponding flip induced by $t_{\bf abcd}$. 
${\cal Q}$ is reduced by one tetrahedron and it is added to ${\cal L}$. 

\begin{figure}[ht]
  \centering
  \begin{tabular}{l l}
  {\textsc Input}: & ${\cal T}_{uv}$ and ${\cal T}_u$;\\
  {\textsc Output}: & ${\cal T}_v$ and ${\cal L}$ (a sorted array of tetrahedra);
  \end{tabular}\\
  \begin{tabular}{l l}
  1 & Let ${\cal T} := {\cal T}_u$ , ${\cal Q} := \{t_{\bf abcd} \; | t_{\bf abcd} \in {\cal T}_{uv}\}$, ${\cal L} := \emptyset$; \\
  2 & {\bf while} ${\cal Q} \neq \emptyset$ {\bf do}\\
  3 & $\quad$ {\bf if} a removable $t_{\bf abcd}$ exists {\bf then}\\ 
  4 & $\quad \quad$ Update ${\cal Q} := {\cal Q} \setminus \{t_{\bf abcd}\}$; \\
  5 & $\quad \quad$ Perform the corresponding flip in ${\cal T}$;\\
  6 & $\quad \quad$ Update ${\cal L}: = {\cal L} \cup \{t_{\bf abcd}\}$;\\
  7 & $\quad$ {\bf endif}\\
  8 & {\bf endwhile}\\
  \end{tabular}
  \caption{A tetrahedralisation-driven flip algorithm.}
  \label{fig:flip-algo}
\end{figure}

If this algorithm terminates, it outputs a list ${\cal L}$ of sorted tetrahedra of ${\cal T}_{uv}$. Since each tetrahedron in ${\cal L}$ corresponds to a directed flip,  it is also a monotone sequence of directed flips which transforms ${\cal T}_u$ into ${\cal T}_v$.

\subsection{Analysis of termination}


The crucial question in this algorithm is: does there exist a removable tetrahedron in ${\cal T}_{uv}$ (line 3)? If all tetrahedra are not removable, then it gets stuck.

We will use the {\it in-front/behind} relation~\cite{Edelsbrunner90acy} to relate two tetrahedra in $\mathbb{R}^3$. 
Without loss of generality, we assume that the lifted triangulation ${\cal T}_u^{\omega}$ lies vertically below ${\cal T}_v^{\omega}$. 
We use ${\cal T}_{u}$ as the input, so ${\cal T}_v$ is the target (output).  By placing our viewpoint at ${\bf x} = (0, 0, -\infty)$, i.e., we're looking from bottom to top. 
We say that a tetrahedron $t_{\bf pqrs}$ is {\it in-front} of a tetrahedron $t_{\bf abcd}$, denoted as $t_{\bf pqrs} \prec t_{\bf abcd}$, if there exists a ray starting from ${\bf x}$ that intersects first $t_{\bf pqrs}$ then $t_{\bf abcd}$. It is the same to say that $t_{\bf abcd}$ is {\it behind} $t_{\bf pqrs}$, denoted as $t_{\bf abcd} \succ t_{\bf pqrs}$.  

Our algorithm initialises the planar triangulation ${\cal T} := {\cal T}_u$.   
At each time, this algorithm searches a removable tetrahedron $t_{\bf abcd} \in {\cal T}_{uv}$. 
The search starts from an arbitrary tetrahedron $t_{\bf abcd} \in {\cal T}_{uv}$ such that it has at least one exposed face $f_{\bf abc} \in {\cal T}$. If $t_{\bf abcd}$ is removable, then the algorithm continues. 
Assume $t_{\bf abcd}$ is not removable. We show that it is always possible to find another tetrahedron $t_{\bf pqrs} \in {\cal T}_{uv}$ such that it is in front of $t_{\bf abcd}$, i.e., $t_{\bf pqrs} \prec t_{\bf abcd}$. There are following cases to be considered.


\begin{itemize}

\item[{(i)}] Two faces of $t_{\bf abcd}$ are exposed in ${\cal T}$. 
Without loss of generality, let the two faces are $f_{\bf abc}$ and $f_{\bf abd}$ which share the common edge $e_{\bf ab} \in {\cal T}$. Since $t_{\bf abcd}$ is not removable, then ${\bf a}, {\bf b}, {\bf c}, {\bf d}$ form a non-convex quadrilateral in the plane. Let ${\bf a}$ be the non-convex vertex of this quadrilateral. 
Then the face $f_{\bf cda}$ of $t_{\bf abcd}$ is missing in ${\cal T}$, see Figure~\ref{fig:flip2d_unflippable} (1).  
Then there must exist another triangle $f_{\bf pqr} \in {\cal T}$, such that $f_{\bf pqr} \not\in \{f_{\bf abc}, f_{\bf abd}\}$ and $f_{\bf pqr} \cap f_{\bf cda} \neq \emptyset$. 
Hence the tetrahedron $t_{\bf pqrs} \prec t_{\bf abcd}$. 

\begin{figure}
  \centering
  \begin{tabular}{cc}
  \includegraphics[width=0.45\textwidth]{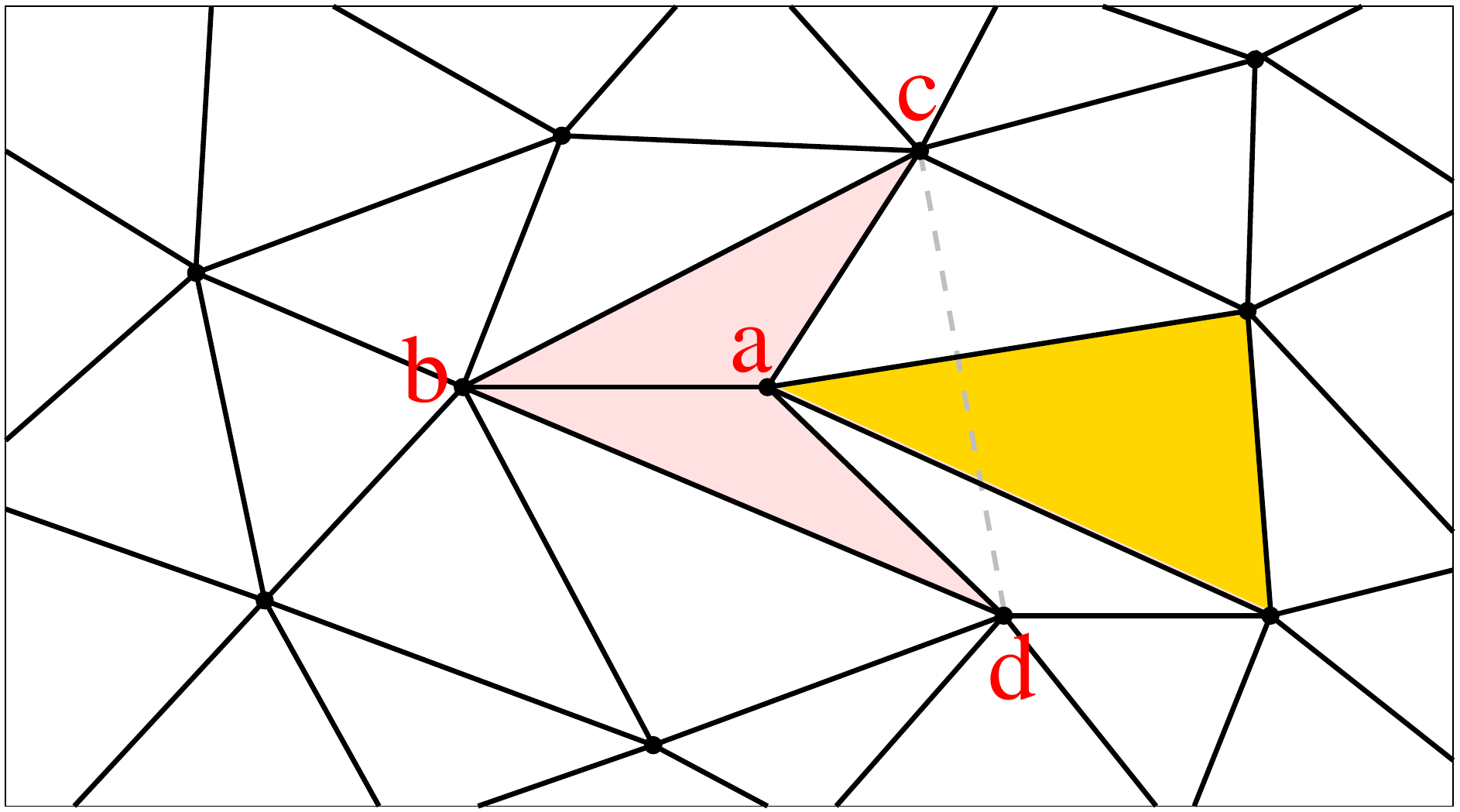} &
  \includegraphics[width=0.45\textwidth]{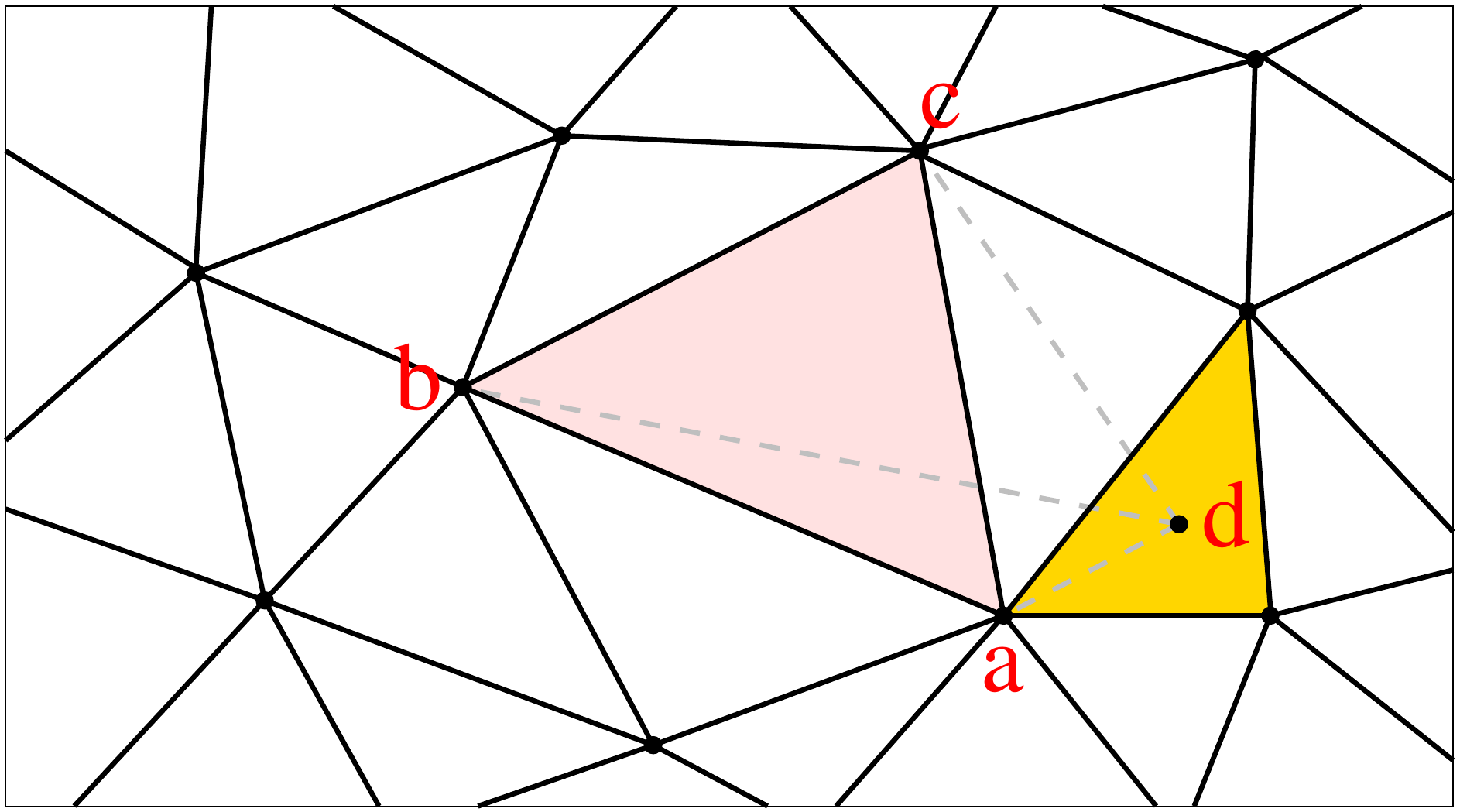}\\
  (1) & (2)\\
  \includegraphics[width=0.45\textwidth]{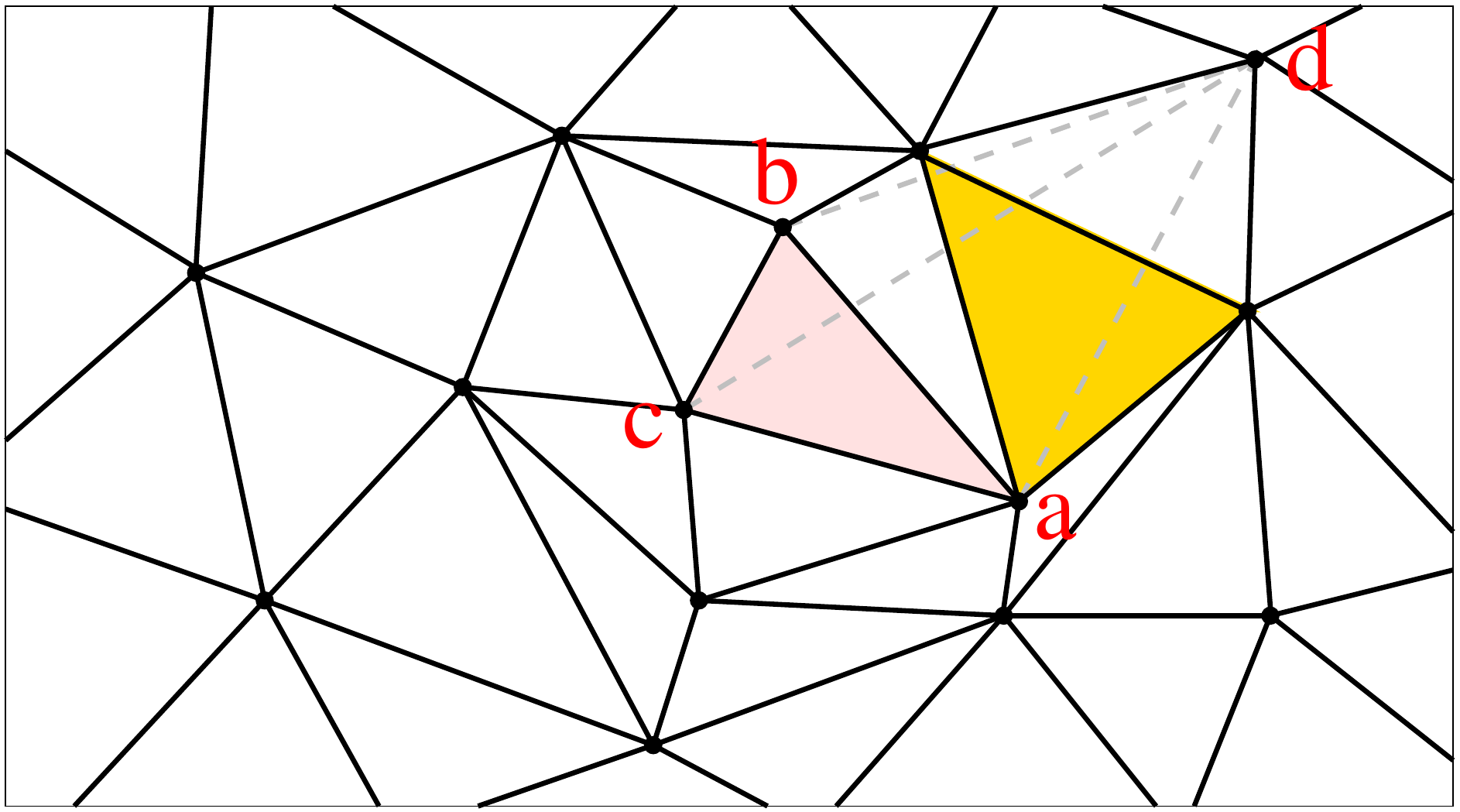} &
  \includegraphics[width=0.45\textwidth]{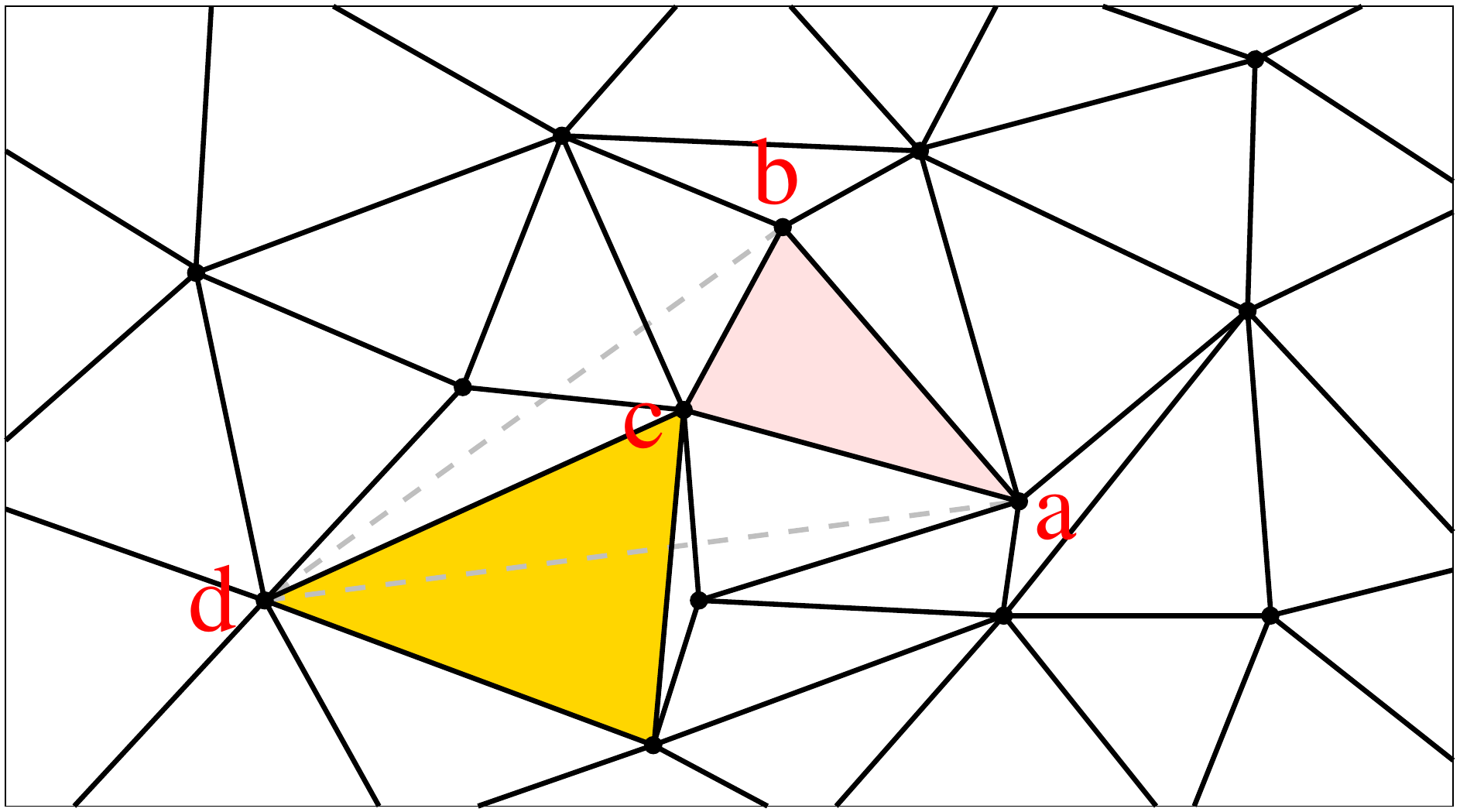}\\
  (3) & (4)\\
  \end{tabular}
  \caption{All cases when a tetrahedron $t_{\bf abcd}$ is not removable. 
  In (1) neither a 2-2 flip nor a 3-1 flip is possible on edge $e_{\bf ab}$. 
  In (2), a 1-3 flip is not possible since the new vertex ${\bf d}$ does not inside the triangle $f_{\bf abc}$.  
  In (3) nd (4) a 1-3 flip is not possible since the vertex ${\bf d}$ already exists in ${\cal T}$. 
  }
  \label{fig:flip2d_unflippable}
\end{figure}

\item[{(ii)}] Only the face $f_{\bf abc}$ of $t_{\bf abcd}$ is exposed in ${\cal T}$, and ${\bf d} \not\in {\cal T}$. Since $t_{\bf abcd}$ is not removable, then ${\bf d}$ lies in a triangle $f_{\bf pqr} \in {\cal T}$, such that $f_{\bf pqr} \neq f_{\bf abc}$, see Figure~\ref{fig:flip2d_unflippable} (2). Hence the tetrahedron $t_{\bf pqrs} \prec t_{\bf abcd}$. 

\item[{(iii)}] Only the face $f_{\bf abc}$ of $t_{\bf abcd}$ is exposed in ${\cal T}$, and ${\bf d} \in {\cal T}$. Then we can show that the face $f_{\bf abd} \not\in {\cal T}$.  We distinguish two cases:
  \begin{itemize}
  \item[(a)] The edge $e_{\bf cd}$ intersects $e_{\bf ab}$, see Figure~\ref{fig:flip2d_unflippable} (3).
    If $f_{\bf abd}$ exists, it must belong to another tetrahedron $t_{\bf abde}$, where ${\bf e} \neq {\bf c}$. Then either the edge $e_{\bf ae}$ or $e_{\bf be}$ must go through the interior of $t_{\bf abcd}$ -- a contradiction that ${\cal T}_{uv}$ is a tetrahedralisation. Then there must exist a face $f_{\bf pqr} \in {\cal T}$, and $f_{\bf pqr} \cap f_{\bf cda} \neq \emptyset$, see Figure~\ref{fig:flip2d_unflippable} (3).

  \item[(b)] The edge $e_{\bf cd}$ does not intersect $e_{\bf ab}$, see Figure~\ref{fig:flip2d_unflippable} (4).
  Then the vertex ${\bf c}$ blocks the appearance of $f_{\bf abd}$. Therefore any face except $f_{\bf abc}$ in ${\cal T}$ which contains ${\bf c}$ as a vertex can be the face  $f_{\bf pqr} \in {\cal T}$, and $f_{\bf pqr} \cap f_{\bf cda} \neq \emptyset$, see Figure~\ref{fig:flip2d_unflippable} (4).
  \end{itemize}

In both cases, there exists a tetrahedron $t_{\bf pqrs} \prec t_{\bf abcd}$. 

\end{itemize}

From the above analysis, we see that starting from an arbitrary tetrahedron $t_{\bf abcd} \in {\cal T}_{uv}$, it is either a removable tetrahedron, or there exists another tetrahedron $t_{\bf pqrs} \in {\cal T}_{uv}$ and  $t_{\bf pqrs} \prec t_{\bf abcd}$. This process can be repeated.  This generates a sequence of tetrahedra, which are:
\[
t_{\bf abcd} = t_{{\bf a}_0{\bf b}_0{\bf c}_0{\bf d}_0} \succ t_{{\bf a}_1{\bf b}_1{\bf c}_1{\bf d}_1} \succ \ldots \succ t_{{\bf a}_m{\bf b}_m{\bf c}_m{\bf d}_m} = t_{\bf pqrs}
\]
such that each tetrahedron $t_{{\bf a}_i{\bf b}_i{\bf c}_i{\bf d}_i} \in {\cal T}_{uv}$, $i = 0, \ldots, m$, is not removable. 
Therefore this algorithm will terminate if the above sequence does not form a cycle, i.e, $t_{{\bf a}_m{\bf b}_m{\bf c}_m{\bf d}_m} = t_{{\bf a}_0{\bf b}_0{\bf c}_0{\bf d}_0}$. 

\begin{theorem}~\label{thm:triang-to-monotone}
This algorithm terminates as long as ${\cal T}_{uv}$ contains no cycle of tetrahedra from the viewpoint ${\bf x} = (0,0,-\infty)$ (or ${\bf x} = (0,0,+\infty)$). 
\end{theorem}

By the Acyclic Theorem~\cite{Edelsbrunner90acy}, a regular triangulation contains no cycle of faces from any viewpoint. An immediate result is that this algorithm will terminate if ${\cal T}_{uv}$ is a regular tetrahedralisation.  

Our analysis shows that only cycles of tetrahedra from the given viewpoint ${\bf x}$ will cause  problem.  
More generally, we can consider a parallel projection of ${\cal T}_{uv}$ through any line direction $L$ (a one-dimensional affine subspace) in $\mathbb{R}^3$. We will also get two planar triangulations ${\cal T}_{u, L}$ and ${\cal T}_{v, L}$ in two planes, one lies in front of ${\cal T}_{uv}$, and one lies behind ${\cal T}_{uv}$ from a viewpoint on $L$.  
We can re-define the height function $\omega$ on vertices of ${\cal T}_{u, L}$ and ${\cal T}_{v, L}$ with respect to $L$ as following: 
Let ${\bf p} = (p_x, p_y, p_z) \in \mathbb{R}^3$ be a vertex in  ${\cal T}_{u, L}$ or ${\cal T}_{v, L}$, we project the vector $(p_x, p_y, p_z)$ onto the line $L$ and let the $\omega({\bf p})$ be the length of the projected vector. 
Our algorithm will transform ${\cal T}_{u, L}$ into ${\cal T}_{v, L}$ or vice versa by a sequence of monotone flips as long as there is no cycle in ${\cal T}_{uv}$ viewed by any viewpoint on $L$, an example is shown in Figure~\ref{fig:cyc3_proj_both}.

We thus have the following theorem about the termination of this algorithm.

\begin{cor}~\label{thm:triang-to-monotone-viewpoint}
This algorithm terminates as long as the 3d trianagulation ${\cal T}_{uv}$ of ${\bf A}^{\omega}$ in $\mathbb{R}^3$ contains no cycle of tetrahedra with respect to a viewpoint on the direction of the height function $\omega$. 
\end{cor}

We also see another nice fact of the 3d triangulations produced by monotone sequence of directed flips. 

\begin{cor}~\label{thm:triang-to-monotone-acyclic}
Let ${\cal T}$ be a 3d triangulation of a polyhedron $P$ in $\mathbb{R}^3$. 
If ${\cal T}$ corresponds to a monotone sequence of directed flips between two triangulations of a point set ${\bf A}$ in $\mathbb{R}^2$ with a height function $\omega : {\bf A} \to \mathbb{R}$, then the set of tetrahedra of ${\cal T}$ is acyclic with respect to a viewpoint on the direction of the height function $\omega$. 
\end{cor}

Note that the 3d triangulation ${\cal T}$ in the above Corollary is not necessarily a regular triangulation.

\subsection{A cycle consists of three tetrahedra}

Our analysis of this algorithm shows, ${\cal T}_{uv}$ may be non-regular and the algorithm still terminates as long as there is no cycle from the chosen viewpoint which produces ${\cal T}_u$ and ${\cal T}_v$. It is thus interesting to study and understand such (non-regular) tetrahedralisations containing cycles (from some viewpoints).  

The simplest non-regular tetrahedralisation contains $7$ vertices. 
It is first appear in~\cite{Joe1989}. 
It is the simplest tetrahedralisation which will cause the 3d Lawson's flip algorithm fails~\cite{Joe1989}.
A construction of this example is given in the book~\cite[Chap 3, Example 3.6.15, page 139]{TriangBook}. 

\begin{figure}
  \centering
  \includegraphics[width=0.7\textwidth]{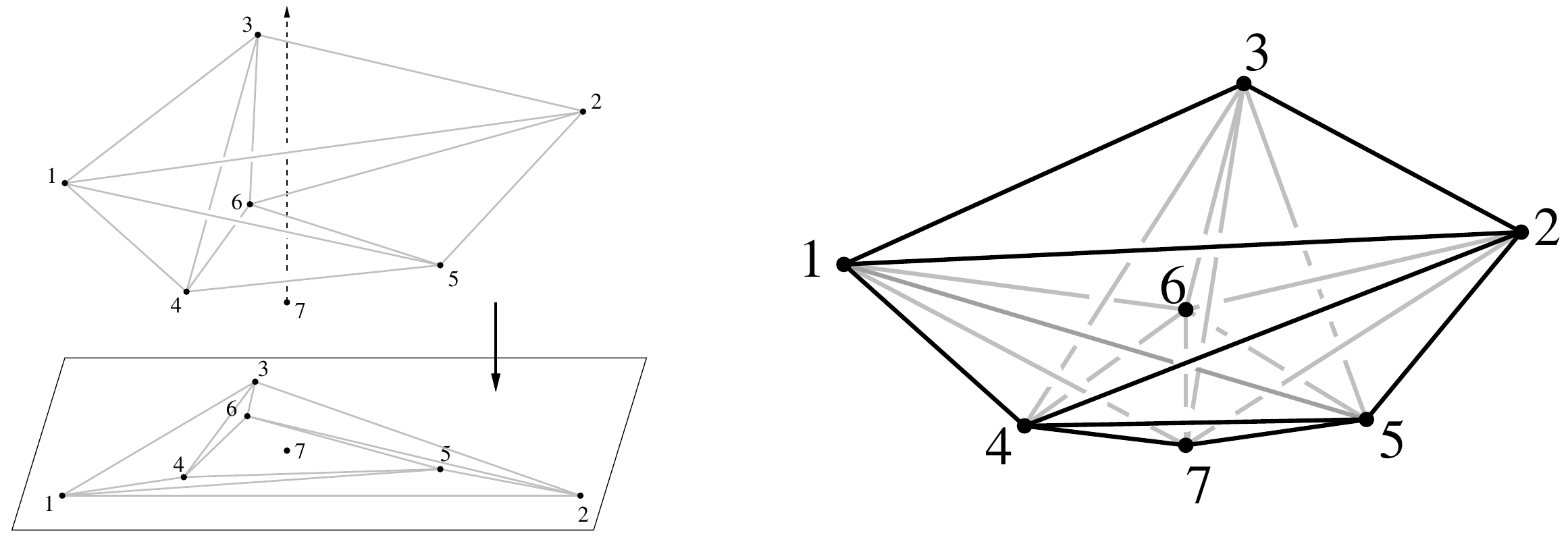}
\caption{Left: The point set of this tetrahedralisation contains the point set of a Sch\"onhardt polyhedron with an additional vertex lie outside this polyhedron and is visible by all vertices of it from the interior of this polyhedron. Right: the tetrahedralisation of the left point set consists of $10$ tetrahedra. }
\label{fig:cyc3_all_tets} 
\end{figure}

The basic fact is that this tetrahedralisation contains the boundary triangulation of the Sch\"onhardt polyhedron. 
It is necessary to contain an additional vertex which does not belong to the Sch\"onhardt polyhedron. 
In particular, this additional vertex is visible by all vertices of the Sch\"onhardt polyhedron through the bottom face, see Figure~\ref{fig:cyc3_all_tets} Left. This property ensures that a tetrahedralisation containing this Sch\"onhardt polyhedron exists. 
Figure~\ref{fig:cyc3_all_tets} Right shows the the tetrahedralisation, denoted as ${\cal T}_{uv}$. There are $10$ tetrahedra in ${\cal T}_{uv}$. Three  tetrahedra are highlighted by different color. 
They form a cycle if the viewpoint is chosen along the $z$ axis. 
This shows that ${\cal T}_{uv}$ is a non-regular triangulation.  
However, these tetrahedra do not contain cycle if the viewpoint is chosen along the $x$ axis, see Figure~\ref{fig:cyc3_proj_both}.  

\begin{figure}[ht]
  \centering
  \includegraphics[width=0.7\textwidth]{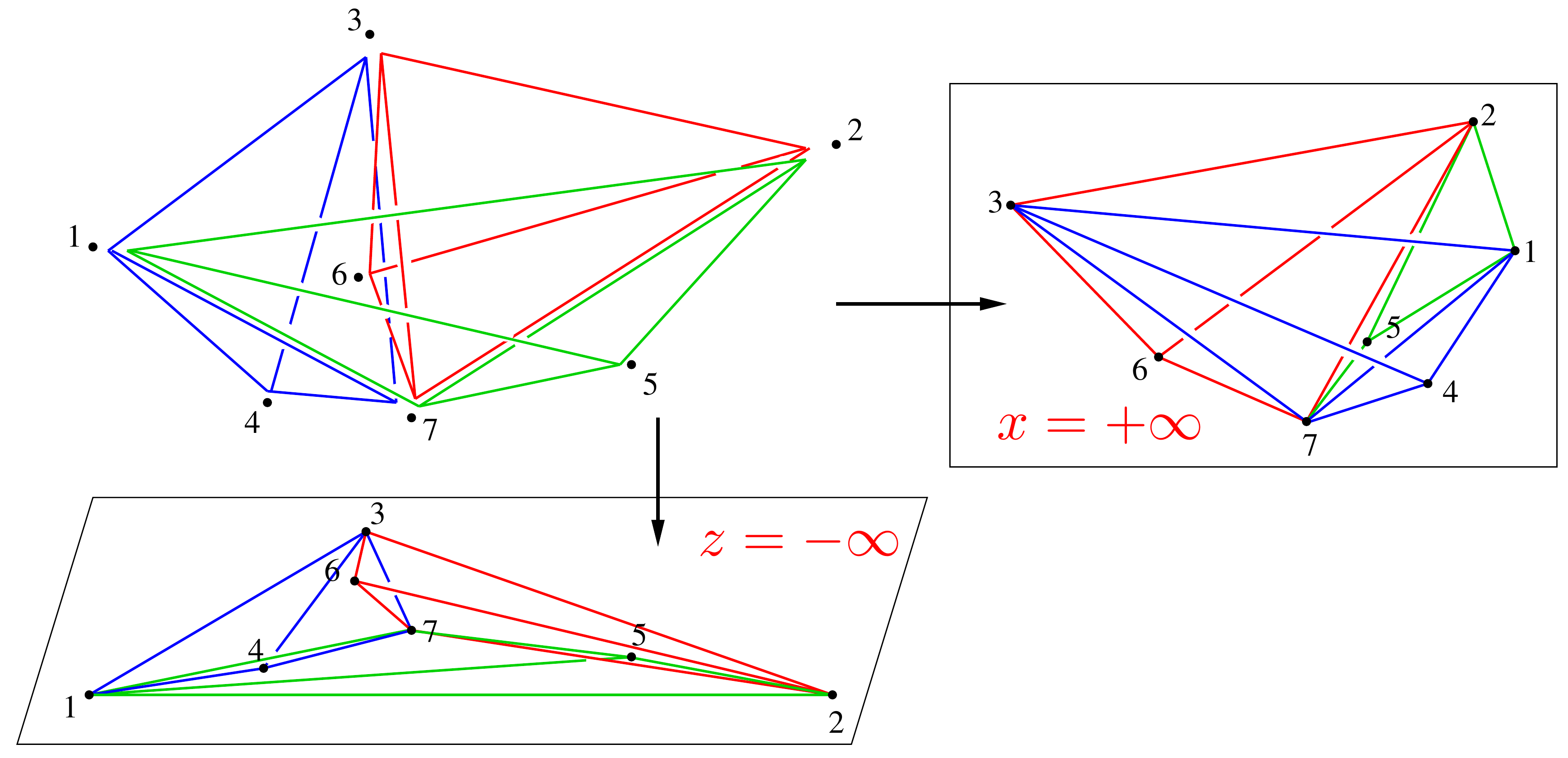}
\caption{Two orthogonal projections of the three tetrahedra in Figure~\ref{fig:cyc3_all_tets}. In particular, the projection along the $z$-direction contains a cycle $t_{1374} \prec_z t_{1275} \prec_z t_{2376} \prec_z t_{1374}$, while along the $x$-direction does not.}
\label{fig:cyc3_proj_both} 
\end{figure}

Let ${\cal T}_{u,z}$ and ${\cal T}_{v,z}$ be the two traingulaitons obtained by the orthogonal projections of ${\cal T}_{uv}$ along the $z$-axis. Likewise, let ${\cal T}_{u,x}$ and ${\cal T}_{v,x}$ be the two traingulaitons obtained by the orthogonal projections of ${\cal T}_{uv}$ along the $x$-axis, see Figure~\ref{fig:cyc3_flip_both}.  Then the flip algorithm will fail to transform ${\cal T}_{u,z}$ to ${\cal T}_{v,z}$, or vice versa. However, it will succeed to transform ${\cal T}_{u,x}$ to ${\cal T}_{v,x}$, or vice versa.

\begin{figure}[ht]
  \centering
  \includegraphics[width=0.7\textwidth]{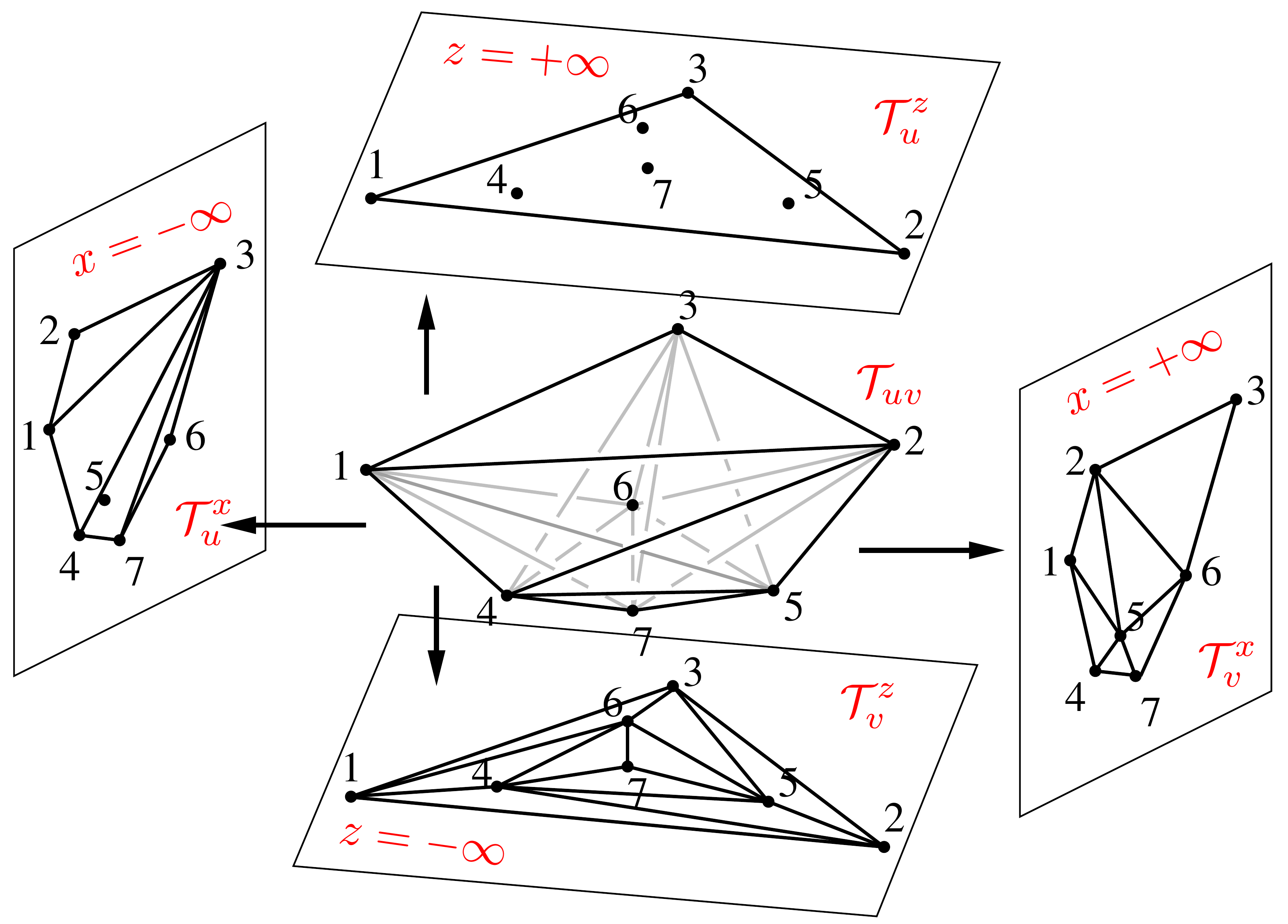}
\caption{Triangulations obtained from different directions of the parallel projection of the same tetrahedralisation ${\cal T}_{uv}$. }
\label{fig:cyc3_flip_both} 
\end{figure}

From this example, we see that the existence of a viewpoint which ``sees" no cycle of tetrahedra is a one of a key properties of a 3d triangulation of a point set. It is necessary to study the following question.

\begin{question}~\label{question:acyclic_viewpoint}
If ${\cal T}_{uv}$ contains no interior point, is it always possible to find a viewpoint which produces two triangulations ${\cal T}_{u}$ and ${\cal T}_{v}$ such that the algorithm will succeed?
\end{question}

\section{Properties of the Directed Flip Graph}
\label{sec:poset-structure}

Since triangulations of ${\bf A}$ can be ordered by the directed flips,  
we obtain a directed flip graph (poset) defined on the set of all triangulations of ${\bf A}$.  
As we have already shown in the motivation example in Section~\ref{sec:graph-poset}, it encodes many useful information of this point set. 
It is important to understand the structure of this poset. 
However, as we could observe that the size of this poset grows exponentially as the number of vertices grows. 
Figure~\ref{fig:poset_7_points} shows partially of this poset for the point set of $7$ vertices for the point set shown in Figure~\ref{fig:cyc3_all_tets}. 

\begin{figure}[ht]
  \centering
  \includegraphics[width=1.0\textwidth]{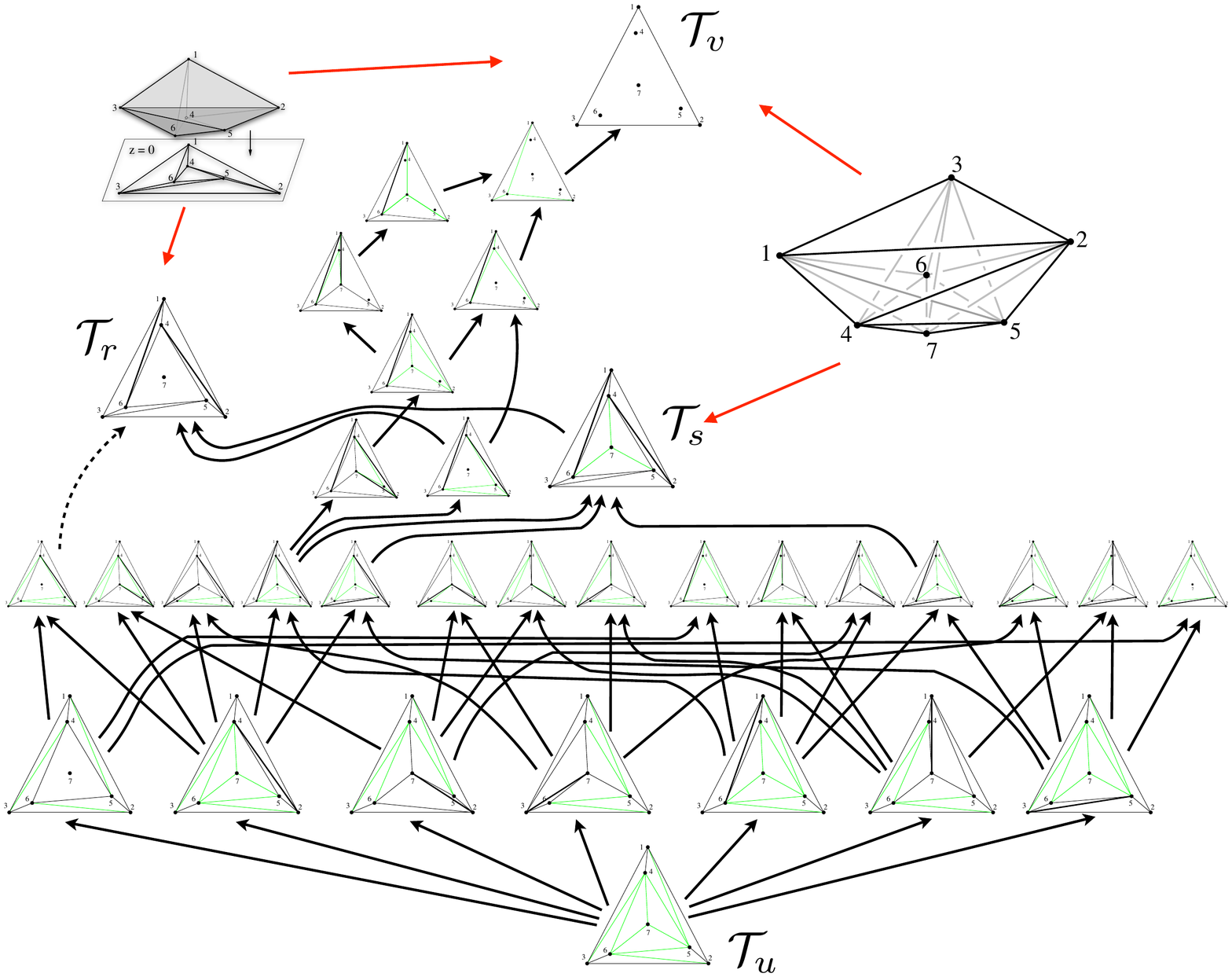}
\caption{A subset of the directed flip graph of a set of $7$ vertices.}
\label{fig:poset_7_points} 
\end{figure}

\subsection{General properties of this poset}
\label{sec:poset_general_properties}

In this section, some (obvious) properties of this poset are proven. 
Let ${\bf A}$ be a finite point set in $\mathbb{R}^2$ and let $\omega: {\bf A} \to  \mathbb{R}$ be a height function.  Let ${\cal G}({\bf A}, \omega)$ be the directed flip graph (poset) of all triangulations of ${\bf A}$ ordered by one of the directed flips.  

\paragraph{Definitions} A {\it path} between two nodes in this poset represents a monotone sequence of directed flips between two triangulations of ${\bf A}$.  
Call a path {\it maximum} if it starts from the regular and ends at the farthest point regular triangulation of $({\bf A}, \omega)$, or vice versa. 
A node of this poset is called an {\it external node} if it can only be either the start or end of a path, otherwise, it is an {\it internal node}.

\begin{theorem}~\label{thm:poset-general}
\item[(1)] Every path between two nodes in this poset corresponds to a triangulation of the 3d polyhedron whose boundary is the union of the two triangulations. 

\item[(2)] 
Every maximum path in this poset corresponds to a triangulation of ${\bf A}^{\omega}$.
Moreover, this triangulation may be non-regular.

\item[(3)] There exists non-regular triangulation of ${\bf A}^{\omega}$ which does not correspond to any maximum path in this poset.  

\item[(4)] This poset is in general not bounded. 

\item[(5)] Except the regular and the farthest point regular triangulations of $({\bf A}, \omega)$, all other external nodes in this poset are non-regular triangulations of ${\bf A}$.

\item[(6)] Internal nodes of of this poset may be non-regular triangulations of ${\bf A}$.


\end{theorem}




\begin{proof}
Case (1) is exactly proven by the Corollary~\ref{cor:triang}. Note that the 3d polyhedron is not necessarily convex. 

The first part of (2) is proven by Theorem~\ref{thm:monotone-to-triang}. We prove the second part of (2) by showing an example. Figure~\ref{fig:cyc3_flip_both} shows such a 3d triangulation. First of all, there exists a monotone sequence of flips from ${\cal T}_{u,{\bf x}}$ to ${\cal T}_{v,{\bf x}}$ which produces this triangulation. Hence a poset produced by viewpoints on the $x$-axis will contain this triangulation.  
However, the fact that there exists no monotone sequence of flips from ${\cal T}_{u,{\bf z}}$ to ${\cal T}_{v,{\bf z}}$ shows that it contains a cycle from viewpoints on the $z$-axis. By the Acyclic Theorem~\cite{Edelsbrunner90acy} it must be non-regular.  
The same example also proves (3). 

Next we show that (2) and (3) together imply (4), i.e., this poset is not bounded. 
This can be proven by showing two examples (therefore there are  at least two reasons). 
Both examples are shown in Figure~\ref{fig:poset_7_points}.  
The first example is that there exists no path of up-flips (as well as down-flips) from the non-regular triangulation ${\cal T}_r$ to the further point regular triangulation ${\cal T}_v$.  
In this case, there exists no 3d triangulation between these two triangulations.  
The second example is that there exists no path of up-flips (as well as down-flips) from the non-regular triangulation ${\cal T}_s$ to the further point regular triangulation ${\cal T}_v$. In this case, although there exists a 3d triangulation between these two triangulations, but there is no monotone path between them.

To prove (5) we use the theory of secondary polytopes. By Theorem~\ref{thm:secondarypolytope},  between every two regular triangulations of ${\bf A}$ there must exist a monotone sequence of directed flips. The negative of this Theorem shows that if there exists no monotone sequence of directed flips between two triangulations of ${\bf A}$, then at least one of the triangulations is non-regular. This shows that all external nodes of this poset, except the regular and farthest-point regular triangulations of $({\bf A}, \omega)$, are non-regular. 

Finally, we construct an example to prove (6). The point set of this example contains $8$ vertices, denoted as ${\bf A}$. The constructed tetrahedralisation ${\cal T}$ has $15$ tetrahedra. They are shown in Figure~\ref{fig:non-regular-sections-triang3d}. It is necessary to have the 8th vertex which makes this example work. Figure~\ref{fig:non-regular-sections-poset} shows one of the maximum chains of the poset of this point set, in particular, it starts from the regular triangulation of ${\bf A}$ to the farthest-point regular triangulation of ${\bf A}$. The sequence of flips are ordered by the tetrahedra as we have ordered in the list shown in the Figure~\ref{fig:non-regular-sections-triang3d}. 

One can show that in this chain, the internal nodes (3), (4), ... , (9) are non-regular triangulations. We use the negative of the Acyclic Theorem~\cite{Edelsbrunner90acy} to show that in each of these triangulations, there exists a cycle of triangles with respect to the in-front/behind relation viewed from the vertex $7$. Figure~\ref{fig:non-regular-sections-cycles} shows two of these cycles. 
\end{proof}

\begin{figure}[ht]
  \centering
  \includegraphics[width=0.8\textwidth]{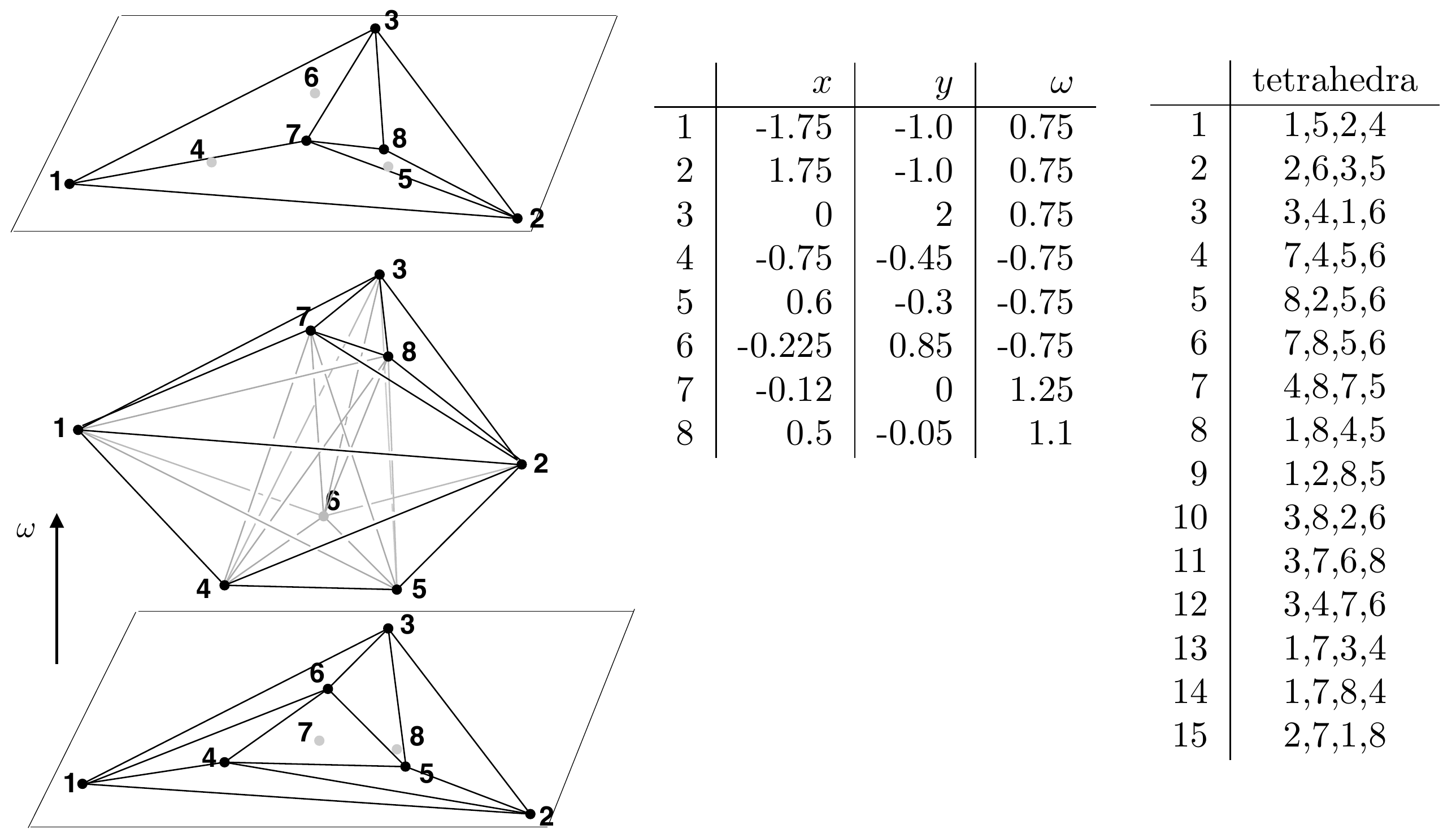}
\caption{A point set ${\bf A}$ of $8$ points in $\mathbb{R}^2$ with a height function $\omega$.  Left: The regular and the farthest-point regular triangulation of $({\bf A}, \omega)$, and a tetrahedralisation ${\cal T}$ of the convex hull of ${\bf A}^{\omega}$.    Right: Tables of coordinates of the $8$ vertices of ${\bf A}$ and the definition of $\omega$, and the list of the $15$ tetrahedra of ${\cal T}$.}
\label{fig:non-regular-sections-triang3d} 
\end{figure}

\begin{figure}[ht]
  \centering
  \includegraphics[width=1.0\textwidth]{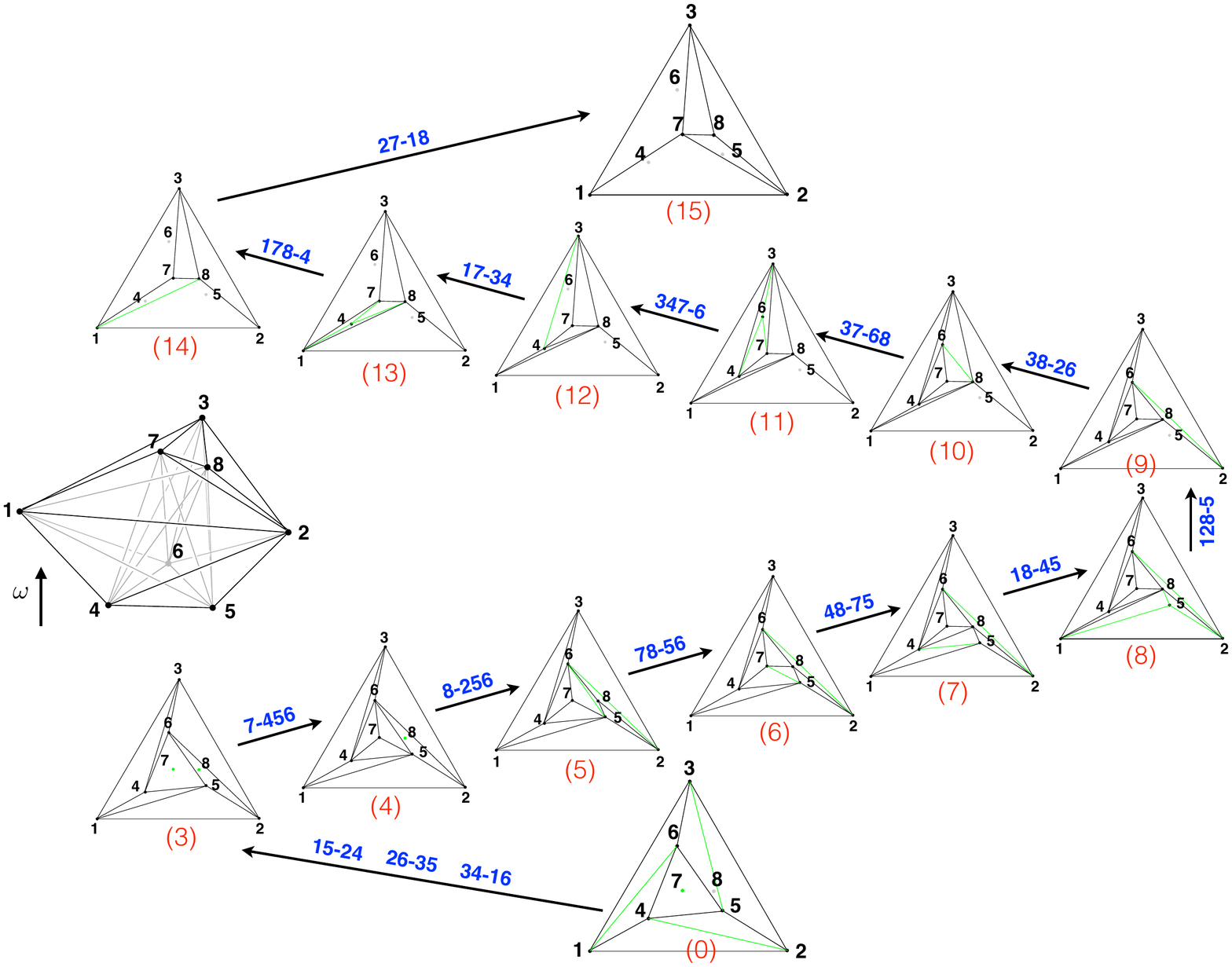}
\caption{A maximum chain of flips from the regular to the farthest-point regular triangulation of $({\bf A}, \omega)$. The triangulations from (3) to (9) are non-regular triangulations of ${\bf A}$.}
\label{fig:non-regular-sections-poset} 
\end{figure}

\begin{figure}[ht]
  \centering
  \includegraphics[width=0.6\textwidth]{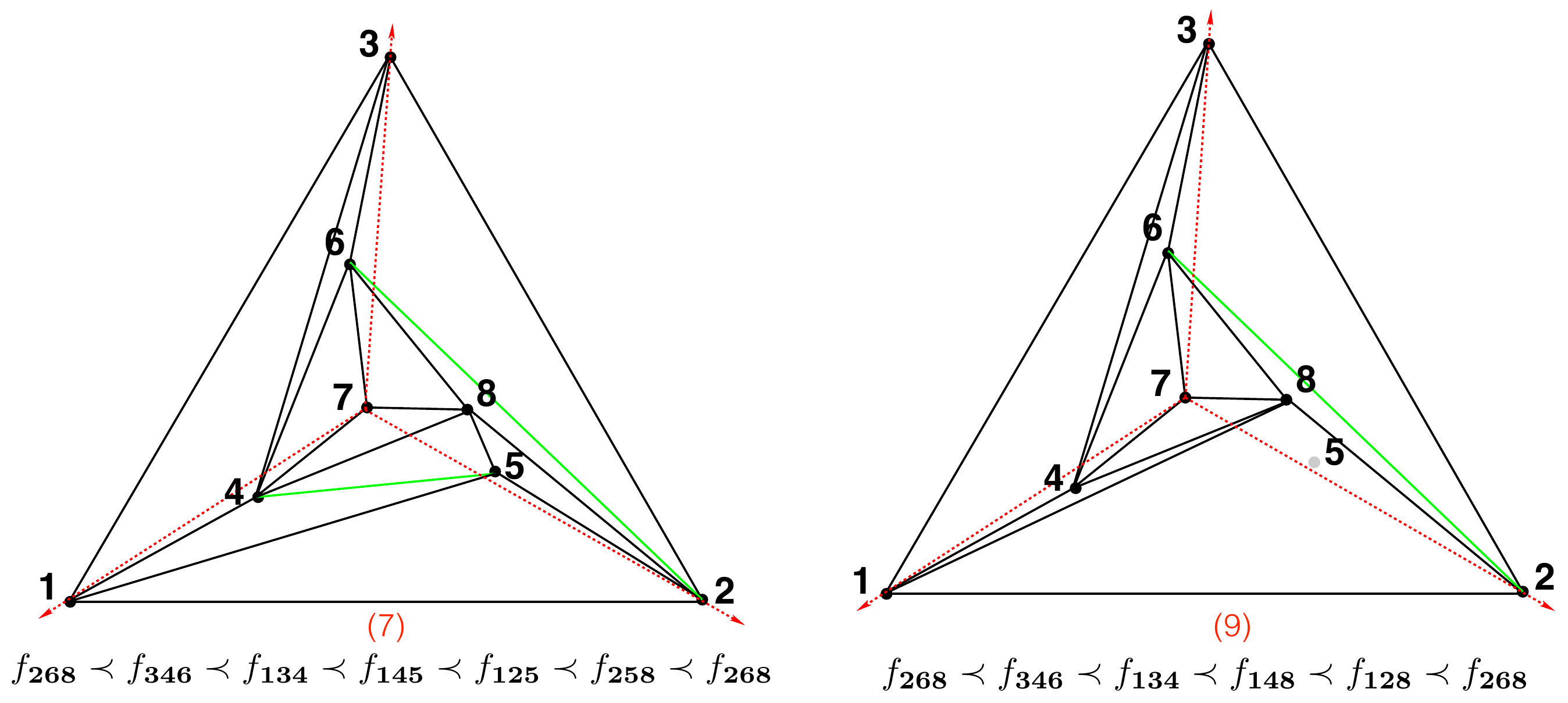}
\caption{Two non-regular trianglations, (7) and (9). The cycles of triangles in each of them are shown. They are all viewed from the vertex $7$. }
\label{fig:non-regular-sections-cycles} 
\end{figure}

\subsection{When $\omega$ is convex or concave}
\label{sec:poset_convex_concave_omega}

In this section we prove a special property of this poset when the height function $\omega$ is a convex or concave.  From this property we can derive the exact conditions for when the Lawson's flip algorithm will terminate. 

\begin{theorem}~\label{thm:poset-convex-heights}
  \item[(1)] If $\omega$ is convex, then the poset of down-flips has a unique minimum, which is the regular triangulation of $({\bf A}, \omega)$. However, it might have no unique maximum. 
  \item[(2)] If $\omega$ is concave, then the poset of up-flips has a unique maximum, which is the farthest point regular triangulation of $({\bf A}, \omega)$. However, it might have no unique minimum. 
\end{theorem}

\begin{proof}
Note that (2) is just the reverse of (1).  We only need to prove (1). 

$\omega$ is convex implies that all lifted points of ${\bf A}$ lies on the lower envelop of the convex hull of ${\bf A}^{\omega}$ in $\mathbb{R}^3$. Then the vertex set of the regular triangulation of $({\bf A}, \omega)$ is equal to ${\bf A}$. The vertex set of any other triangulation of ${\bf A}$ must be either a subset of or equal to ${\bf A}$. This means that all monotone sequences of down-flips are either 2-2 flips (edge flips) or 1-3 flips (vertex insertions). There is no need of 3-1 flips (vertex deletions).  

In the following, we show that between any triangulation ${\cal T}$ of ${\bf A}$ and the regular triangulation of $({\bf A}, \omega)$, denoted as ${\cal R}$, there always exists a sequence of down-flips. 
The proof is essentially the same as proving the termination of Lawson's flip algorithm for producing the Delaunay triangulation of ${\bf A}$. In addition to using the 2-2 flips (edge flips), we need to use 1-3 flip (vertex insertions) as well.    

Assume ${\cal T} \neq {\cal R}$. Call an edge $e_{\bf ab} \in {\cal T}$ {\it locally non-regular with respect to down-flips} if it is shared by two triangles $f_{\bf abc}, f_{\bf abd} \in {\cal T}$, and the lifted point ${\bf d}'$ lies vertically below the plane containing the lifted points ${\bf a', b', c'}$. In other words, the lifted edge $e_{\bf a'b'}$ is locally non-convex between the lifted triangles $f_{\bf a'b'c'}$ and $f_{\bf a'b'd'}$ when viewed from below. Then $e_{\bf ab}$ must admit a 2-2 flip. Otherwise, at least one of the lifted points,  either ${\bf a}'$ or ${\bf b}'$ does not lie on the lower convex hull of ${\bf A}^{\omega}$ which implies that $\omega$ is non-convex -- a contradiction. Applying a 2-2 flip to replace $e_{\bf ab}$ by $e_{\bf cd}$ is the same as attaching a tetrahedron $t_{\bf a'b'c'd'}$ below the edge $e_{\bf a'b'}$. This means that the new edge $e_{\bf cd} \in {\cal T}$ is locally regular.  Whenever there exists a locally non-regular edge in ${\cal T}$ we can flip it and replace it by a locally regular edge of ${\cal T}$, this process is irreversible, since it never creates locally non-regular edges. 

Assume ${\cal T} \neq {\cal R}$ and all edges in ${\cal T}$ are locally regular. Then there must exist a vertex in ${\cal R}$ but not yet in ${\cal T}$. 
Let ${\bf p} \in {\cal R}$ and ${\bf p} \not\in {\cal T}$. There must exist a triangle $f_{\bf abc} \in {\cal T}$ such that ${\bf p} \in f_{\bf abc}$. Assume no three points in ${\bf A}$ are collinear, then $f_{\bf abc}$ and ${\bf p}$ admit a 1-3 flip (vertex insertion) in ${\cal T}$. After ${\bf p}$ is inserted, apply 2-2 flips if there are locally non-regular edges in ${\cal T}$. 

The above process creates no cycles.  This means that every triangulation of ${\bf A}$ can be eventually  transformed into the regular triangulation of $({\bf A}, \omega)$ by a sequence of down-flips. 
This shows that this poset is bounded below, i.e., it has a unique minimum.  

The fact that it might have no unique maximum in (1) can be proven by showing an example such that there exists no monotone sequence of up-flips between a triangulation of ${\bf A}$ and the farthest point regular triangulation of ${\bf A}$. Such examples are shown in~\ref{fig:poset_7_points}. There is no monotone sequence of flips between the triangulation ${\cal T}_r$ and ${\cal T}_v$, and between ${\cal T}_s$ and ${\cal T}_v$.
\end{proof}

From the above theorem, we can state precise conditions when Lawson's flip algorithm will terminate.

\begin{cor}
The Lawson's flip algorithm terminates in one of the following cases.
\begin{itemize}
\item[(1)] $\omega$ is convex and the directed flips are down-flips.
\item[(2)] $\omega$ is concave and the directed flips are up-flips.
\end{itemize}
\end{cor}

From the above theorem, we see that the condition that there is no 3-1 flip (vertex deletion) makes the poset bounded below or above.  
We remark that there are point sets which automatically satisfy this condition, such as the point set of a convex $n$-gon.  In this case, no matter what property $\omega$ has, Lawson's flip algorithm will terminate. 
This provides another proof of the Theorem~\ref{thm:poset-cyclicpoly} of the properties of point sets of cyclic polytopes in $\mathbb{R}^2$.

\subsection{Characterisation of external nodes of the poset}
\label{sec:redundant_vertices}

When $\omega$ is neither convex nor concave, starting from an arbitrary triangulation of ${\bf A}$, a monotone sequence of directed flips (up- or down-flips) might end at a non-regular triangulation of ${\bf A}$ which is an external node of this poset. Besides the non-regular property, we show another property of these external nodes (triangulations). 

\paragraph{Definitions}  We assume a given poset is produced by one of the two directed flips (up- and down-flips). 
We call the minimum and maximum nodes of this poset  {\it lower extreme} and {\it upper extreme} nodes of this poset, respectively.   External nodes which are not extreme nodes are called {\it non-extreme} nodes of this poset.
We further distinguish two types of the non-extreme nodes. We call a non-extreme node {\it upper non-extreme} if there exists no directed flip in this triangulation towards the triangulation of the upper extreme node, and call it {\it lower non-extreme} if there exists no directed flip from this triangulation towards the triangulation of the lower extreme node. Figure~\ref{fig:non-extreme_nodes} Left illustrates these definitions.  

\begin{figure}[ht]
  \centering
  \begin{minipage}{0.6\textwidth}
  \includegraphics[width=1.0\textwidth]{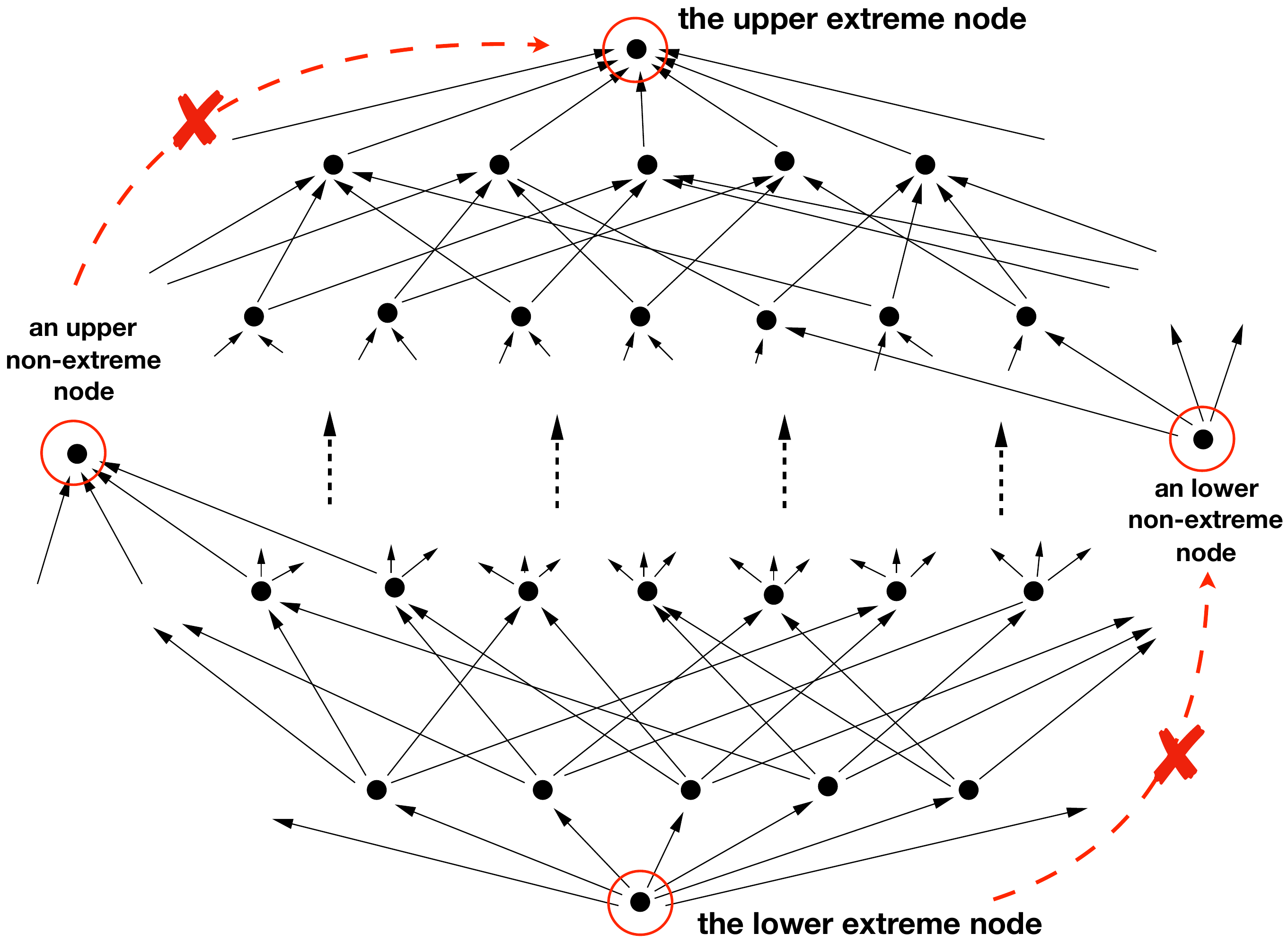}
  \end{minipage}
  \begin{minipage}{0.38\textwidth}
  \includegraphics[width=1.0\textwidth]{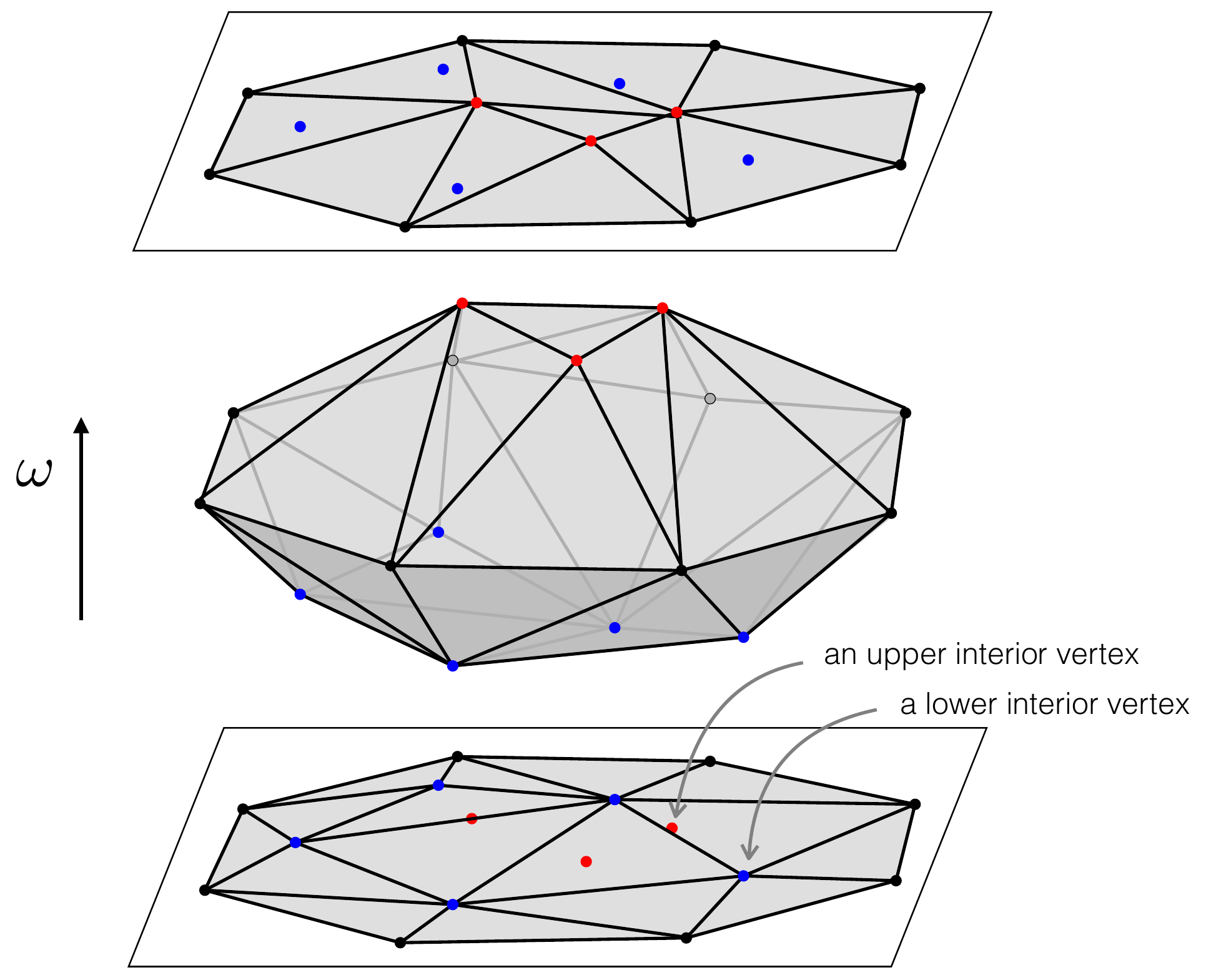}
  \end{minipage}
\caption{Left: Assume $\omega$ is neither convex nor concave. An example of a poset of triangulations of $({\bf A}, \omega)$. Extreme nodes, upper non-extreme nodes, and lower non-extreme nodes are shown.
Right: The lower and upper interior vertices of a point set $({\bf A}, \omega)$.
}
\label{fig:non-extreme_nodes} 
\end{figure}

A vertex of ${\bf A}$ is an {\it interior vertex} if it is not on the convex hull of ${\bf A}$. 
We distinguish two types of the interior vertices of ${\bf A}$. We call an interior vertex ${\bf a} \in {\bf A}$ {\it upper interior} if the lifted point ${\bf a}'$ appears in the triangulation of the upper extreme node. and call it {\it lower interior} if ${\bf a}'$ appears in the triangulation of the lower extreme node, 
see Figure~\ref{fig:non-extreme_nodes} Right for examples.

\paragraph{Redundant interior vertices} By Theorem~\ref{thm:poset-convex-heights}, when the directed flips are up-flips, if a triangulation contains no interior vertices of ${\bf A}$, then it must can be transformed into the farthest point regular triangulation of $({\bf A}, \omega)$ since there is no 3-1 flip (vertex deletion) needed. Hence this triangulation must not be an upper extreme node.
We then can state the following property of all non-extreme nodes of this poset. 

\begin{itemize}
\item[(1)] Any triangulation corresponds to an upper non-extreme node in this poset must contain lower interior vertices of ${\bf A}$.
\item[(2)] Any triangulation corresponds to a lower non-extreme node in this poset must contain upper interior vertices of ${\bf A}$.
\end{itemize}

Any triangulation of a non-extreme node has no directed flip towards one of the two regular triangulations of $({\bf A}, \omega)$. We see the reason of the failure of Lawson's flip algorithm -- there exist interior vertices in a triangulation which cannot be removed by  the corresponding directed flips.  
These vertices, either lower or upper interior vertices, in (1) and (2) are call {\it redundant interior vertices}, see Figure~\ref{fig:redundant-vertices} for examples. \\

In the following, we study the question: why a redundant interior vertex cannot be removed by any monotone sequence of directed flips?

\paragraph{Locally non-regular edges} 
Recall that the definition of ``locally non-regular edges with respect to down flips" was given in Section~\ref{sec:poset_convex_concave_omega} (in the proof of Theorem~\ref{thm:poset-convex-heights}). 
Below we give the definition of ``locally non-regular edges with respect to up flips" which is the opposite of the former. 
Let ${\cal T}$ be a triangulation of ${\bf A}$. 
We say an edge $e_{\bf ab} \in {\cal T}$ is {\it locally non-regular with respect to up-flips} if the lifted vertex ${\bf d}' \in \mathbb{R}^3$ lies vertically above the plane passing through the lifted vertices ${\bf a}', {\bf b}', {\bf c}' \in \mathbb{R}^3$, where $f_{\bf abc}, f_{\bf abd} \in {\cal T}$ are the two triangles sharing at $e_{\bf ab}$. In other words, the lifted edge $e_{\bf a'b'}$ is locally non-convex when viewed from above (the viewpoint is placed at $(0,0,+\infty)$). For examples, the blue edges in Figure~\ref{fig:redundant-vertices} are all locally non-regular with respect to up-flips.  
The two versions of ``locally non-regular edges" together make the clear definition of ``locally non-regular edges with respect to the corresponding directed flips".  

In the rest of this section, we will simply say a ``locally regular" or ``locally non-regular" edge when the directed flips (up-flips or down-flips) are clear from the context. 
  
The following two facts are due to the definition of non-extreme nodes. Recall that the definition of unflippability is given in Section~\ref{subsec:flips}. 

\begin{fact}\label{fact:non-extreme-nodes}
Let ${\cal T}$ be a triangulation of ${\bf A}$ and it corresponds to a non-extreme node of the poset. Then 
\begin{itemize}
\item[(1)] All flippable edges in ${\cal T}$ are locally regular. 
\item[(2)] All locally non-regular edges in ${\cal T}$ are unflippable.  
\end{itemize}
\end{fact}

We first prove an important lemma. 
Let ${\bf a} \in {\cal T}$ be a edundant interior vertex, the following lemma shows two special properties of ${\bf a}$.

\begin{lemma}~\label{lem:cyclic_edge_lemma} 
Let ${\cal T}$ be a triangulation of ${\bf A}$ and it corresponds to a non-extreme node of the poset. 
Let ${\bf a} \in {\bf A}$ be a redundant interior vertex, then we show the following:
\begin{itemize}
\item[(i)] There exist at least two locally non-regular edges at ${\bf a}$.
\item[(ii)] There exists at least one locally non-regular edge, $e_{\bf ap}$, and the unflippability is due to ${\bf p}$, moreover, ${\bf p}$ is a redundant interior vertex of ${\bf A}$. 
\end{itemize}
\end{lemma}

\begin{proof}
Without loss of generality, we assume that the poset is produced by up-flips, and ${\cal T}$ corresponds to an upper non-extreme node of the poset, see Figure~\ref{fig:redundant-vertices}.
Then ${\bf a} \in {\cal T}$ is a lower interior vertex of ${\bf A}$.  
 
To shorten the notation, in the rest of this proof, all ``locally regular" and ``locally non-regular" are with respect to up-flips. \\

We first prove (i).
Let the link vertices of ${\bf a}$ in clockwise direction be ${\bf p_0}, {\bf p_1}, \ldots, {\bf p_m}$. 
Figure~\ref{fig:vertex_remove_2d_case1} shows an example of a possible star of ${\bf a}$. 
Since ${\cal T}$ corresponds to a non-extreme node, then all flippable edge in star of ${\bf a}$ must be locally regular. Only unflippable edges in star of ${\bf a}$ might be locally non-regular. For examples, the edges ${\bf a}{\bf p}_i, {\bf a}{\bf p}_j, {\bf a}{\bf p}_k$ in Figure~\ref{fig:vertex_remove_2d_case1}. 
Without loss of generality, assume the edge $e_{{\bf a}{\bf p}_0}$ is a locally non-regular edge. It is unflippable and  
the unflippability is due to ${\bf a}$. 

\begin{figure}[ht]
  \centering
  \begin{minipage}{0.35\textwidth}
  \centering
  \vspace{2cm}
  \includegraphics[width=1.0\textwidth]{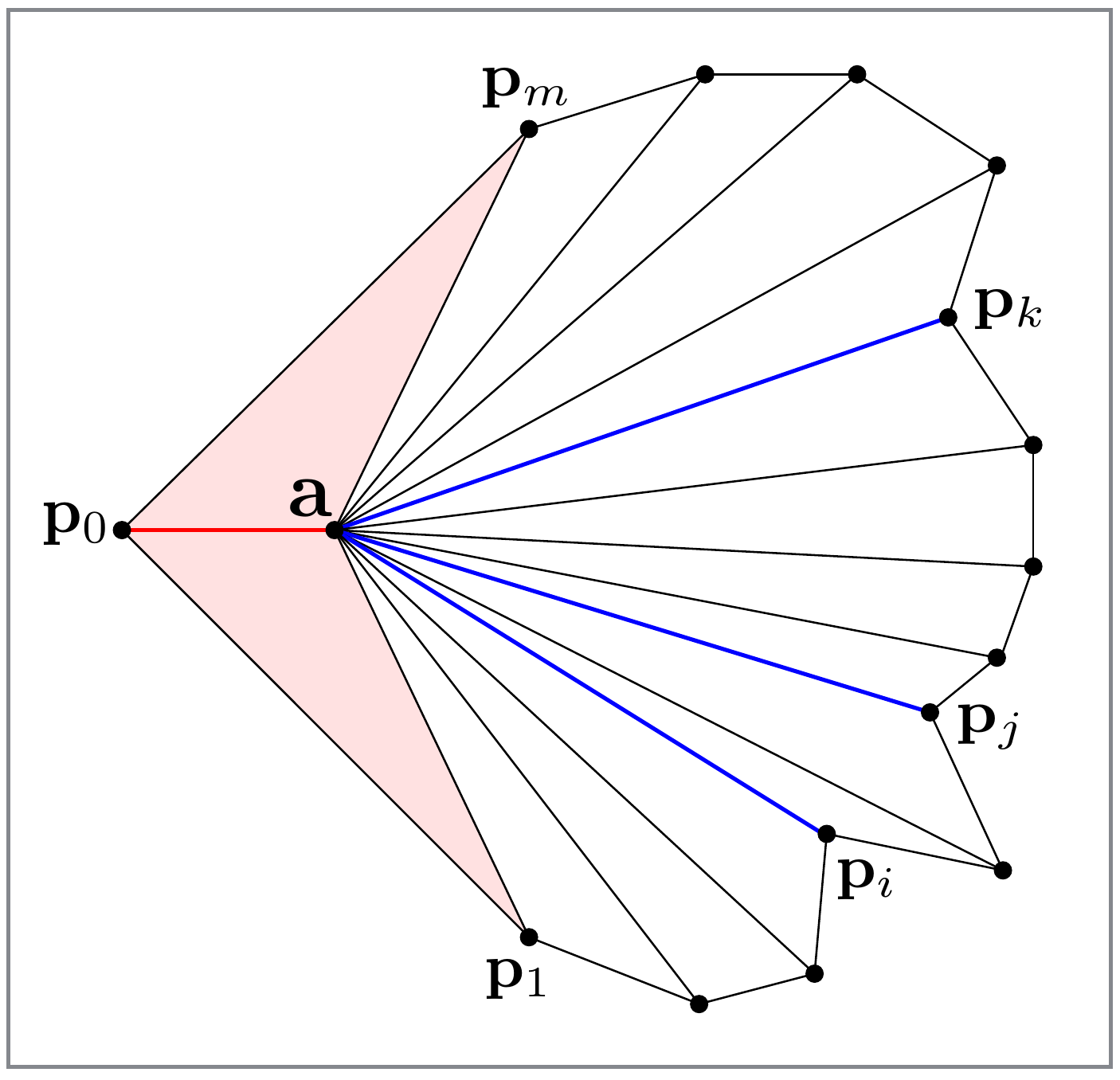}
  \end{minipage}
  \hspace{0.3cm}
  \begin{minipage}{0.6\textwidth}
  \centering
  \includegraphics[width=1.0\textwidth]{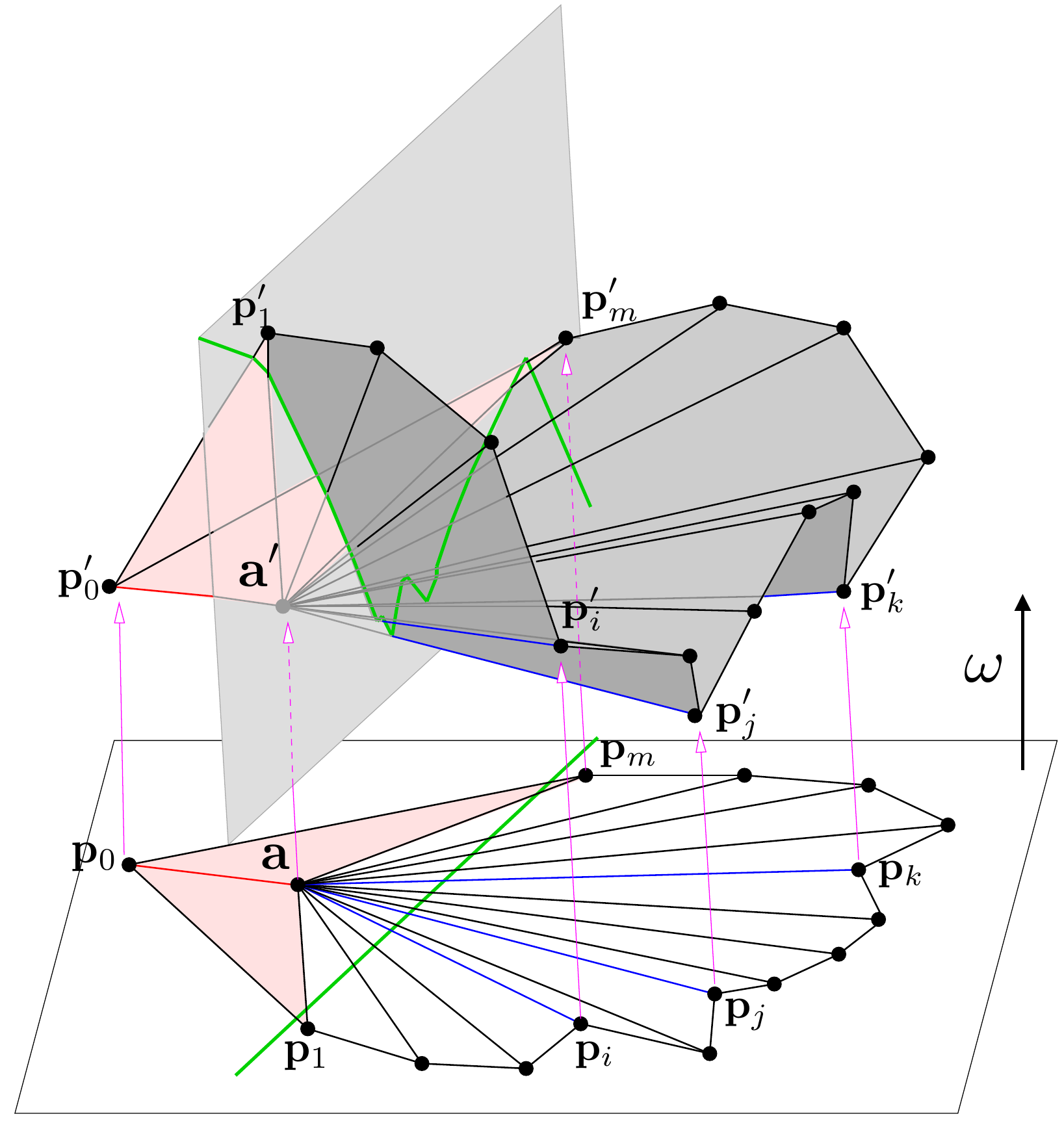}
  \end{minipage}
\caption{Left: The star of a redundant interior vertex ${\bf a} \in {\cal T}$. Since ${\cal T}$ is a triangulation corresponds to a non-extreme node, all flippable edges in this figure must be locally regular.  Unflipables edges (shown in red and blue) might be locally non-regular. 
Right: A vertical plane (shown in grey) cuts the lifted triangles in the star of a redundant interior vertex ${\bf a} \in {\cal T}$. The intersection curve (shown in green) is a piecewise linear line segment in $\mathbb{R}^3$. It must contain at least a vertex which is locally convex. In this example, there are three such vertices on this curve, which corresponds to the intersections between the plane and the lifted edges $e_{{\bf a}'{\bf p}_i'}$, $e_{{\bf a}'{\bf p}_j'}$, and $e_{{\bf a}'{\bf p}_k'}$.
}
\label{fig:vertex_remove_2d_case1} 
\end{figure}


We search another locally non-regular edge from the clockwise direction of $e_{{\bf a}{\bf p}_0}$.
Before we meet this edge, we will get a sequence of edges: $e_{{\bf a}{\bf p}_1}, e_{{\bf a}{\bf p}_2}, \ldots, e_{{\bf a}{\bf p}_m}$, such that all edges are locally regular (with respect to up-flips). 
This means, ${\bf p}_{i+1}'$ lies vertically below the plane passing through the lifted vertices ${\bf a}', {\bf p}_{i}', {\bf p}_{i-1}'$, $i = 1, \ldots, m-1$. 
If we do not find a locally non-regular edge, 
this means that ${\bf p}_{m}'$ lies vertically below the plane passing through the lifted points ${\bf a}', {\bf p}_{0}', {\bf p}_{1}'$, this contradicts to the assumption that the edge $e_{{\bf a}{\bf p}_0}$ is locally non-regular (with respect to up-flips). 
In other words, if we intersect the set of all lifted triangles in the star of ${\bf a}$ in $\mathbb{R}^3$ by a vertical  plane, the intersection curve (which is a piecewise linear line segment) must not be concave (ideal for up-flips). 
It must contain at least a vertex which is locally convex, and the edge in the star of ${\bf a}$ which corresponds to this locally convex intersection point must be locally non-regular (with respect to up-flips). 
Figure~\ref{fig:vertex_remove_2d_case1} Right shows an example.
Therefore, we must find at least one edge $e_{{\bf a}{\bf p}_i} \neq e_{{\bf a}{\bf p}_0}$, $0 < i \le m$, which is locally non-regular (with respect to up-flips) and it is unflippable. 
This proves (i).\\

Before we prove (ii). We show two basic facts. First, there can exist at least two locally non-regular edges whose unflippability are due to ${\bf a}$. The other edge can be either ${\bf a}{\bf p}_1$ or ${\bf a}{\bf p}_m$, i.e., one of the adjacent edges of ${\bf a}{\bf p}_0$. They are shown in Figure~\ref{fig:vertex_remove_2d_case2}. 
Second we show that there are at most two locally non-regular edges at ${\bf a}$. 
This can be proven by basic planar geometry. Without loss of generality, consider the case when ${\bf a}{\bf p}_1$ is locally non-regular, see Figure~\ref{fig:vertex_remove_2d_case2} Left.  Then we have the following inequalities about the internal angles at ${\bf a}$:  $\alpha + \beta > 180^o$, $\beta + \gamma > 180^o$, then $\gamma + \delta < 180^o$ and $\alpha + \delta < 180^o$, which implies there exists no more unflippable edge at ${\bf a}$ whose unflippability can be caused by ${\bf a}$.  

\begin{figure}[ht]
  \centering
  \begin{tabular}{c c}
  \includegraphics[width=0.4\textwidth]{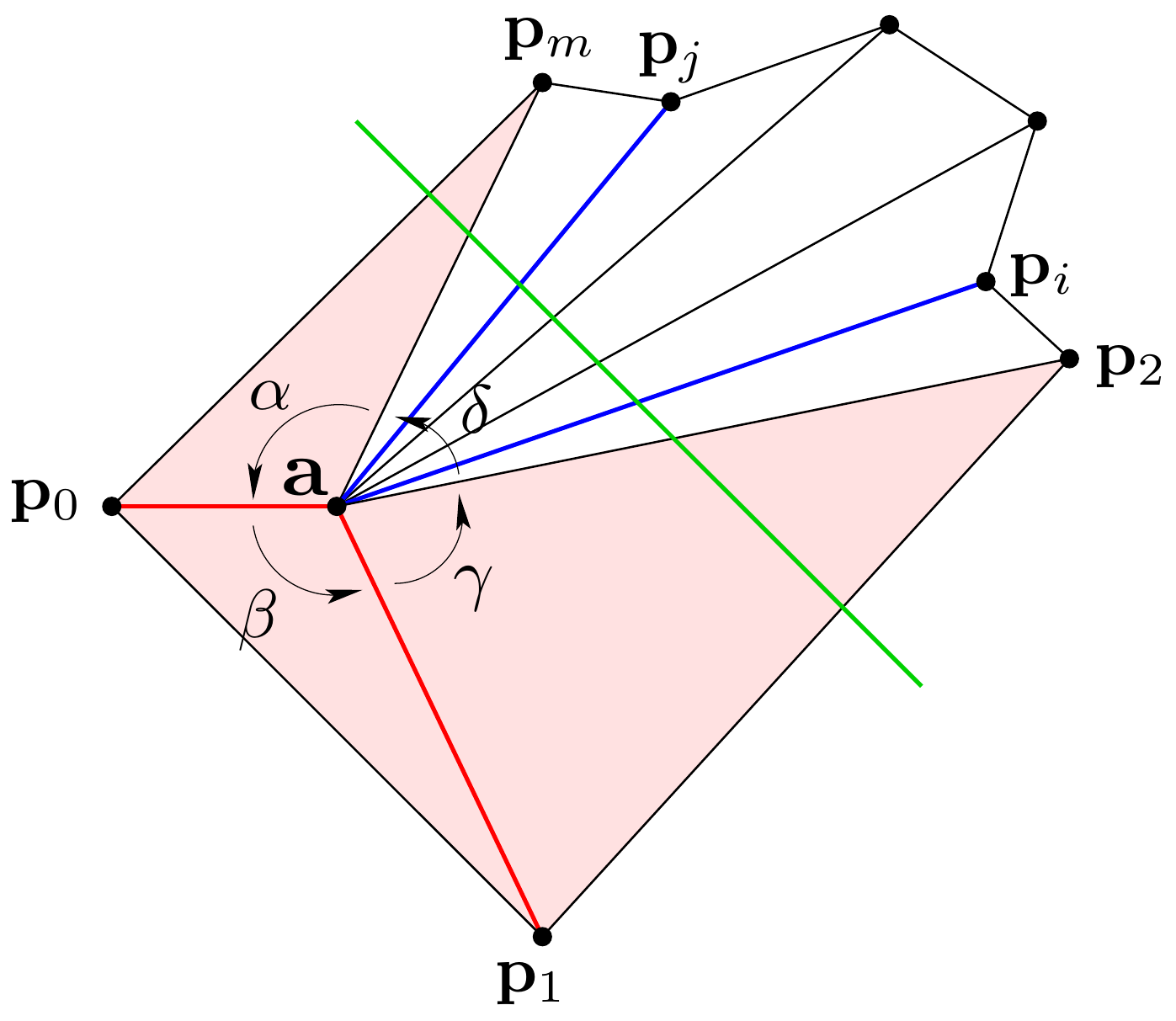} &
  \includegraphics[width=0.4\textwidth]{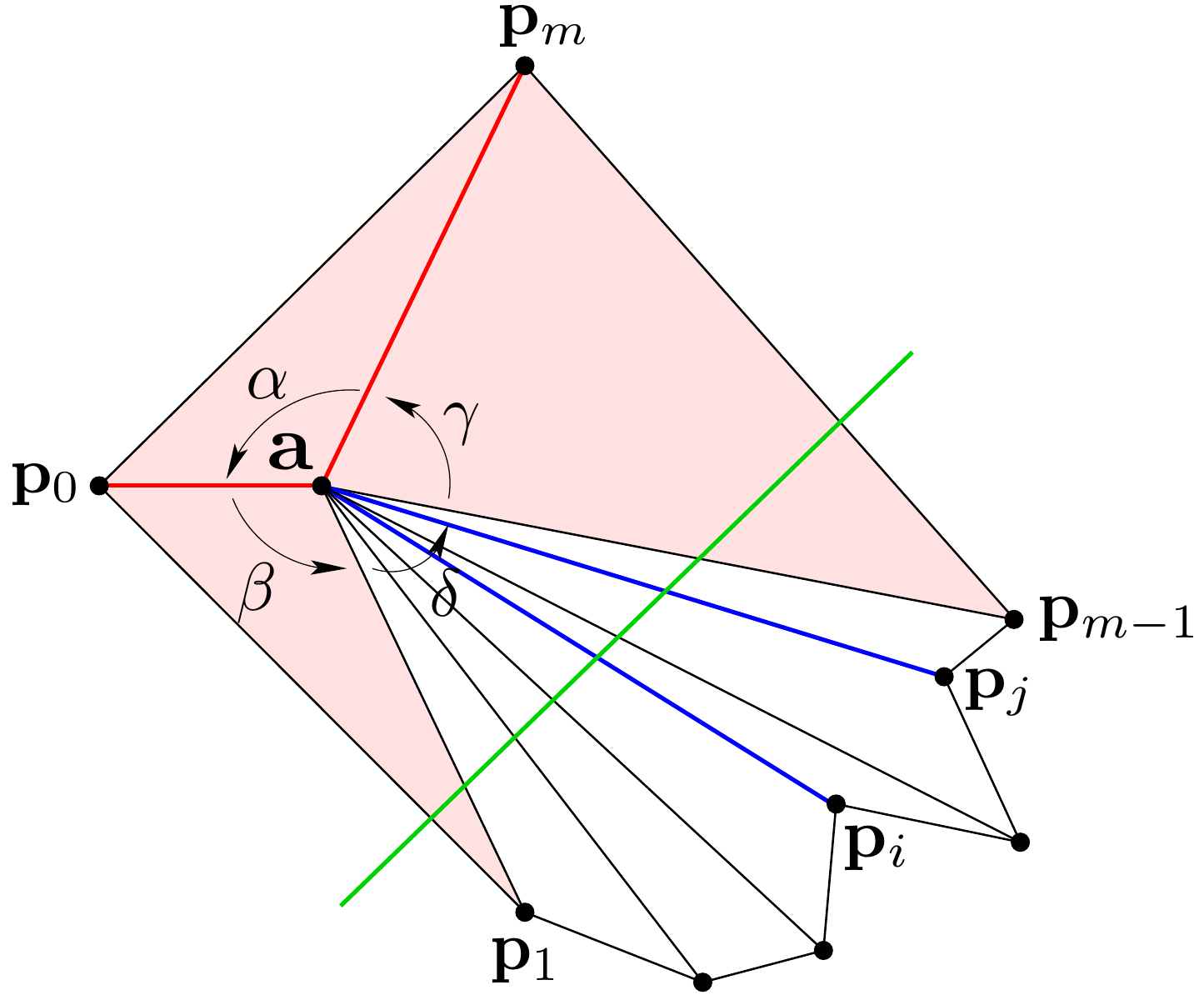}\\
  \end{tabular}
\caption{Two examples of the two special stars of a redundant interior vertex. There are two unflippable locally non-regular edges (red) at ${\bf a}$. The green lines are the projections of the intersection curves in $\mathbb{R}^3$ between a vertical plane and the set of lifted triangles at ${\bf a}$. }
\label{fig:vertex_remove_2d_case2} 
\end{figure}

Now we prove (ii). First of all, if there exists more than two locally non-regular edges at ${\bf a}$, then by the above facts, at least one edge whose unflippability is not due to ${\bf a}$. If there exists only two locally non-regular edges at ${\bf a}$, then we show that the unflippability of the other locally non-regular edge must not due to ${\bf a}$.  Consider the contrary, the unflippability of the other locally non-regular edge is also due to ${\bf a}$, without loss of the generality, let this edge be ${\bf a}{\bf p}_1$, see Figure~\ref{fig:vertex_remove_2d_case2} Left. Same as we prove (i), we must find another locally non-regular edge between edges ${\bf a}{\bf p}_1$ and ${\bf a}{\bf p}_m$. In other words, the cutting curve (shown in green in Figure~\ref{fig:vertex_remove_2d_case2}) between a vertical plane and all lifted triangles, $f_{{\bf a}'{\bf p}_1'{\bf p}_2'}, \ldots, f_{{\bf a}'{\bf p}_{m-1}'{\bf p}_m'}$ must contain at least one locally convex point. We arrive a contradiction. Therefore, in this case, the unflippability of the other edge must not due to ${\bf a}$. 

Finally, let $e_{{\bf a}{\bf p}_j}$ be a locally non-regular edge at ${\bf a}$ whose unflippability is due to ${\bf p}_j$. 
We show that ${\bf p}_j$ is a lower interior vertex of ${\bf A}$. The fact that $e_{{\bf a}{\bf p}_j}$ is locally non-regular with respect to up-flip shows that the lifted vertex ${\bf p}_j'$ lies above the plane passing through the lifted vertices ${\bf a}', {\bf p}_{j-1}', {\bf p}_{j+1}$, where $f_{{\bf a}{\bf p}_j{\bf p_{j-1}}}, f_{{\bf a}{\bf p}_j{\bf p_{j+1}}} \in {\cal T}$. This shows that ${\bf p}_j$ must be a lower interior vertex of ${\bf A}$.  
\end{proof}

The above lemma shows that in the triangulation ${\cal T}$ of a non-extreme node, around a redundant  interior vertex ${\bf a} \in {\cal T}$ there must exist another redundant interior vertex ${\bf p} \in {\cal T}$, such that the edge $e_{\bf ab} \in {\cal T}$ is locally non-regular with respect to directed flip and unflippable. The same happens to ${\bf p}$. This will naturally generate a sequence of unflippable locally non-regular edges with respect to directed flips. 
\[
{\cal E} := \{e_{{\bf a}_0{\bf p}_0}, e_{{\bf p}_0{\bf p}_1}, e_{{\bf p}_1{\bf p}_2}, \ldots, e_{{\bf p}_{m-1}{\bf p}_{m}}, e_{{\bf p}_{m}{\bf p}_{m+1}}\},
\]
such that every edge $e_{{\bf a}_i{\bf p}_i} \in {\cal T}$ is locally non-regular with respect to directed flip and is unflippanble. \\

{\bf Remark 1} Indeed, there might exist many sequences at a redundant interior vertex. This is due to the fact that there might exist more than one locally non-regular edges at this vertex whose unflippability are due to other vertices.\\  

Since ${\cal T}$ contains no directed flip at all, this means that at least one of such sequences form a cycle, i.e., $e_{{\bf p}_{m}{\bf p}_{m+1}} = e_{{\bf a}_{0}{\bf p}_{0}}$. 
For example, in Figure~\ref{fig:redundant-vertices} the three edges $e_{\bf ab}$, $e_{\bf bc}$, and $e_{\bf ca}$ form such a cycle. 

Another easy outcome of the above lemma is that a cycle needs at least $3$ edges. Hence there are at least $3$ redundant interior vertices of ${\bf A}$ in each triangulation.

We then see that any triangulation of a non-extreme external node of the poset must have the following property. 

\begin{theorem}~\label{thm:cycles_edges}
Any triangulation corresponds to a non-extreme node in this poset must contain at least $3$  redundant interior vertices of ${\bf A}$.  Moreover, at least $3$ of these vertices are connected by cycles of unflippable  locally non-regular edges with respect to the directed flips of this poset.
\end{theorem}

Interestingly, 
Joe~\cite{Joe1989} proved a special case in 3d triangulations, which showed that the failure of Lawson's algorithm in 3d is due to the existence of a cycle of connected unflippable locally non-Delaunay faces~\cite[Lemma 4, 5]{Joe1989}. This results implies that there must exist cycles of unflippable locally non-Delaunay edges in such 3d triangulations. Moreover, these edges are redundant interior edges.\\

{\bf Remark 2}. From the above theorem, we see that cycles of edges which are unflippable locally non-regular with respect to directed flips are the cause of failure of Lawson's flip algorithm. In other words, if a triangulation is non-regular but it does not contain such cycles, Lawson's flip algorithm will terminate. 

{\bf Remark 3}. We only showed that there must exist such cycles in non-extreme triangulations. There are many interesting questions regarding the number of cycles, and the relation between these cycles. For example, do they link to each other?

\section{Triangulating 3d Non-convex Polyhedra}
\label{sec:triang3d}

Based on the known structural properties of the directed flip graph, this section describes an  algorithm to triangulate a special class of 3d decomposable polyhedron without using Steiner points. Such polyhedra may be non-convex.  
This algorithm is simple since it only performs flips between triangulations of a planar point set. 

\subsection{The input polyhedra}

Let ${\cal P}$ be a family of 3d polyhedra such that any polyhedron $P \in {\cal P}$ satisfies the following conditions:
\begin{itemize}
\item[(i)] $P$ is homeomorphic to a $3$-ball. 
\item[(ii)] The vertices of $P$ are in convex position. 
\item[(iii)] The vertex set of $P$ can be represented by a planar point set ${\bf A}$ with a height function $\omega : {\bf A} \to \mathbb{R}$, such that $\forall {\bf p} \in {\bf A}$, the lifted point ${\bf p}'$ is a vertex of $P$. 
\item[(iv)] The boundary of $\partial P = {\cal T}_u^{\omega} \cup {\cal T}_v^{\omega}$, where ${\cal T}_u$ and ${\cal T}_v$ are two triangulations of ${\bf A}$.
\item[(v)] $\omega$ is either convex or concave. 
\end{itemize}

All convex polyhedra are in this family. 
However, only non-convex polyhedra which satisfy the (strong conditions) (iv) and (v) are in this family. 
The condition (iv) enforces that the orthogonal projection of the boundary of $P$ are two planar triangulations of the same point set which share at the same convex hull of this point set. 
Figure~\ref{fig:triang_poly} Left shows an example of such polyhedron $P$. 

The condition (v) implies that either ${\cal T}_u$ or ${\cal T}_v$ has the property that it does not contain interior vertices of ${\bf A}$, see Figure~\ref{fig:triang_poly} Left. With this property, the existence of a path from this triangulation to the two extreme triangulations ${\bf A}$ is guaranteed (Theorem~\ref{thm:poset-convex-heights}). They can be computed by the Lawson's flip algorithm in time $O(n^2)$. 

\begin{figure}[ht]
  \centering
  \begin{minipage}{0.43\textwidth}
  \includegraphics[width=1.0\textwidth]{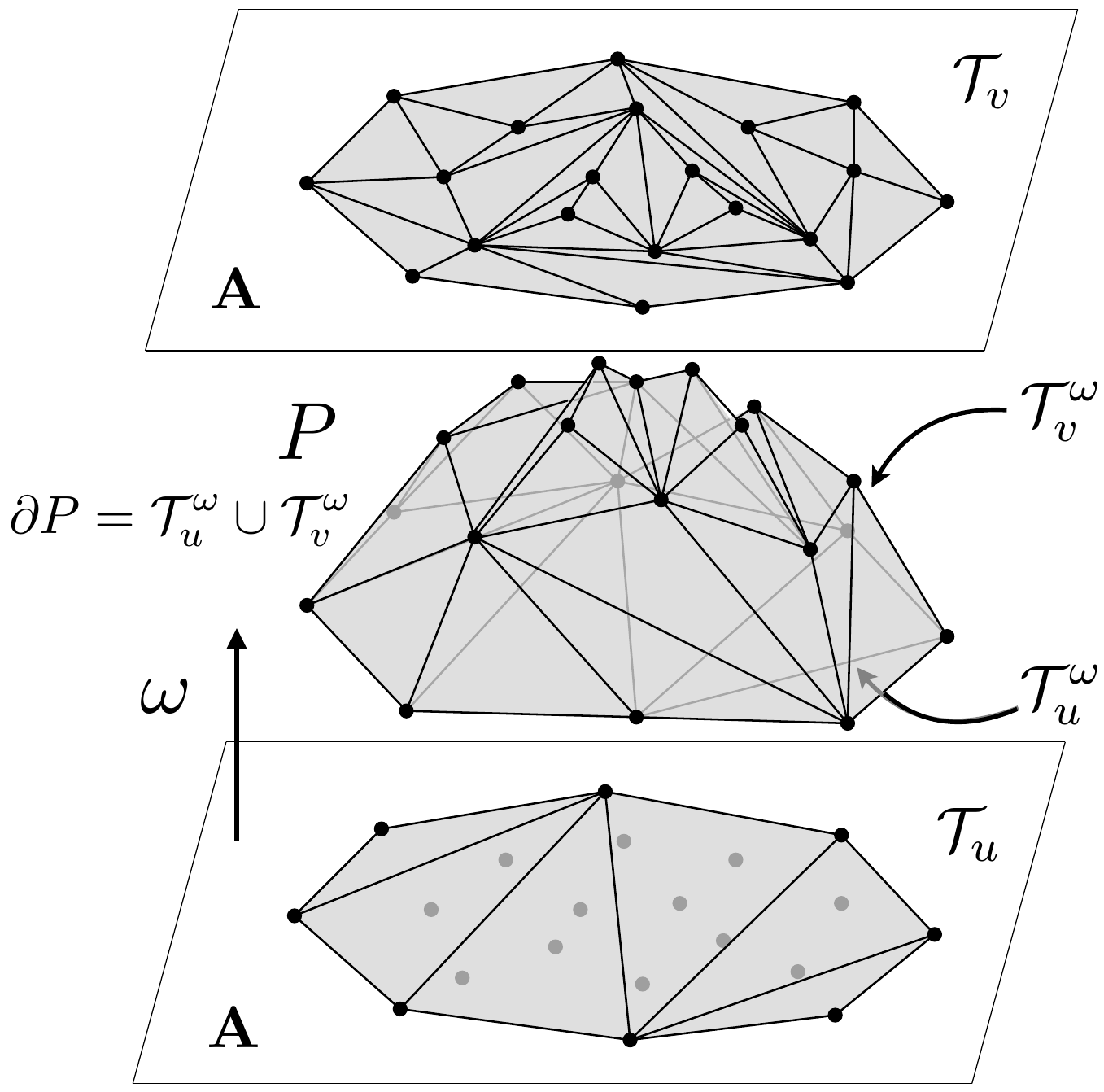}
  \end{minipage}
  \begin{minipage}{0.55\textwidth}
  \includegraphics[width=1.0\textwidth]{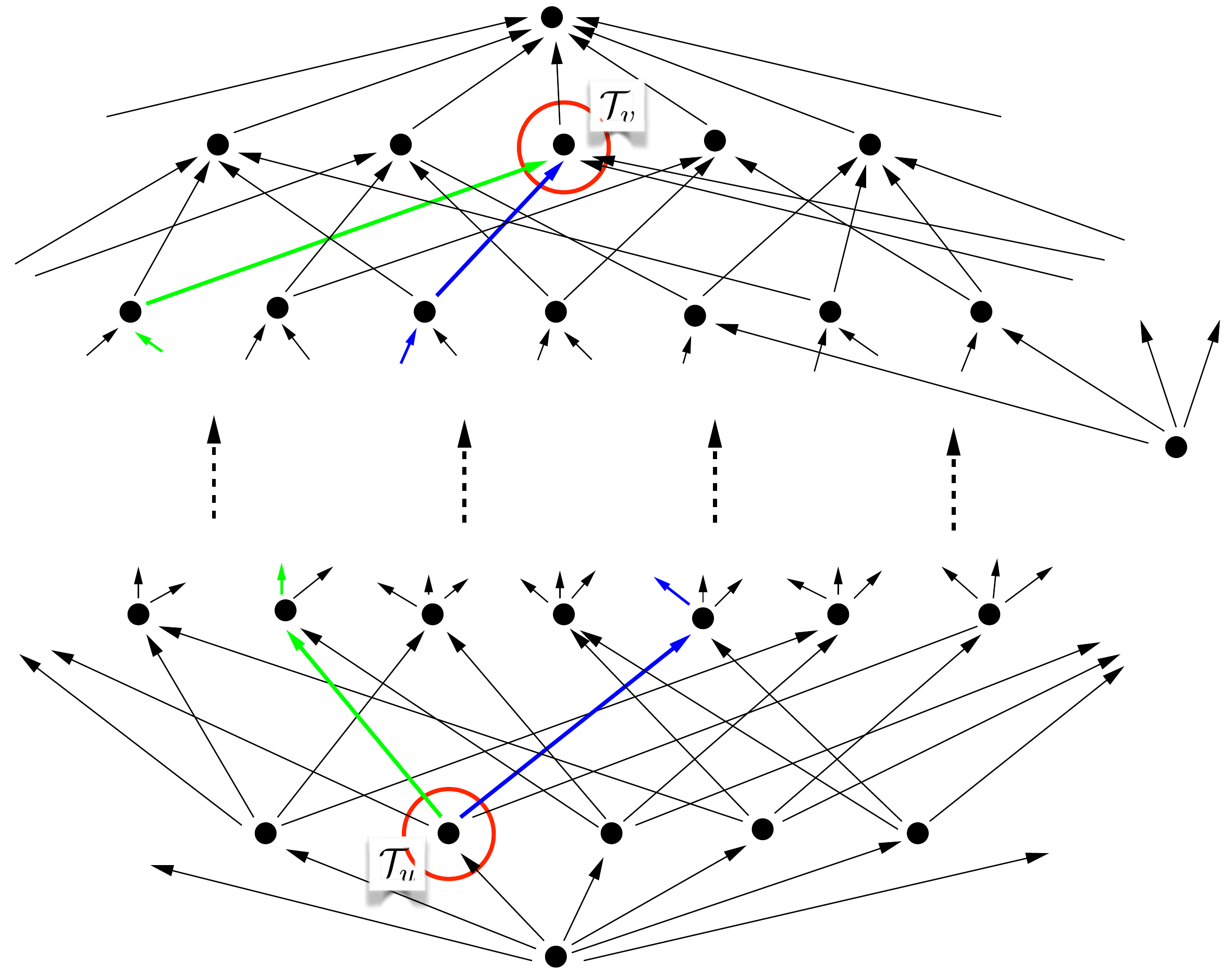}
  \end{minipage}
\caption{Left: An example of the input polyhedron.
Right: Examples of existing paths from the source to the target triangulations.
}
\label{fig:triang_poly} 
\end{figure}


Obviously, $P$ is convex when ${\cal T}_u$ and ${\cal T}_v$ are just the two extreme triangulations of $({\bf A}, \omega)$.  
While if they are not (or one of them is not), then $P$ must be non-convex. 

\subsection{The triangulation algorithm}

Let $P$ be a 3d polyhedron satisfies the above conditions. We describe an algorithm to triangulate $P$ without using Steiner points.  

The 3d triangulation problem is then transformed into finding a monotone sequence of directed flips  between the source and target triangulations.
It is then equivalent to find a path in the directed flip graph between these two nodes. 
Let us assume that the existence of a path between these two nodes in this graph is known. 
The question is how to find such a path?  

\paragraph{The ideas of the algorithm}
Let us assume the source triangulation is the one which contain no interior vertices of ${\bf A}$, and the target triangulation is an internal nodes of this graph.  
If we follow an arbitrary path starting from the source triangulation, it is guaranteed that we will reach the extreme triangulation (in its direction) of ${\bf A}$, but we may by-pass the internal node corresponds to the target triangulation.  

Here comes the main ideas of our algorithm:
Assume there exists a directed flip. Before doing this flip in a triangulation ${\cal T}$ of ${\bf A}$, we ensure that the corresponding tetrahedron of this flip does not intersect the target triangulation in its interior. 
Let us assume ${\cal T}_v$ is the target triangulation, let ${\bf a,b,c,d} \in {\cal T}$ are the support of this flip. We ensure the following condition: 
\[
\textrm{int}(t_{\bf abcd}) \cap |{\cal T}_v| = \emptyset,
\]
where $\textrm{int}(t_{\bf abcd})$ means the interior of the tetrahedron $t_{\bf abcd}$, and $|{\cal T}_v|$ is the underlying space of ${\cal T}_v$. 
We call such flip is {\it conforming to the target triangulation}.
If this flip does not satisfy the above condition, we skip this flip.  
We will show later, if we only do conforming flips, then we must end at the target triangulation.


\paragraph{The algorithm}
Without loss of generality, assume ${\cal T}_u$ is the source triangulation and ${\cal T}_v$ is the target. And assume that ${\cal T}_u$ lies vertically below ${\cal T}_v$. Hence the directed flips we need are up-flips. 
We initialise a working triangulation ${\cal T} := {\cal T}_u$, and initialise a working list $L := \emptyset$ for returning the list of tetrahedra of the triangulation.  


Let ${\cal E}$ be the set of all locally non-regular edges with respect to up-flips in ${\cal T}$. For each edge $e_{\bf ab} \in {\cal E}$. if it is already an edge of ${\cal T}_v$, we simply remove it from ${\cal E}$. Otherwise, if it is flippable, and this flip is conforming to ${\cal T}_v$, then flip $e_{\bf ab}$, and update ${\cal E}$ by all new locally non-regular edges with respect to up-flips in ${\cal T}$, add the tetrahedron $t_{\bf abcd}$ of this flip to $L$.  This phase ends until ${\cal E}$ is empty. 

If ${\cal E} = \emptyset$ and ${\cal T} \not= {\cal T}_v$, check if there exists a vertex ${\bf a}  \in {\cal T}_v$ and ${\bf a} \not\in {\cal T}$. If such a vertex exists, there must exist a triangle $f_{\bf bcd} \in {\cal T}$ which contains ${\bf a}, $ insert ${\bf a}$ into ${\cal T}$ by performing a 1-3 flip in ${\cal T}$. update ${\cal E}$ by all new locally non-regular edges with respect to up-flips in ${\cal T}$, add the tetrahedron $t_{\bf abcd}$ of this flip to $L$. 

\subsection{Termination and running time}

In this section, we prove the following theorem regarding the termination as well as running time of this algorithm. 

\begin{theorem}
If the target triangulation is a regular triangulation, 
then this algorithm terminates, and it runs in $O(n^3)$ time, where $n$ is the number of vertices of the input polyhedron $P$. 
\end{theorem}

\begin{proof}
If the target triangulation is the unique extreme node, then every path starting from the source triangulation will lead to it. There is even no need to check conformity of each flip. It is just the Lawson's algorithm with edge flips and vertex insertions.  In this case, the algorithm runs in $O(n^2)$ which is also the upper bound of the number of tetrahedra of a set of $n$ vertices. 
 
From now on we assume the target triangulation is an internal node of this poset. 
The condition imposed on the target triangulation means that it must not be a non-extreme node of the poset, which means that there exists a path from the source triangulation toward the target triangulation, this also means a 3d triangulation between the source and target triangulations exist.

Without loss of generality, assume the source triangulation ${\cal T}_u$ lies vertically below ${\cal T}_v$, and the directed flips we need are  up-flips.  By performing only conforming up-flips, we ensure that every intermediate triangulation ${\cal T}$ must lie vertically below ${\cal T}_v$, i.e., the lifted triangulation ${\cal T}^{\omega}$ never crosses the characteristic section of ${\cal T}_v^{\omega}$. This shows that we will never by-pass the target node in the poset. 

We still need to show that the target node must be reached. It is based on the following facts: there exists a 3d triangulation between the source and target triangulation. Our input ensures that the source triangulation contains no interior vertex. This implies we only need 2-2 flips (edge flips) and 1-3 flips (vertex insertions). This means that if the target triangulation is not reached, there always exists at least one up-flip in current triangulation. Therefore, the target triangulation must be reached. 

Now we show the upper bound of the running time of this algorithm. 
The number of total tetrahedra of this triangulation is $O(n^2)$.
At creation of each tetrahedron, we need to test the conformity against the target triangulation, which has $O(n)$ triangles.
In the worst case, we will need to test $t$ locally non-regular edges in order to find a valid conforming flip. 
If $t = O(n)$, then we get a naive upper bound $O(n^4)$. 
We can avoid this addition $O(n)$ tests at each creation of a tetrahedron by tagging all the tested locally non-regular edges which correspond to non-conforming flips. So we only need to test such flip once and will automatically avoid to test it again.   
Therefore the worst-case running time of this algorithm is $O(n^3)$. 
\end{proof}

{\bf Remark 1}. We remark that the the condition on source triangulations is sufficient but not necessary. 
This means that it is possible to start with a source triangulation which does contain interior vertices of ${\bf A}$. 
This will end at a node which is non-regular. As we showed in Theorem~\ref{thm:poset-general} (6) that non-regular triangulation might be also an interior node. This means, it still can find directed flips but this must be a 1-3 flip (vertex insertion), and the new vertex must be in the target triangulation. 

{\bf Remark 2}. The target triangulation may be non-regular triangulation as long as it does not correspond to a non-extreme node of this poset.  

\section{Discussions}
\label{sec:discussion}

In this paper, we studied monotone sequence of directed flips between triangulations of a finite point set $({\bf A}, \omega)$ in $\mathbb{R}^2 \times \mathbb{R}$, were $\omega$ is a height function which lifts each vertex of ${\bf A}$ into $\mathbb{R}^3$. It is shown that such sequences correspond to 3d triangulations of the lifted point set ${\bf A}^{\omega}$. We then studied the structural properties of the directed flip graph of the set of all triangulations of this point set. 
These properties clearly explained the general behaviour of the Lawson's flip algorithm. 
Based on this study, we proposed a simple algorithm to triangulation a special family of 3d non-convex polyhedra without using additional vertices. 

There are many interesting questions and problems remain to be studied. We have mentioned some of them within the paper. 

\begin{itemize}
\item[(1)] Question~\ref{question:acyclic_viewpoint} in Section~\ref{sec:triang-to-monotone} asks: given a 3d triangulation ${\cal T}$ of a point set ${\bf V}$ in $\mathbb{R}^3$. Assume ${\bf V}$ contains no interior vertices. Is it always possible to find a viewpoint in $\mathbb{R}^3$ such that ${\cal T}$ is acyclic with respect to this viewpoint? 

\item[(2)] In Section~\ref{sec:poset-structure} we have characterised all external (extreme and non-extreme) nodes of this poset. It remains to understand what are properties of internal nodes?  For examples, we would like to know where the set of internal nodes form a graded lattice? 

\item[(3)] Continue from (2), we would like to know more properties of triangulations corresponding to non-extreme nodes. Although it is proven there must exist cycles of unflippable locally non-regular edges, it remains unclear how different cycles relate to each other. For example, can they form a non-trivial link?

\item[(4)] There is still a piece of information from the motivation example in Section~\ref{sec:graph-poset-example} not be studied. In this example, we could easily ``see" that two triangulations are connected by a 3d flip (2-3 or 3-2 flip), see Figure~\ref{fig:prism_poset_2}. How to characterise this 3d flips in the poset?

\item[(5)] We have seen from many previous examples, e.g., Figure~\ref{fig:prism_poset_1}, that the existence of redundant interior vertices may cause the failure of the Lawson's algorithm.  
It is then necessary to study the following problem: {\it How to remove an interior vertex by a monotone sequence of directed flips?}  We comment that if we know how to do this, we could generalise our algorithm to triangulate a more general class of non-convex polyhedron whose vertices are on the convex hull.
\end{itemize}

The last question of this paper is how to generalise the result of this paper to higher dimensional point sets. The case of point sets in 3d is already a challenging question to consider.  


\begin{thebibliography}{10}

\bibitem{Bagemihl48-decomp-polyhedra}
F.~Bagemihl.
\newblock On indecomposable polyhedra.
\newblock {\em The American Mathematical Monthly}, 55(7):411--413, 1948.

\bibitem{Bern93-tetra}
M.~Bern.
\newblock Compatible tetrahedralizations.
\newblock In {\em Proc. 9th Annual ACM Symposium on Computational Geometry},
  pages 281--288, 1993.

\bibitem{Bezdek2016}
Andras Bezdek and Braxton Carrigan.
\newblock On nontriangulable polyhedra.
\newblock {\em Beitr{\"a}ge zur Algebra und Geometrie (Contributions to Algebra
  and Geometry)}, 57(1):51--66, 2016.

\bibitem{Bowyer81}
A.~Bowyer.
\newblock Computing {Dirichlet} tessellations.
\newblock {\em Comp. Journal}, 24(2):162--166, 1981.

\bibitem{Chazelle1984}
Bernard Chazelle.
\newblock Convex partition of polyhedra: a lower bound and worst-case optimal
  algorithm.
\newblock {\em SIAM Journal on Computing}, 13(3):488--507, 1984.

\bibitem{TriangBook}
Jes{\'u}s~A. {De Loera}, J{\"o}rg. Rambau, and Francisco Santos.
\newblock {\em Triangulations, Structures for Algorithms and Applications},
  volume~25 of {\em Algorithms and Computation in Mathematics}.
\newblock Springer Verlag Berlin Heidelburg, 2010.

\bibitem{Delaunay1934}
B.~N. Delaunay.
\newblock Sur la sph\`{e}re vide.
\newblock {\em Izvestia Akademii Nauk SSSR, Otdelenie Matematicheskikh i
  Estestvennykh Nauk}, 7:793--800, 1934.

\bibitem{Edelman1996}
Paul Edelman and Victor Reiner.
\newblock The higher {Stasheff-Tamari} posets.
\newblock {\em Mathematika}, 43:127--154, 1996.

\bibitem{Edelsbrunner90acy}
H.~Edelsbrunner.
\newblock An acyclicity theorem for cell complex in $d$ dimension.
\newblock {\em Combinatorica}, 10(3):251--260, 1990.

\bibitem{EdelsbrunnerShah96}
Herbert Edelsbrunner and N.~R. Shah.
\newblock Incremental topological flipping works for regular triangulations.
\newblock {\em Algorithmica}, 15:223--241, 1996.

\bibitem{EPPSTEIN1992143}
David Eppstein.
\newblock The farthest point delaunay triangulation minimizes angles.
\newblock {\em Computational Geometry}, 1(3):143 -- 148, 1992.

\bibitem{EPPSTEIN2009790}
David Eppstein, Marc van Kreveld, Elena Mumford, and Bettina Speckmann.
\newblock Edges and switches, tunnels and bridges.
\newblock {\em Computational Geometry}, 42(8):790 -- 802, 2009.
\newblock Special Issue on the 23rd European Workshop on Computational
  Geometry.

\bibitem{Gelfand1994Discriminants}
Israel~M. Gelfand, Mikhail~M. Kapranov, and Andrei~V. Zelevinsky.
\newblock {\em Discriminants, Resultants, and Multidimensional Determinants}.
\newblock Birkh\"{a}user Boston, Boston, MA, 1994.

\bibitem{GeorgeBorouchakiSaltel03}
{Paul-Louis} George, Houman Borouchaki, and Eric Saltel.
\newblock \'ultimate\' robustness in meshing an arbitrary polyhedron.
\newblock {\em International Journal for Numerical Methods in Engineering},
  58:1061--1089, 2003.

\bibitem{GeorgeHechtSaltel91}
{Paul-Louis} George, Fr{\'e}d{\'e}ric Hecht, and Eric Saltel.
\newblock Automatic mesh generator with specified boundary.
\newblock {\em Computer Methods in Applied Mechanics and Engineering},
  92:269--288, 1991.

\bibitem{Goodman1988}
J.~Goodman and J.~Pach.
\newblock Cell decomposition of polytopes by bending.
\newblock {\em Israel J. Mathematics}, 64:129--138, 1988.

\bibitem{JaumeRote2016}
Rafel Jaume and G\"unter Rote.
\newblock Recursively regular subdivisions and applications.
\newblock {\em Journal of Computational Geometry}, 7(1), 2016.

\bibitem{Jessen1967}
B.~Jessen.
\newblock Orthogonal icosahedra.
\newblock {\em Nordisk Mat. Tidskr}, 15:90--96, 1967.

\bibitem{Joe1989}
B.~Joe.
\newblock Three-dimensional triangulations from local transformations.
\newblock {\em SIAM Journal on Scientific and Statistical Computing},
  10(4):718--741, 1989.

\bibitem{Joe91-flip}
B.~Joe.
\newblock Construction of three-dimensional {Delaunay} triangulations using
  local transformations.
\newblock {\em Computer Aided Geometric Design}, 8:123--142, 1991.

\bibitem{Lawson1972}
C.~L. Lawson.
\newblock Transforming triangulations.
\newblock {\em Discrete Mathematics}, 3(4):365--372, 1972.

\bibitem{Lawson1977}
C.~L. Lawson.
\newblock Software for $c^1$ surface interpolation.
\newblock {\em Mathematical Software III, Academic Press}, pages 164--191,
  1977.

\bibitem{LAWSON1986231}
Charles~L. Lawson.
\newblock Properties of n-dimensional triangulations.
\newblock {\em Computer Aided Geometric Design}, 3(4):231 -- 246, 1986.

\bibitem{Radon1921}
J.~Radon.
\newblock Mengen konvexer {K}{\"o}rper, die einen gemeinschaftlichen {Punkt}
  enthalten.
\newblock {\em Math. Ann.}, 83:113--115, 1921.

\bibitem{Rambau96-thesis}
J.~Rambau.
\newblock {\em Polyhedral Subdivisions and Projections of Polytopes}.
\newblock PhD thesis, Fachbereich 3 Mathematik der Technischen Universit{\"a}t
  berlin, Berlin, Germany, Oktober 1996.

\bibitem{Rambau05}
J{\"o}rg Rambau.
\newblock On a generalization of {S}ch{\"o}nhardt's polyhedron.
\newblock In J.~E. Goodman, J.~Pach, and E.~Welzl, editors, {\em Combinatorial
  and Computational Geometry}, volume~52, pages 501--516. MSRI publications,
  2005.

\bibitem{RuppertSeidel92}
Jim Ruppert and Raimund Seidel.
\newblock On the difficulty of triangulating three-dimensional nonconvex
  polyhedra.
\newblock {\em Discrete \& Computational Geometry}, 7:227--253, 1992.

\bibitem{Santos2000a}
F.~Santos.
\newblock A point set whose space of triangulations is disconnected.
\newblock {\em Amer. Math. Soc.}, 13:611--637, 2000.

\bibitem{Santos2000}
F.~Santos.
\newblock Triangulations with very few geometric bistellar neighbors.
\newblock {\em Discrete \& Computational Geometry}, 23(1):15--33, 2000.

\bibitem{Schonhardt1928}
E.~Sch{\"o}nhardt.
\newblock {\"U}ber die zerlegung von dreieckspolyedern in tetraeder.
\newblock {\em Mathematische Annalen}, 98:309--312, 1928.

\bibitem{Si2015-TetGen}
Hang Si.
\newblock {TetGen}, a {Delaunay}-based quality tetrahedral mesh generator.
\newblock {\em ACM Trans. Math. Softw.}, 41(2):11:1--11:36, February 2015.

\bibitem{SI201892}
Hang Si and Nadja Goerigk.
\newblock Generalised {Bagemihl} polyhedra and a tight bound on the number of
  interior {Steiner} points.
\newblock {\em Computer-Aided Design}, 103:92 -- 102, 2018.

\bibitem{Sleator1988}
Daniel~D. Sleator, William~P. Thurston, and Robert~Endre Tarjan.
\newblock Rotation distance,triangulations,and hyperbolic geometry.
\newblock {\em J. Amer. Math. Soc.}, 1:647--682, 1988.

\bibitem{Toussaint93}
G.~T. Toussaint, C.~Verbrugge, C.~Wang, and B.~Zhu.
\newblock Tetrahedralization of simple and non-simple polyhedra.
\newblock In {\em Proc. 5th Canadian Conference on Computational Geometry},
  pages 24--29, 1993.

\bibitem{Watson81}
D.~F. Watson.
\newblock Computing the $n$-dimensional {Delaunay} tessellations with
  application to {V}oronoi polytopes.
\newblock {\em Comput. Journal}, 24(2):167--172, 1981.

\bibitem{WeatherillHassan94}
Nigel~P. Weatherill and Oubay Hassan.
\newblock Efficient three-dimensional {Delaunay} triangulation with automatic
  point creation and imposed boundary constraints.
\newblock {\em International Journal for Numerical Methods in Engineering},
  37:2005--2039, 1994.

\bibitem{Ziegler1995-book}
G\"unter~M. Ziegler.
\newblock {\em Lectures on Polytopes}, volume 152 of {\em Graduate Texts in
  Mathematics}.
\newblock Springer-Verlag, New York, second edition edition, 1997.

\end{thebibliography}

\end{document}